\def\anonymous{0}
\def\extended{1}

\def\acm{0}
\if\extended 0
\if\anonymous 0
\def\acm{1}
\fi
\fi

\if\acm 0
\documentclass[sigplan,screen]{acmart}
\renewcommand\footnotetextcopyrightpermission[1]{}
\else
\documentclass[sigplan,screen]{acmart}
\fi

\def\includecomments{0}
\def\diffhighlight{0} 
\def\tikzmode{1}

\usepackage{amsthm,amsfonts,stmaryrd,mathtools}
\usepackage{listings}
\usepackage{xspace}
\usepackage[shortlabels]{enumitem}
\usepackage{makecell} 
\usepackage{soul} 

\usepackage[noabbrev]{cleveref}

\makeatletter
\if\tikzmode 1
\def\kc@figpicture{\let\endtikzpicture\relax\Collect@Body\kc@figpdf}
\long\def\kc@figpdf#1{%
  \expandafter\filename@parse\expandafter{\CurrentFile}%
  \includegraphics{figures-pdf/\filename@base}}
\AtBeginDocument{%
  \let\kc@tikzpicture\tikzpicture
  \def\tikzpicture{%
    \ifnum\pdf@strcmp{\CurrentFilePath}{figures}=0
      \expandafter\kc@figpicture
    \else
      \expandafter\kc@tikzpicture
    \fi}}
\else
\usepackage{tikz}
\usetikzlibrary{calc}
\usetikzlibrary{backgrounds}
\usetikzlibrary{arrows.meta}
\usetikzlibrary{decorations.pathreplacing} 
\if\tikzmode 2
\usetikzlibrary{external}
\AddToHook{env/tikzpicture/before}{%
  \ifnum\pdf@strcmp{\CurrentFilePath}{figures}=0
    \expandafter\filename@parse\expandafter{\CurrentFile}%
    \tikzsetnextfilename{\filename@base}%
  \else
    \tikzset{external/export next=false}
  \fi}
\AtBeginDocument{\tikzifexternalizing{\gdef\@abspage@last{1}}{}}
\fi
\fi
\makeatother

\newfloat{problem}{htbp}{lop}
\floatstyle{plaintop}  
\restylefloat{problem} 
\floatname{problem}{Problem}
\crefname{problem}{Problem}{Problems}
\Crefname{problem}{Problem}{Problems}

\makeatletter
\lst@AddToHook{Init}{\def\@currentcounter{lstnumber}}%
\makeatother

\if\includecomments 1
\usepackage{todonotes} 
\setuptodonotes{size=\tiny, fancyline, color=blue!3}
\else
\newcommand{\todo}[1]{}
\fi

\microtypecontext{spacing=nonfrench}

\newtheorem{theorem}{Theorem}[section]
\newtheorem{lemma}[theorem]{Lemma}
\newtheorem{claim}[theorem]{Claim}
\newtheorem{corollary}[theorem]{Corollary}

\newtheorem{definition}{Definition}
\newtheorem{observation}[theorem]{Observation}

\AddToHook{env/lemma/begin}{\crefalias{theorem}{lemma}}
\AddToHook{env/claim/begin}{\crefalias{theorem}{claim}}
\AddToHook{env/corollary/begin}{\crefalias{theorem}{corollary}}
\AddToHook{env/invariant/begin}{\crefalias{theorem}{invariant}}
\AddToHook{env/definition/begin}{\crefalias{theorem}{definition}}
\AddToHook{env/observation/begin}{\crefalias{theorem}{observation}}
\AddToHook{env/assumption/begin}{\crefalias{theorem}{assumption}}
\crefname{theorem}{theorem}{theorems}
\crefname{lemma}{lemma}{lemmas}
\crefname{claim}{claim}{claims}
\crefname{corollary}{corollary}{corollaries}
\crefname{invariant}{invariant}{invariants}
\crefname{definition}{definition}{definitions}
\crefname{observation}{observation}{observations}
\crefname{assumption}{assumption}{assumptions}

\crefname{lstlisting}{algorithm}{algorithms}
\Crefname{lstlisting}{Algorithm}{Algorithms}

\makeatletter
\lst@Key{countblanklines}{true}[t]%
{\lstKV@SetIf{#1}\lst@ifcountblanklines}

\lst@AddToHook{OnEmptyLine}{%
	\lst@ifnumberblanklines\else%
	\lst@ifcountblanklines\else%
	\advance\c@lstnumber-\@ne\relax%
	\fi%
	\fi}
\makeatother

\def\ContinueLineNumber{\lstset{firstnumber=last}}
\def\StartLineAt#1{\lstset{firstnumber=#1}}

\if\includecomments 1
\newcommand{\fix}[1]{\textcolor{blue}{#1}}
\newcommand{\fixg}[1]{\textcolor{red}{#1}}
\newcommand{\fixb}[1]{\textcolor{teal}{#1}}

\definecolor{mygreen}{RGB}{0, 180, 0}

\newcommand{\strike}[1]{\textcolor{blue}{\st{#1}}}

\newcommand{\rg}[1]{\textcolor{brown}{Rachid: #1}}
\newcommand{\antoine}[1]{\textcolor{orange}{Antoine: #1}}
\newcommand{\zyf}[1]{\textcolor{cyan}{Clément: #1}}

\newcommand{\marcos}[1]{\textcolor{orange}{Marcos: #1}}
\newcommand{\galy}[1]{\textcolor{red}{Galy: #1}}
\else
\newcommand{\fix}[1]{#1}
\newcommand{\fixg}[1]{#1}
\newcommand{\fixb}[1]{#1}

\newcommand{\strike}[1]{}

\newcommand{\rg}[1]{}
\newcommand{\antoine}[1]{}
\newcommand{\zyf}[1]{}
\newcommand{\marcos}[1]{}
\newcommand{\galy}[1]{}
\fi

\if\diffhighlight 1
\newcommand{\camera}[1]{\textcolor{blue}{#1}} 
\newcommand{\cameraz}[1]{\textcolor{blue}{#1}} 
\newcommand{\cameraa}[1]{\textcolor{red}{#1}} 
\newcommand{\camerab}[1]{\textcolor{red}{#1}} 
\newcommand{\camerac}[1]{\textcolor{red}{#1}} 
\newcommand{\camerad}[1]{\textcolor{red}{#1}} 
\else
\if\diffhighlight 2
\newcommand{\camera}[1]{\textcolor{purple}{#1}} 
\newcommand{\cameraz}[1]{\textcolor{blue}{#1}} 
\newcommand{\cameraa}[1]{\textcolor{magenta}{#1}} 
\newcommand{\camerab}[1]{\textcolor{teal}{#1}} 
\newcommand{\camerac}[1]{\textcolor{brown}{#1}} 
\newcommand{\camerad}[1]{\textcolor{red}{#1}} 
\else
\newcommand{\camera}[1]{#1}
\newcommand{\cameraz}[1]{#1}
\newcommand{\cameraa}[1]{#1}
\newcommand{\camerab}[1]{#1}
\newcommand{\camerac}[1]{#1}
\newcommand{\camerad}[1]{#1}

\fi
\fi

\newcommand{\sysnameplain}{KCensus}
\newcommand{\sysname}{\sysnameplain\xspace}
\newcommand{\smr}{KSMR\xspace}

\if\extended 1
\newcommand{\slogan}{Synthesizing Latency-Optimal Consensus Fast Paths \newline [Extended Version]}
\else
\newcommand{\slogan}{Synthesizing Latency-Optimal Consensus Fast Paths}
\fi

\newcommand{\us}{µs\xspace}
\newcommand{\rmv}[1]{}
\newcommand{\revisit}[1]{}
\renewcommand{\t}[1]{\texttt{\small #1}}

\newcommand{\rt}{roundtrip\xspace}
\newcommand{\rts}{roundtrips\xspace}

\newenvironment{myitemize}{\begin{list}{\labelitemi}{%
\setlength{\topsep}{0.5pt plus 0pt minus 0pt}%
\setlength{\itemsep}{0pt plus 0pt minus 0pt}%
\setlength{\parsep}{0pt plus 0pt minus 0pt}%
\setlength{\parskip}{0pt plus 0pt minus 0pt}%
\setlength{\leftmargin}{1em} 
}}{\end{list}}

\if\acm 0
\fi

\AtBeginDocument{%
  \providecommand\BibTeX{{%
    \normalfont B\kern-0.5em{\scshape i\kern-0.25em b}\kern-0.8em\TeX}}}

\if\acm 1
\copyrightyear{2027}
\acmYear{2027}
\setcopyright{cc}
\setcctype{by}
\acmConference[EuroSys '27]{22nd European Conference on Computer Systems}{April 19--23, 2027}{Rabat, Morocco}
\acmBooktitle{22nd European Conference on Computer Systems (EuroSys '27), April 19--23, 2027, Rabat, Morocco}
\acmDOI{10.1145/3842654.3848528}
\acmISBN{979-8-4007-2971-3/2027/04}
\else
\acmConference[Preprint]{}{}{}
\fi

\acmBadgeR[https://www.acm.org/publications/policies/artifact-review-and-badging-current]{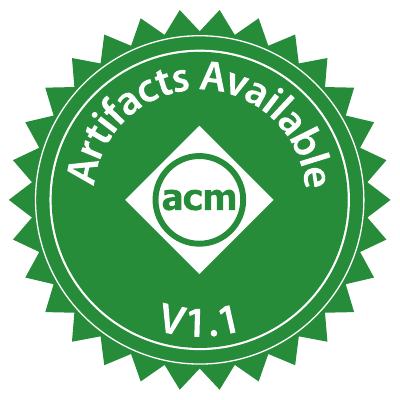}
\acmBadgeR[https://www.acm.org/publications/policies/artifact-review-and-badging-current]{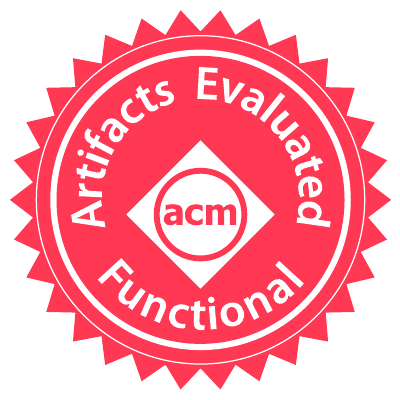}
\acmBadgeR[https://www.acm.org/publications/policies/artifact-review-and-badging-current]{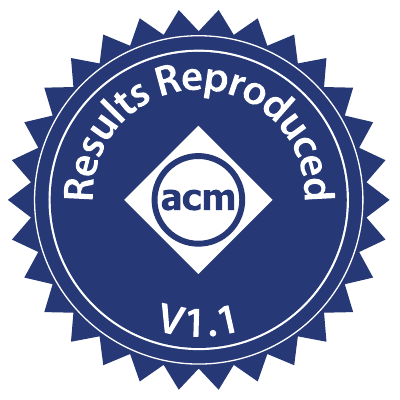}

\begin{document}

\title
[\sysnameplain: \slogan]
{\sysname: \slogan}

\if\extended 1
\if\anonymous 0
\titlenote{\camerad{Extended version of the paper to appear in the 22nd European Conference on Computer Systems (EuroSys~'27). Proceedings version: \url{https://doi.org/10.1145/3842654.3848528}. \copyright~2027 Copyright held by the owner/author(s). Licensed under \href{https://creativecommons.org/licenses/by/4.0/}{CC BY 4.0}.}}
\fi
\fi

\makeatletter
\renewcommand{\@titlefont}{\huge\sffamily\bfseries}
\if\extended 1
\newlength{\extskip}\setlength{\extskip}{16.5pt}
\patchcmd{\@mkauthors@iii}{\unvbox\mktitle@bx\par\medskip}{\unvbox\mktitle@bx\par\medskip\vspace{\extskip}}{}{\errmessage{extskip title patch failed}}
\patchcmd{\@mkauthors@iii}{\par\bigskip}{\par\bigskip\vspace{2\extskip}}{}{\errmessage{extskip author patch failed}}
\fi
\makeatother

\if\anonymous 1
\author{Submission \#29}
\else

\author{Clément Burgelin}
\affiliation[obeypunctuation=true]{%
  \institution{EPFL}, \city{Lausanne}, \country{Switzerland}%
}

\author{Antoine Murat}
\affiliation[obeypunctuation=true]{%
  \institution{EPFL}, \city{Lausanne}, \country{Switzerland}%
}

\author{Gal Sela}
\affiliation[obeypunctuation=true]{%
  \institution{EPFL}, \city{Lausanne}, \country{Switzerland}%
}

\author{Marcos K. Aguilera}
\affiliation[obeypunctuation=true]{%
  \institution{NVIDIA}, \city{Santa Clara},\\ \country{USA}%
}

\author{Rachid Guerraoui}
\affiliation[obeypunctuation=true]{%
  \institution{EPFL}, \city{Lausanne}, \country{Switzerland}%
}

\makeatletter
\renewcommand{\@authorfont}{\normalsize\sffamily}
\renewcommand{\@affiliationfont}{\footnotesize\normalfont}
\makeatother
\settopmatter{authorsperrow=5}

\fi

\begin{abstract}
\cameraz{Strongly consistent geo-replication often relies on fast paths to reduce latency in the common case of no failures or contention.}
Existing fast-path schemes, \cameraz{however}, are 
\camerad{ad hoc}
and restrictive: each corresponds to a point in a broad
design space shaped by network topology, workload, and latency objective, so no single scheme works best across settings. This paper looks at
fast-path schemes from a new perspective, as mechanisms
that spread
\cameraz{knowledge} about \cameraz{proposals}.
With this view, we identify a fundamental condition on the spread
of knowledge for a fast-path scheme to work. We then introduce \sysname, a framework that turns this condition into an optimization problem, synthesizing new fast-path schemes that are optimal for a given setting.
We use \sysname to build a geo-replicated key-value store and evaluate it across AWS regions.
Our system outperforms competing protocols, with up to \cameraz{16\%} lower average latency.
\end{abstract}

\if\acm 1


\fi

\maketitle

\if\acm 0
\pagestyle{plain}
\fi

\if\extended 1
\vspace*{3\extskip}
\fi

\section{Introduction}
\label{sec:intro}

\cameraz{Modern Internet services rely on geo-replication to remain available despite regional outages. Using consensus for this replication provides the strong consistency many applications require~\cite{linearizability, smr}. Consensus, however, adds latency to operations because processes must communicate to safely agree on their ordering~\cite{dls}. This cost is particularly high in geo-replicated systems, where communication incurs large wide-area delays. At the same time, every millisecond matters for Internet services~\cite{akamai, amazon}. Consequently, much work has focused on reducing consensus latency, from Paxos variants to more recent leaderless and topology-aware protocols~\cite{paxos, fast-paxos, epaxos, wpaxos, mencius, pando, swiftpaxos, wheat, aware, mercury}.}

\begin{figure}
    \centering
    \includegraphics[width=\columnwidth]{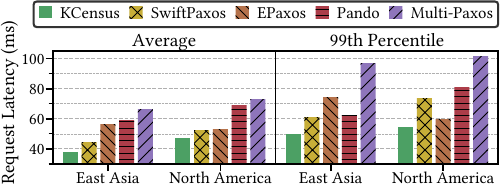}
    \caption{Average and tail latency in two \cameraz{deployments}.
    The best \cameraz{prior} protocol
    changes with the latency metric and topology. \sysname always performs best for the targeted metric.}
    \label{fig:intro}
    \Description{Bar chart of average and 99th-percentile request latency in the East Asia and North America deployments for KCensus, SwiftPaxos, EPaxos, Pando, and Multi-Paxos. The best prior protocol differs between metrics and deployments, while KCensus has the lowest latency in every case.}
\end{figure}

\cameraz{These protocols typically use a \emph{fast path} to reduce latency in common executions without failures or contention~\cite{fast-paxos, epaxos, pando, swiftpaxos}.
Existing fast paths, however, have two drawbacks.
First, they are complex and error-prone, as subtle concurrent scenarios have led to buggy correctness arguments~\cite{epaxoscorrectness,epaxosfixed}.
Second, they are restrictive: each hard-codes a communication pattern, quorum criteria, and recovery mechanism.
Yet no fixed choice is optimal: the best one depends on the network topology, the proposer, and the target latency metric.}

\cameraz{Figure~\ref{fig:intro} illustrates this effect: among existing protocols, the one with the lowest average latency does not always have the lowest tail latency, and the tail-latency winner changes across deployments.
More fundamentally, using one fast-path strategy for all proposers is itself too restrictive.
\cameraz{A proposer's fast-path \emph{strategy} specifies what it
must know about other processes' states before deciding.}
Best performance may require different strategies for different proposers, tailored to their location and proposal rate.
These strategies must nevertheless coexist safely: a decision made with one must preclude a conflicting decision made with another.}

\cameraz{The key challenge is recovery.
If a fast path fails, the fallback must preserve any decision that may already have occurred.
Thus, different possible fast decisions must remain distinguishable after failures.
We formalize this \camerad{as} a general \camerad{\emph{recoverability}} property of consensus: \emph{two reachable system states in which different values were decided must differ in the local states of more than $f$ processes}, where $f$ is the number of tolerated crashes.
Otherwise, after those processes crash, the survivors cannot distinguish the two states.}

\cameraz{To translate this property into fast-path design, we view fast paths as mechanisms that spread evidence about acceptances (\camerad{i.e., votes}): which processes accepted a value and who observed those acceptances.
Specializing the recoverability property yields two fundamental conditions: \emph{every fast-path quorum must contain more than $f$ processes}, and \emph{conflicting quorums must intersect and evidence about the intersection must survive $f$ failures}.
We further prove these conditions necessary and sufficient, thereby tightly characterizing how much work a fast path must do.}

\cameraz{We use these conditions to build}
\emph{\sysname} (\textbf{K}nowledge \textbf{Census}), a framework for synthesizing consensus fast paths.
\cameraz{\sysname encodes each proposer's fast-path strategy as a knowledge requirement, reducing correctness to compatibility checks between them, and selects a compatible set of requirements that minimizes the target latency (e.g., mean or tail) for the network topology and proposer distribution.}

To demonstrate \sysname's practicality, we use it to build a strongly consistent geo-replicated key-value store whose 
fast path is optimized based on network topology, proposer
  distribution, and target latency metric.
In our
\cameraz{contention-free} experiments, our system always achieves the best
targeted metric---cutting average latency by up to \cameraz{16\%}.
This shows that synthesized fast paths are not only cleaner and more flexible than hard-coded ones, but also significantly faster.

In summary, this paper makes the following contributions:
\begin{myitemize}
    \item We \cameraz{formulate a recoverability property of consensus, derive from it} \cameraz{the} fundamental knowledge \cameraz{conditions} that fast paths must satisfy to remain safely recoverable, and prove that \cameraz{they} are both necessary and sufficient.

    \item We present \sysname, a knowledge-based framework for expressing and verifying per-proposer
    fast-path strategies.

    \item We show how to synthesize an optimal fast path from network topology, proposer distribution, and target latency metric, thereby uncovering new fast-path strategies.

    \item We implement \sysname in a geo-replicated key-value store, open-sourced at \href{https://github.com/LPD-EPFL/kcensus}{github.com/LPD-EPFL/kcensus}, and evaluate it \camerad{across geo-distributed AWS deployments}, showing substantial latency improvements over prior protocols.
\end{myitemize}

\section{Background}
\label{sec:background}

We review consensus-based strong geo-replication (\S\ref{sec:bg:repl}),
explain how prior protocols improve latency through fast paths using
adopt-commit (\S\ref{sec:bg:fastpaths}), and state our system model~(\S\ref{sec:bg:model}).

\subsection{Replication and consensus}
\label{sec:bg:repl}

State Machine Replication (SMR)~\cite{smr} is a standard technique to build reliable services with strong consistency, i.e., linearizability~\cite{linearizability}. Replicas start from the same initial state and apply the same stream of deterministic commands, thereby giving clients the illusion of a single fault-tolerant server. 
\cameraz{Replicas agree on this stream by repeatedly running a consensus
protocol, typically once per stream slot~\cite{dls}.}
\cameraz{In each consensus instance, processes propose \emph{values}, no two decide differently, and correct processes eventually decide.}

\subsection{Fast paths and adopt-commit}
\label{sec:bg:fastpaths}
\label{sec:bg:ac}

\cameraz{Consensus protocols reduce latency by optimizing a common-case \emph{fast path}, typically for executions without failures or contention.
Multi-Paxos-style stable-leader protocols keep this path simple: a leader \cameraz{decides values} by contacting a majority~\cite{paxos, multipaxos}.
This is fast for clients near the leader, but its latency depends on leader placement.
Other protocols make different fast-path choices.
Fast Paxos and EPaxos allow multiple proposers to use the fast path concurrently~\cite{fast-paxos, epaxos}, while topology-aware protocols tailor its communication pattern to wide-area deployments~\cite{wpaxos, swiftpaxos, pando}.
For our purposes, these fast paths differ mainly in \emph{which} processes participate and \emph{how} knowledge flows among them, which dictates their performance for a given setting.}

\cameraz{As previously observed~\cite{leaderless-smr, fast-mutex}, a fast path can be modeled as an \emph{adopt-commit} protocol~\cite{DBLP:conf/podc/YangNG98}.
Processes propose values as in consensus, but a proposer returns either \emph{commit}($v$) or \emph{adopt}($v$).
Committing $v$ means that $v$ can be decided immediately: the fast path succeeded, and no different value may be adopted or committed.
Adopting $v$ means that the proposer could not commit on the fast path, but $v$ may have committed elsewhere.
The proposer therefore carries $v$ into a fallback consensus protocol, such as Paxos, ensuring that recovery preserves any fast-path decision.}

\cameraz{Formally, adopt-commit satisfies three properties:
(1) \textbf{Termination}: every correct proposer eventually adopts or commits;
(2) \textbf{Validity}: only proposed values are adopted or committed; and
(3) \textbf{Agreement}: if $v$ commits, then only $v$ may be adopted or committed.
Note that, in our formulation, not all correct processes need to propose.}

\cameraz{Building on this abstraction, \sysname synthesizes optimized adopt-commit protocols with per-proposer strategies.}

\subsection{System model}
\label{sec:bg:model}

\cameraz{We consider a standard point-to-point message-passing system with $n \ge 2f{+}1$ processes, of which up to $f$ may crash.}
\cameraz{Links are FIFO and lossless between correct processes.}

\cameraz{The adopt-commit protocols synthesized by \sysname, and the SMR protocol built atop them, are safe under full asynchrony.
Our adopt-commit protocols need eventual crash suspicion to terminate, while good performance requires no false suspicions during stable periods~\cite{ct96}.
As usual, we rely on partial synchrony for deterministic SMR termination~\cite{dls}.}

\cameraz{For synthesis, we also assume stable periods long enough to estimate inter-process latencies and proposal rates, and install new strategies~(\S\ref{sec:smr}). This assumption is needed only for optimization, not safety.}

\section{Recoverable Fast Paths}
\label{sec:constraints}

\cameraz{This section characterizes what information a fast path must leave behind for safe recovery~(\S\ref{sec:constraints:condition}), and shows how the resulting conditions explain familiar protocols~(\S\ref{sec:constraints:examples}).}

\subsection{\cameraz{Recoverability}}
\label{sec:constraints:condition}

\cameraz{Intuitively, safe recovery requires different possible commits to
remain distinguishable after failures. We formalize this intuition with
a new \emph{recoverability} property of adopt-commit:}
\begin{quote}
\camerac{\emph{Any two \camerad{reachable states} in which different values were
committed must differ in the local states of \camerad{more than $f$} processes.}}
\end{quote}

\cameraz{Suppose otherwise that states $S_1$ and $S_2$\camerad{, each reached by some execution~\cite{flp},} contain commits
of $v_1$ and $v_2{\neq}v_1$, but differ at no more than $f$ processes.
If these processes crash and their pending messages are lost, the
survivors have identical local states in both executions. They therefore
behave identically under the same continuation, yet \camerad{if a survivor proposes,} Termination requires
\camerac{that \camerad{it} eventually adopt or commit,}
while Agreement allows only $v_1$
after $S_1$ and only $v_2$ after $S_2$, a contradiction.}
\cameraa{By the same reasoning, consensus satisfies an analogous
recoverability property, with  \emph{decided} in place of \emph{committed}.}

\paragraph{Knowledge requirements}\cameraz{We specialize this property to \emph{acceptance-based} fast paths,
\camerab{which are algorithms \camerad{in which}}
each process accepts the first value it sees in an adopt-commit instance,
never accepts another, and all tracked knowledge concerns these acceptances.
We say that a process
\emph{witnesses} an acceptance when its local state records it.}

\cameraz{We represent proposer $p$'s strategy by a \emph{knowledge requirement}
$R_p$, 
\camerab{which is a mapping where $R_p.\t{keys}()$ are the
\emph{required witnesses} of $p$, and for each $q\in R_p.\t{keys}()$,
$R_p[q]$ is the set of processes that $q$ must have recorded as acceptors of $p$'s value.
We say that $R_p$ is}
\emph{met} once $p$ knows that each such $q$ holds these records.}

\camerac{For example, if
$R_p=\{p{:}\{p,q,r\},q{:}\{p,q\},r{:}\{p,r\}\}$,
then $p$ may commit its value $v$ only after it knows that
(1) $p$, $q$, and $r$ accepted $v$;
(2) $q$ witnessed the acceptances of $p$ and $q$; and
(3) $r$ witnessed the acceptances of $p$ and $r$.}

\cameraz{\camerad{Because each process first witnesses its own acceptance,}
\cameraa{we assume} \camerab{without loss of generality} \cameraa{the following normal form:}
\camerab{for every $q \in R_p.\t{keys}()$,}
\camerac{$R_p[q]\subseteq R_p.\t{keys}()$ and $q\in R_p[q]$.}
\camerad{Indeed, by making every recorded acceptor a required witness and every
required witness a witness of its own acceptance, any requirement can be put
into this form without increasing fast-path latency.}
\cameraa{By the normal form, the} required witnesses
$R_p.\t{keys}()$ \cameraa{are exactly the} processes whose acceptances must be
recorded\cameraa{; we thus call this set} the \emph{quorum} $Q_p$ of $R_p$.}

\camerac{A proposer may hold several strategies, e.g., one per quorum it may
use~(\S\ref{sec:constraints:examples}), and commit when any is met; the
two conditions below apply to each.}

\paragraph{Validity.}\cameraz{Compare two executions in which the same
proposer $p$ commits different values using $R_p$. Only the members of $Q_p$
\camerab{need to}
have accepted when $p$ commits, so the two executions may differ only at
$Q_p$. Recoverability therefore
yields the \emph{validity} condition: \emph{every quorum must contain more than $f$
processes} (i.e., $|Q_p| > f$). We call
\camerac{a requirement that satisfies this condition} \emph{valid}.}

\paragraph{\cameraa{Compatibility}.}\cameraz{\camerac{Consider requirements $R_p$ and $R_q$ for distinct proposers
$p$ and $q$.
Their quorums must intersect: otherwise, both requirements could be met
concurrently for different values, allowing conflicting commits and
violating Agreement.
Yet intersection alone is insufficient: recoverability requires
\cameraa{\emph{evidence}---information in local states that distinguishes
conflicting commits---}to survive $f$ failures.}}

\cameraz{Figure~\ref{fig:intersection} \camerac{shows why}.
Suppose
\camerab{$n{=}7$},
$f{=}3$, and $Q_p\cap Q_q=\{x\}$. Consider two executions that differ
only at $p$, $q$, and $x$.
\camerab{In both, $p$ proposes $v_p$ and $q$ proposes
$v_q{\not=}v_p$.}
In one execution,
$x$ accepts $v_p$ and $p$ commits;
in the other, $x$ accepts $v_q$ and $q$ commits \camerac{(not shown)}. If $p$, $q$, and $x$
then fail, the survivors have identical local states and cannot
distinguish the two.}

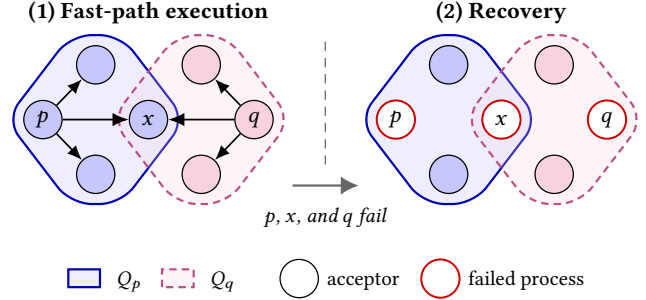
\begin{figure}[t]
\centering
\resizebox{\columnwidth}{!}{%
\begin{tikzpicture}[
  x=1cm,y=1cm,
  kcacc/.style={
    circle, draw=black, fill=white,
    minimum size=5.8mm, inner sep=0pt, line width=.55pt
  },
  kcfail/.style={
    circle, draw=red!80!black, fill=white,
    minimum size=5.8mm, inner sep=0pt, line width=.85pt
  },
  kcqp/.style={
    draw=blue!70!black, fill=blue!22,
    fill opacity=.30, line width=.95pt
  },
  kcqq/.style={
    draw=magenta!75!black, fill=magenta!22,
    fill opacity=.24, line width=.95pt,
    dashed, dash pattern=on 3.2pt off 2.2pt
  },
  kctitle/.style={font=\normalsize\bfseries},
  kcann/.style={font=\small\itshape},
  kcleg/.style={font=\small}
]

\def\dx{0.80}
\def\dy{0.83}
\def\xX{1.60}
\def\qX{3.20}
\def\rightshift{5.35}
\def\titleY{1.62}
\def\sepX{4.28}
\def\sepTop{1.18}
\def\sepBot{-1.18}
\def\annY{-1.48}
\def\legendY{-2.42}

\newcommand{\blueblob}{%
  \path[kcqp,rounded corners=9pt]
    (-0.56,0.00) --
    ( 0.42,1.28) --
    ( 1.18,1.28) --
    ( 2.16,0.00) --
    ( 1.18,-1.28) --
    ( 0.42,-1.28) -- cycle;
}
\newcommand{\pinkblob}{%
  \path[kcqq,rounded corners=9pt]
    ( 1.04,0.00) --
    ( 2.02,1.28) --
    ( 2.78,1.28) --
    ( 3.76,0.00) --
    ( 2.78,-1.28) --
    ( 2.02,-1.28) -- cycle;
}

\begin{scope}
  \node[kctitle] at (1.60,\titleY) {(1) Fast-path execution};

  \blueblob
  \pinkblob

  \node[kcacc,fill=blue!22] (p) at (0,0) {$p$};
  \node[kcacc,fill=blue!22] (pu) at (\dx,\dy) {};
  \node[kcacc,fill=blue!22] (pd) at (\dx,-\dy) {};
  \node[kcacc,fill=blue!22] (x) at (\xX,0) {$x$};

  \node[kcacc,fill=magenta!22] (qu) at (2.40,\dy) {};
  \node[kcacc,fill=magenta!22] (qd) at (2.40,-\dy) {};
  \node[kcacc,fill=magenta!22] (q) at (\qX,0) {$q$};

  \draw[-{Latex[length=2.25mm,width=1.85mm]},line width=.72pt]
    (p) -- (pu);
  \draw[-{Latex[length=2.25mm,width=1.85mm]},line width=.72pt]
    (p) -- (pd);
  \draw[-{Latex[length=2.45mm,width=2.05mm]},line width=.76pt]
    (p) -- (x);

  \draw[-{Latex[length=2.25mm,width=1.85mm]},line width=.72pt]
    (q) -- (qu);
  \draw[-{Latex[length=2.25mm,width=1.85mm]},line width=.72pt]
    (q) -- (qd);
  \draw[-{Latex[length=2.45mm,width=2.05mm]},line width=.76pt]
    (q) -- (x);
\end{scope}

\def\triY{-1.00}
\draw[black!60,line width=.55pt,dashed,dash pattern=on 4pt off 3pt]
  (\sepX,\sepTop) -- (\sepX,\triY+.32);

\draw[black!60,-{Latex[length=2.7mm,width=2.8mm]},line width=.95pt]
  (3.785,\triY) -- (4.785,\triY);

\begin{scope}[shift={(\rightshift,0)}]
  \node[kctitle] at (1.60,\titleY) {(2) Recovery};

  \blueblob
  \pinkblob

  \node[kcacc,fill=blue!22] at (\dx,\dy) {};
  \node[kcacc,fill=blue!22] at (\dx,-\dy) {};
  \node[kcacc,fill=magenta!22] at (2.40,\dy) {};
  \node[kcacc,fill=magenta!22] at (2.40,-\dy) {};

  \node[kcfail] at (0,0) {$p$};
  \node[kcfail] at (\xX,0) {$x$};
  \node[kcfail] at (\qX,0) {$q$};
\end{scope}

\node[kcann] at (4.28,\annY) {$p$, $x$, and $q$ fail};

\begin{scope}[shift={(3.635,\legendY)}]
  \path[kcqp] (-3.25,-.12) rectangle (-2.77,.15);
  \node[kcleg,anchor=base west] at (-2.63,-.09) {$Q_p$};

  \path[kcqq] (-1.83,-.12) rectangle (-1.35,.15);
  \node[kcleg,anchor=base west] at (-1.21,-.09) {$Q_q$};

  \node[kcacc] at (0.25,.01) {};
  \node[kcleg,anchor=base west] at (.55,-.09) {acceptor};

  \node[kcfail] at (2.39,.01) {};
  \node[kcleg,anchor=base west] at (2.69,-.09) {failed process};
\end{scope}

\end{tikzpicture}%
}

\caption{\camerac{Loss of intersection evidence.
(1) $x$ accepts $p$'s value, allowing $p$ to commit (return messages are omitted).
(2) After $p$, $x$, and $q$ fail, no survivor records $x$'s acceptance.}}
\label{fig:intersection}
\Description{Diagram of seven processes where the quorums of proposers p and q intersect only at process x. In the first panel, x accepts p's value and p commits. In the second panel, p, q, and x have failed, and no surviving process records x's acceptance.}
\end{figure}

\cameraz{\camerac{We now ask which processes are guaranteed to distinguish
the two executions.
Every process in $Q_p\cap Q_q$ does so directly, because it accepts a
different value in each execution.
A process outside the intersection is guaranteed to distinguish them only
if $R_p$ or $R_q$ requires it to witness an acceptance in the intersection
before the corresponding commit.}
Recoverability therefore yields the
\emph{\cameraa{compatibility}} condition: \emph{every pair of requirements
$R_p$ and $R_q$ must together require \camerad{more than $f$} processes to witness an
acceptance by some process in $Q_p\cap Q_q$.} The condition is symmetric, but
permits asymmetric evidence: one requirement may contribute most witnesses.
We call two requirements satisfying \camerac{this condition} \emph{compatible}.}

\cameraz{For example, $p$, $q$, and $x$ in Figure~\ref{fig:intersection}
provide three distinguishing witnesses. If $R_p$ also requires another process
\cameraa{$w$ to have recorded $x$'s acceptance (i.e., $x\in R_p[w]$), then
$w$ becomes a fourth witness, making $R_p$ and $R_q$ compatible for $f{=}3$.}
\cameraz{Each witness $w$ that $R_p$ requires costs a causal chain: from $p$ to some $x\in Q_p\cap Q_q$, then to $w$, and back to $p$.}}

\medskip
\cameraz{Together, validity and \cameraa{compatibility} exactly characterize
correct fast paths under our model.
\cameraa{Both conditions are necessary for any correct fast path, even in
protocols not defined in terms of acceptance.}
\camerac{Intuitively, for each proposer we can represent the causal message
chains preceding its fast decision as a knowledge requirement.
Because the protocol must satisfy the recoverability property, these induced
requirements satisfy both conditions.}
\cameraz{Conversely, any collection of requirements satisfying both conditions
admits safe recovery.}
We prove \cameraz{the necessity of both conditions} in
\if\extended 1%
\camerad{Appendix~\ref{sec:fast-path-proofs} (\Cref{theorem:necessity})},
\else%
the \cameraz{extended version~\cite{extended}},
\fi%
and Sections~\ref{sec:template} and~\ref{sec:optimizer} demonstrate
\cameraz{their} sufficiency.}

\subsection{Prior protocols and recoverable evidence}
\label{sec:constraints:examples}

\cameraz{Prior protocols can be understood through the knowledge requirements
induced by their common-case communication.}
\cameraz{Each protocol differs in where \cameraa{the evidence required for compatibility} lies, how it is
spread, and how it is recovered, yielding different latency profiles
across topologies.}
Figure~\ref{fig:recoverable-patterns} shows the main cases discussed below.

\paragraph{Paxos, Multi-Paxos, and Pando.}
\cameraz{These protocols simplify recovery by electing one proposer before
values are accepted, so only one path is active at a time.
Paxos \camerac{and Pando perform this election in their respective first phases~\cite{paxos,pando}, while}
Multi-Paxos amortizes it across values~\cite{multipaxos}.
Once a leader is elected, recovery only needs to preserve its value, so any quorum larger than $f$ suffices.}

\cameraz{\camerad{Their induced requirements nonetheless satisfy both conditions.}
For Paxos and Pando, meeting the induced requirement first entails
having a majority accept the value---any two such majorities intersect---and
then making a majority witness these acceptances, ensuring that evidence
of the intersection survives $f$ crashes.
\camerac{In Multi-Paxos, all quorums intersect at the}
stable leader, whose acceptance is witnessed by a majority.}
\cameraz{\sysname's optimal fast path is therefore no slower than these
protocols' common-case paths.}

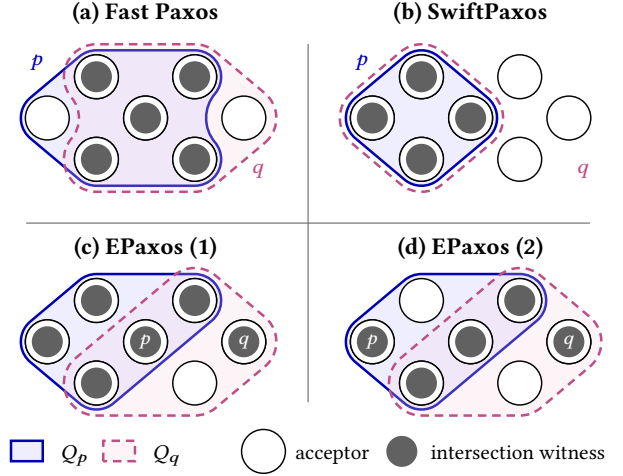
\begin{figure}[t]
\centering
\begin{tikzpicture}[
  x=1cm,y=1cm,
  kcacc/.style={
    circle, draw=black, fill=white,
    minimum size=5.8mm, inner sep=0pt, line width=.55pt
  },
  kcqp/.style={
    draw=blue!70!black, fill=blue!22,
    fill opacity=.30, line width=.95pt
  },
  kcqq/.style={
    draw=magenta!75!black, fill=magenta!22,
    fill opacity=.24, line width=.95pt,
    dashed, dash pattern=on 3.2pt off 2.2pt
  },
  kcqqline/.style={
    draw=magenta!75!black,
    line width=.95pt,
    dashed, dash pattern=on 3.2pt off 2.2pt
  },
  kctitle/.style={font=\small\bfseries},
  kcprop/.style={font=\footnotesize\bfseries},
  kcinprop/.style={font=\scriptsize\bfseries,text=white},
  kcleg/.style={font=\footnotesize}
]


\def\kcDX{4.30}

\def\kcTopY{1.39}
\def\kcBottomY{-1.57}

\def\kcTopTitleY{1.39}
\def\kcBottomTitleY{1.24}

\def\kcCrossHalfWidth{3.73}
\def\kcCrossHalfHeight{2.37}

\def\kcLegendY{-3.03}

\def\kcwr{2.05mm}

\newcommand{\kccoords}{%
  \coordinate (L)  at (-1.30,  0.00);
  \coordinate (UL) at (-0.650, 0.546);
  \coordinate (C)  at ( 0.00,  0.00);
  \coordinate (LL) at (-0.650, -0.546);
  \coordinate (UR) at (0.650, 0.546);
  \coordinate (LR) at (0.650, -0.546);
  \coordinate (R)  at ( 1.30,  0.00);
}

\newcommand{\kcdrawacc}{%
  \foreach \n in {L,UL,C,LL,UR,LR,R} {%
    \node[kcacc] at (\n) {};%
  }%
}

\newcommand{\kcwit}[1]{%
  \node[kcacc] at (#1) {};
  \fill[black!62] (#1) circle[radius=\kcwr];
}

\newcommand{\kcwitlabel}[2]{%
  \node[kcacc] at (#1) {};
  \fill[black!62] (#1) circle[radius=\kcwr];
  \node[kcinprop] at (#1) {#2};
}


\newcommand{\kcfastblue}{%
  \path[kcqp]
    ([shift={(-229.97:0.35)}]L) arc[start angle=-229.97, end angle=-130.03, radius=0.35]
    -- ([shift={(-130.03:0.35)}]LL) arc[start angle=-130.03, end angle=-90.00, radius=0.35]
    -- ([shift={(-90.00:0.35)}]LR) arc[start angle=-90.00, end angle=40.03, radius=0.35]
    arc[start angle=-139.97, end angle=-220.03, radius=0.499]
    arc[start angle=-40.03, end angle=90.00, radius=0.35]
    -- ([shift={(90.00:0.35)}]UL) arc[start angle=90.00, end angle=130.03, radius=0.35]
    -- cycle;
}
\newcommand{\kcfastpink}{%
  \path[kcqq]
    ([shift={(-49.97:0.43)}]R) arc[start angle=-49.97, end angle=49.97, radius=0.43]
    -- ([shift={(49.97:0.43)}]UR) arc[start angle=49.97, end angle=90.00, radius=0.43]
    -- ([shift={(90.00:0.43)}]UL) arc[start angle=90.00, end angle=220.03, radius=0.43]
    arc[start angle=40.03, end angle=-40.03, radius=0.419]
    arc[start angle=139.97, end angle=270.00, radius=0.43]
    -- ([shift={(-90.00:0.43)}]LR) arc[start angle=-90.00, end angle=-49.97, radius=0.43]
    -- cycle;
}

\newcommand{\kcswiftfill}{%
  \path[kcqp]
    ([shift={(-229.97:0.35)}]L) arc[start angle=-229.97, end angle=-130.03, radius=0.35]
    -- ([shift={(-130.03:0.35)}]LL) arc[start angle=-130.03, end angle=-49.97, radius=0.35]
    -- ([shift={(-49.97:0.35)}]C) arc[start angle=-49.97, end angle=49.97, radius=0.35]
    -- ([shift={(49.97:0.35)}]UL) arc[start angle=49.97, end angle=130.03, radius=0.35]
    -- cycle;
}
\newcommand{\kcswiftdash}{%
  \path[kcqqline]
    ([shift={(-229.97:0.43)}]L) arc[start angle=-229.97, end angle=-130.03, radius=0.43]
    -- ([shift={(-130.03:0.43)}]LL) arc[start angle=-130.03, end angle=-49.97, radius=0.43]
    -- ([shift={(-49.97:0.43)}]C) arc[start angle=-49.97, end angle=49.97, radius=0.43]
    -- ([shift={(49.97:0.43)}]UL) arc[start angle=49.97, end angle=130.03, radius=0.43]
    -- cycle;
}

\newcommand{\kcepblue}{%
  \path[kcqp]
    ([shift={(-229.97:0.35)}]L) arc[start angle=-229.97, end angle=-130.03, radius=0.35]
    -- ([shift={(-130.03:0.35)}]LL) arc[start angle=-130.03, end angle=-49.97, radius=0.35]
    -- ([shift={(-49.97:0.35)}]UR) arc[start angle=-49.97, end angle=90.00, radius=0.35]
    -- ([shift={(90.00:0.35)}]UL) arc[start angle=90.00, end angle=130.03, radius=0.35]
    -- cycle;
}
\newcommand{\kceppink}{%
  \path[kcqq]
    ([shift={(-49.97:0.43)}]R) arc[start angle=-49.97, end angle=49.97, radius=0.43]
    -- ([shift={(49.97:0.43)}]UR) arc[start angle=49.97, end angle=130.03, radius=0.43]
    -- ([shift={(-229.97:0.43)}]LL) arc[start angle=-229.97, end angle=-90.00, radius=0.43]
    -- ([shift={(-90.00:0.43)}]LR) arc[start angle=-90.00, end angle=-49.97, radius=0.43]
    -- cycle;
}

\begin{scope}[shift={(-\kcDX/2,\kcTopY)}]
  \node[kctitle] at (0,\kcTopTitleY) {(a) Fast Paxos};

  \kccoords
  \kcfastblue
  \kcfastpink

  \kcdrawacc

  \foreach \n in {UL,C,LL,UR,LR} {%
    \kcwit{\n}%
  }

  \node[kcprop, text=blue!70!black]    at (-1.43,0.74) {$p$};
  \node[kcprop, text=magenta!75!black] at (1.49,-0.70) {$q$};
\end{scope}

\begin{scope}[shift={(\kcDX/2,\kcTopY)}]
  \node[kctitle] at (0,\kcTopTitleY) {(b) SwiftPaxos};

  \kccoords
  \kcswiftfill
  \kcswiftdash

  \kcdrawacc

  \foreach \n in {L,UL,C,LL} {%
    \kcwit{\n}%
  }

  \node[kcprop, text=blue!70!black]    at (-1.43,0.74) {$p$};
  \node[kcprop, text=magenta!75!black] at (1.49,-0.70) {$q$};
\end{scope}

\begin{scope}[shift={(-\kcDX/2,\kcBottomY)}]
  \node[kctitle] at (0,\kcBottomTitleY) {(c) EPaxos (1)};

  \kccoords
  \kcepblue
  \kceppink

  \kcdrawacc

  \kcwit{L}
  \kcwit{UL}
  \kcwitlabel{C}{$p$}
  \kcwit{LL}
  \kcwit{UR}
  \kcwitlabel{R}{$q$}
\end{scope}

\begin{scope}[shift={(\kcDX/2,\kcBottomY)}]
  \node[kctitle] at (0,\kcBottomTitleY) {(d) EPaxos (2)};

  \kccoords
  \kcepblue
  \kceppink

  \kcdrawacc

  \kcwitlabel{L}{$p$}
  \kcwit{C}
  \kcwitlabel{R}{$q$}
  \kcwit{UR}
  \kcwit{LL}
\end{scope}

\draw[black!55,line width=.45pt]
  (0,-\kcCrossHalfHeight) -- (0,\kcCrossHalfHeight);
\draw[black!55,line width=.45pt]
  (-\kcCrossHalfWidth,0) -- (\kcCrossHalfWidth,0);

\begin{scope}[shift={(-0.66,\kcLegendY)}]
  \path[kcqp] (-3.27,-.12) rectangle (-2.85,.15);
  \node[kcleg, anchor=base west] at (-2.73,-.085) {$Q_p$};

  \path[kcqq] (-2.05,-.12) rectangle (-1.63,.15);
  \node[kcleg, anchor=base west] at (-1.51,-.085) {$Q_q$};

  \node[kcacc] at (0.07,.01) {};
  \node[kcleg, anchor=base west] at (.37,-.085) {acceptor};

  \fill[black!62] (1.90,.01) circle[radius=\kcwr];
  \node[kcleg, anchor=base west] at (2.16,-.085) {intersection witness};
\end{scope}

\end{tikzpicture}

\caption{\cameraz{Intersection witnesses in prior protocols for $n{=}7$ and $f{=}3$.
  The witnesses of Fast Paxos and SwiftPaxos lie directly in quorum intersections, while EPaxos also uses witnesses outside them \camerad{because each proposer both witnesses its quorum and is witnessed by it}.}}
\label{fig:recoverable-patterns}
\Description{Four diagrams of seven processes showing proposer quorums and the processes that witness an acceptance in the quorum intersection: (a) Fast Paxos with quorums of six processes, (b) SwiftPaxos with a shared quorum of four processes, and (c, d) EPaxos with quorums of five processes, where proposers in their own quorums add witnesses outside the intersection.}
\end{figure}

\paragraph{Fast Paxos.}
\cameraz{Fast Paxos allows concurrent proposers to commit in one round
trip~\cite{fast-paxos}. Because there is no second round to spread
evidence, the processes in the intersection of two quorums must
themselves provide the $f{+}1$ witnesses required for recovery.
Moreover, Fast Paxos does not require proposers to commit to quorums in advance: a proposer
may commit with any sufficiently large set of acceptors.
Thus, every two possible fast quorums of size $q$ must intersect in more
than $f$ processes, requiring
$2q{-}n{>}f$, or $q{>}(n{+}f){/}2$.
For example, with $n{=}7$ and $f{=}3$, a fast quorum needs $6$ processes,
rather than the minimum $f{+}1{=}4$~(\Cref{fig:recoverable-patterns}(a)).}

\paragraph{SwiftPaxos.}
SwiftPaxos also supports multiple concurrent proposers and admits two kinds of fast paths~\cite{swiftpaxos}.
In one mode, it uses large quorums à la Fast Paxos. In another, \camerad{all proposers share a single} \cameraz{optimally small}
\camerad{quorum} of size $f{+}1$~(\Cref{fig:recoverable-patterns}(b)). Either way, intersections themselves are large enough for direct recovery.
\cameraz{Moreover, since
\camerac{all quorums intersect at a designated leader}, these
requirements \camerac{are compatible} with Multi-Paxos requirements that use it.}

\paragraph{EPaxos.}
\cameraz{EPaxos also commits in one round trip, but unlike \cameraa{Fast Paxos} and SwiftPaxos, each
proposer is itself an acceptor and belongs to its own fast quorum~\cite{epaxos}.
This lets the protocol create witnesses outside the quorum intersection.
For example, with $n{=}7$ and $f{=}3$, an EPaxos fast quorum has $5$
processes, so two quorums may intersect in only $3$---fewer than the
$f{+}1{=}4$ witnesses required for recovery.
If the intersection contains proposer $p$, then $p$'s acceptance lies in
the intersection; because $p$ accepts before sending its value, every
process in $Q_p$ witnesses that acceptance
(\Cref{fig:recoverable-patterns}(c)).
Otherwise, the intersection excludes both proposers: its members witness
their own acceptances, while the replies make $p$ and $q$ additional
witnesses, yielding at least $|Q_p\cap Q_q|{+}2{>}f$ witnesses
(\Cref{fig:recoverable-patterns}(d)).
Thus, EPaxos spreads evidence beyond the intersection within the same
round trip, allowing smaller intersections than Fast Paxos.}

\medskip
\cameraz{In practice, SwiftPaxos and EPaxos do not decide one command per
slot, but dependencies for each command.
This distinction is not fundamental.
Their correctness relies on \emph{Visibility} and \emph{Agreement}~\cite{epaxosfixed}:
for any two conflicting commands, at least one must depend on the other,
and all processes must agree on each command's dependencies.
Thus, for conflicting commands $c_p$ and $c_q$, deciding whether $c_q$
belongs to $c_p$'s dependencies is itself consensus and must remain
recoverable.
Their fast paths therefore satisfy the same validity and
\cameraa{compatibility} conditions.}

\medskip

\cameraz{These examples show that familiar protocols satisfy the same
recoverability conditions through different evidence patterns, yielding
different latency profiles across topologies.
\sysname captures these patterns as per-proposer knowledge requirements, which instantiate a generic fast-path template.}

\section{\sysname's Fast-Path Template}
\label{sec:template}
\label{sec:ac}

\cameraz{\sysname turns the recoverability conditions of
\Cref{sec:constraints} into a generic adopt-commit template.}
\camerac{This template is parameterized by a knowledge
requirement per proposer. \camerab{This knowledge requirement is checked
for validity and pairwise \camerad{compatibility}~(\S\ref{sec:constraints:condition})
by Algorithm~\ref{alg:ac:compatible}.}
\Cref{sec:optimizer} explains how to select an optimal combination of requirements for a given deployment.}

\cameraz{The algorithms below describe
\cameraa{a single instance of \camerac{a} one-shot adopt-commit \camerac{protocol}---\camerad{our} SMR uses one
such instance as the fast path for consensus in each log slot}~(\S\ref{sec:smr}).}

\cameraz{After an overview~(\S\ref{sec:ac:overview}), we describe evidence
dissemination~(\S\ref{sec:ac:disseminate}),
\camera{commitment}~(\S\ref{sec:ac:commit}), recovery~(\S\ref{sec:ac:adopt}),
and non-voting proposers~(\S\ref{sec:ac:nonvoting}).}

\subsection{Overview}
\label{sec:ac:overview}

\begin{lstlisting}[caption={Valid and Compatible Requirements},label={alg:ac:compatible},float]
def valid(req): return |req.keys()| > f 

def compatible(req_a, req_b):@\label{alg:comp:def}@
  intersection = req_a.keys() @$\cap$@ req_b.keys()
  // Acceptors of A/B that must witness the inter.
  witnesses_a = {p s.t. req_a[p] @$\cap$@ intersection@\;\;$\neq\varnothing$@}@\label{alg:comp:witnesses_a}@
  witnesses_b = {p s.t. req_b[p] @$\cap$@ intersection@\;\;$\neq\varnothing$@}@\label{alg:comp:witnesses_b}@
  return |witnesses_a @$\cup$@ witnesses_b| > f@\label{alg:comp:bigger than f}@
\end{lstlisting}

\ContinueLineNumber

\begin{lstlisting}[float,caption={\sysname's Adopt-Commit's Propose Method},label={alg:ac:propose}]
def AdoptCommit::Propose(v):
  send Accept&Spread(v,{}) to myself // init spread@\label{alg:ac:propose:disseminate}@
  wait until can_commit(v) or DoAdopt() was received@\label{alg:ac:propose:until}@
  if wait was ended by can_commit(v): @\label{alg:ac:propose:can_commit}@
    trigger AdoptCommit::Commit(v)@\label{alg:ac:propose:commit}@
  else: // DoAdopt() was received, recovery@\label{alg:ac:propose:conflict}@
    send Freeze() to all@\label{alg:ac:propose:freeze}@ // including self
    wait for n-f Frozen(*) replies@\label{alg:ac:propose:wait}@
    v' = potential_commit(replies)@\label{alg:ac:propose:potential}@
    if v' == @$\bot$@: v' = v // any proposed value works@\label{alg:ac:propose:any}@
    trigger AdoptCommit::Adopt(v')@\label{alg:ac:propose:adopt}@
\end{lstlisting}

\cameraz{Algorithm~\ref{alg:ac:propose} gives the top-level
\t{Propose} method.
To propose a value $v$, process $p$ starts disseminating evidence about
$v$'s acceptance by invoking its own \t{Accept\&Spread} handler
(\cref{alg:ac:propose:disseminate}).
It commits once it learns enough evidence to meet \camerac{its requirement} $R_p$
(\crefrange{alg:ac:propose:until}{alg:ac:propose:commit});
a conflict or fault suspicion instead abandons the fast path and
triggers recovery~(\cref{alg:ac:propose:conflict}).}

\cameraz{Recovery begins by freezing $n{-}f$ processes
(\crefrange{alg:ac:propose:freeze}{alg:ac:propose:wait}).
Because every valid requirement has more than $f$ witnesses, the frozen
set contains at least one from each requirement.
Once frozen, \camerad{such} a witness cannot gain evidence, \camerad{so} its \t{Frozen} reply
\camerad{either rules out any commit requiring evidence the witness lacks, or}
records the evidence a commit could rely on.
These replies form the \emph{\textbf{K}nowledge \textbf{Census}} that
gives \sysname its name.
From this census, $p$ determines what value, if any, 
\camerab{can be}
a possible
fast commitment~(\cref{alg:ac:propose:potential}).
If \camerac{this value exists}, $p$ adopts it
(\cref{alg:ac:propose:adopt}); otherwise, any proposed value is safe,
so $p$ adopts its own~(\cref{alg:ac:propose:any}).}

\subsection{Disseminating evidence}
\label{sec:ac:disseminate}

\cameraz{Algorithm~\ref{alg:ac:disseminate} disseminates evidence as
quickly as possible.}
\cameraz{A process accepts the first value it receives and records it in
\t{accepted}~(\cref{alg:ac:ifaccepted}).
It never changes its acceptance.
If it later receives evidence for another value, it detects a conflict
and asks proposers to abandon the fast path
(\crefrange{alg:ac:ifdiffaccepted}{alg:ac:conflict1}).}

\cameraz{The \t{known\_acceptors} array records witnessed evidence.
\cameraa{At each process, when $\t{accepted}=v$,
$\t{known\_acceptors}[q]$ represents its current view of which processes $q$
has witnessed accepting $v$.
Thus, $\t{known\_acceptors}[\t{me}]$ is the process's own first-order
evidence, while other entries are its second-order view of other
processes' evidence.
Whenever this evidence grows, the process broadcasts the updated array}
(\crefrange{alg:ac:grew}{alg:ac:broadcast}).
\cameraa{Every acceptance represented in $\t{known\_acceptors}[q]$ originates
in a report from $q$, even if relayed by other processes; consequently, this
entry attributes to $q$ only acceptances that $q$ has witnessed.}}

\begin{figure}
    \centering
    \Description{Pseudocode of the Accept\&Spread and Freeze message handlers.}
    \input{pseudocode/disseminate}
\end{figure}

\cameraz{Without conflicts, evidence therefore propagates in three
conceptual steps: (1) processes accept the value, (2) processes learn which
acceptances occurred, and (3) the proposer learns what evidence each
witness can report.
\Cref{fig:dissemination} illustrates this propagation in a five-process
ring.}

\cameraz{Upon receiving \t{Freeze}, a process stops
handling further \t{Accept\&Spread} messages \cameraa{for any
value}~(\cref{alg:ac:conflict2}) and replies with
its accepted value and $\t{known\_acceptors}[\t{me}]$
(\cref{alg:ac:frozen}).
\cameraa{Because a process's own entry only grows until it freezes, the entry
in its \t{Frozen} reply is a superset of every earlier version, including any
version on which a proposer may have relied.}}

\begin{figure} 
\begin{minipage}{\columnwidth}
    \newlength{\kccolsix}\newlength{\kccolthree}%
    \setlength{\kccolsix}{\dimexpr(\textwidth-12\tabcolsep-7\arrayrulewidth)/6\relax}%
    \settowidth{\kccolthree}{Excludes P1's}%
    \setlength{\kccolthree}{\dimexpr(\textwidth-\kccolthree-8\tabcolsep-5\arrayrulewidth)/3\relax}%
    \centering
    \resizebox{\columnwidth}{!}{%
\begin{tikzpicture}[
  x=1cm,y=1cm,
  proc/.style={circle,draw=black,fill=white,minimum size=5.0mm,inner sep=0pt,line width=.55pt,font=\large},
  pwhite/.style={proc,fill=white},
  pblue/.style={proc,fill=blue!22},
  ppink/.style={proc,fill=magenta!22},
  tlabel/.style={font=\Large},
  link/.style={draw=black,line width=.58pt},
  barr/.style={draw=blue!55,line width=.90pt,-{Latex[length=2.0mm,width=1.75mm]}},
  marr/.style={draw=magenta!65!black,line width=.90pt,-{Latex[length=2.0mm,width=1.75mm]}},
  trans/.style={draw=black,line width=.78pt,-{Latex[length=2.4mm,width=2.1mm]}}
]

\def\StepX{2.10}

\def\TitleY{1.36}
\def\TopY{0.61}
\def\SideY{0.13}
\def\BottomY{-0.66}
\def\SideX{0.66}
\def\BottomX{0.41}

\def\ClockX{1.05}
\def\ClockY{1.10}
\def\ClockR{0.10}
\def\TransY{0.80}
\def\TransL{0.78}
\def\TransR{1.32}


\def\LongStart{0.10}
\def\LongEnd{0.64}

\def\TopDx{0.055}
\def\TopDy{0.085}
\def\RightDx{0.095}
\def\RightDy{-0.015}
\def\BottomDy{-0.105}

\def\LeftBlueDx{-0.12}
\def\LeftBlueDy{-0.030}
\def\LeftMagDx{0.09}
\def\LeftMagDy{0.04}

\newcommand{\edgearrowxy}[7]{%
  \draw[#1]
    ($ (#2)!#4!(#3) + (#6,#7) $)
    --
    ($ (#2)!#5!(#3) + (#6,#7) $);
}

\newcommand{\ringcoords}{%
  \coordinate (n1) at (0,\TopY);
  \coordinate (n2) at (\SideX,\SideY);
  \coordinate (n3) at (\BottomX,\BottomY);
  \coordinate (n4) at (-\BottomX,\BottomY);
  \coordinate (n5) at (-\SideX,\SideY);
}
\newcommand{\drawring}{\draw[link] (n1)--(n2)--(n3)--(n4)--(n5)--cycle;}
\newcommand{\drawnodes}[5]{%
  \node[#1] at (n1) {1};
  \node[#2] at (n2) {2};
  \node[#3] at (n3) {3};
  \node[#4] at (n4) {4};
  \node[#5] at (n5) {5};
}

\newcommand{\bluefromone}{%
  \edgearrowxy{barr}{n1}{n5}{\LongStart}{\LongEnd}{-\TopDx}{\TopDy}%
  \edgearrowxy{barr}{n1}{n2}{\LongStart}{\LongEnd}{ \TopDx}{\TopDy}%
}

\newcommand{\bluetoone}{%
  \edgearrowxy{barr}{n5}{n1}{\LongStart}{\LongEnd}{-\TopDx}{\TopDy}%
  \edgearrowxy{barr}{n2}{n1}{\LongStart}{\LongEnd}{ \TopDx}{\TopDy}%
}

\newcommand{\bluesides}{%
  \edgearrowxy{barr}{n5}{n4}{\LongStart}{\LongEnd}{\LeftBlueDx}{\LeftBlueDy}%
  \edgearrowxy{barr}{n2}{n3}{\LongStart}{\LongEnd}{\RightDx}{\RightDy}%
}

\newcommand{\magleftup}{%
  \edgearrowxy{marr}{n4}{n5}{\LongStart}{\LongEnd}{\LeftMagDx}{\LeftMagDy}%
}

\newcommand{\magrightup}{%
  \edgearrowxy{marr}{n3}{n2}{\LongStart}{\LongEnd}{\RightDx}{\RightDy}%
}

\newcommand{\magbottomlr}{%
  \edgearrowxy{marr}{n4}{n3}{\LongStart}{\LongEnd}{0}{\BottomDy}%
}
\newcommand{\magbottomrl}{%
  \edgearrowxy{marr}{n3}{n4}{\LongStart}{\LongEnd}{0}{\BottomDy}%
}

\newcommand{\transition}{%
  \draw[line width=.46pt] (\ClockX,\ClockY) circle[radius=\ClockR];
  \draw[line width=.46pt,line cap=round]
    (\ClockX,\ClockY)--(\ClockX,\ClockY+0.065);
  \draw[line width=.46pt,line cap=round]
    (\ClockX,\ClockY)--(\ClockX+0.055,\ClockY+0.035);
  \draw[trans] (\TransL,\TransY)--(\TransR,\TransY);
}

\begin{scope}[shift={(0,0)}]
  \ringcoords
  \drawring
  \bluefromone
  \drawnodes{pblue}{pwhite}{pwhite}{pwhite}{pwhite}
  \node[tlabel] at (0,\TitleY) {$t{=}0$};
  \transition
\end{scope}

\begin{scope}[shift={(\StepX,0)}]
  \ringcoords
  \drawring
  \bluetoone
  \bluesides
  \magleftup
  \magbottomlr
  \drawnodes{pblue}{pblue}{pwhite}{ppink}{pblue}
  \node[tlabel] at (0,\TitleY) {$t{=}1$};
  \transition
\end{scope}

\begin{scope}[shift={(2*\StepX,0)}]
  \ringcoords
  \drawring
  \bluefromone
  \magrightup
  \magbottomrl
  \drawnodes{pblue}{pblue}{ppink}{ppink}{pblue}
  \node[tlabel] at (0,\TitleY) {$t{=}2$};
  \transition
\end{scope}

\begin{scope}[shift={(3*\StepX,0)}]
  \ringcoords
  \drawring
  \bluetoone
  \bluesides
  \magleftup
  \magbottomlr
  \drawnodes{pblue}{pblue}{ppink}{ppink}{pblue}
  \node[tlabel] at (0,\TitleY) {$t{=}3$};
  \transition
\end{scope}

\begin{scope}[shift={(4*\StepX,0)}]
  \ringcoords
  \drawring
  \drawnodes{pblue}{pblue}{ppink}{ppink}{pblue}
  \node[tlabel] at (0,\TitleY) {$t{=}4$};
\end{scope}

\end{tikzpicture}%
}%
    \captionof{figure}{\cameraz{Evidence spread in a five-process ring.
    Different colors represent different accepted values, and arrows
    represent broadcasts.
    At $t{=}0$, P1 starts spreading its value, which reaches P2 and P5
    at $t{=}1$ and is relayed onward.
    At the same time, P4 spreads its own value.
    At $t{=}2$, P3 accepts P4's value, so P3 and P4 ignore P1's.
    Processes broadcast until $t{=}4$.
    \Cref{tab:disseminate:known,tab:disseminate:known2} show the
    resulting evidence.}}
    \label{fig:dissemination}
    \Description{Five snapshots, at times 0 to 4, of five processes arranged in a ring. Process P1 spreads its value to P2 and P5, which relay it, while P4 spreads a different value that P3 accepts. Arrows show broadcasts between neighboring processes.}

    \vspace{2em}
    \captionof{table}{\cameraz{Processes' $1^{\text{st}}$-order evidence:
    \t{known\_acceptors[me]}}}
    \label{tab:disseminate:known}
    \vspace{-0.5em}
    \begin{tabular}{|*{6}{>{\centering\arraybackslash}p{\kccolsix}|}}
    \hline
    Process & t=0           & t=1           & t=2       & t=3       & t=4       \\ \hline
    P1 & \{1\}         & \{1\}         & \{1,2,5\} & \{1,2,5\} & \{1,2,5\} \\ \hline
    P2 & $\varnothing$ & \{1,2\}       & \{1,2\}   & \{1,2,5\} & \{1,2,5\} \\ \hline
    P3 & $\varnothing$ & $\varnothing$ & \{3,4\}   & \{3,4\}   & \{3,4\}   \\ \hline
    P4 & $\varnothing$ & \{4\}         & \{4\}     & \{3,4\}   & \{3,4\}   \\ \hline
    P5 & $\varnothing$ & \{1,5\}       & \{1,5\}   & \{1,2,5\} & \{1,2,5\} \\ \hline
    \end{tabular}

    \vspace{2em}
    \captionof{table}{\cameraz{Proposers' $2^{\text{nd}}$-order evidence: \t{known\_acceptors}}}
    \label{tab:disseminate:known2}
    \vspace{-0.5em}
    \begin{tabular}{|*{6}{>{\centering\arraybackslash}p{\kccolsix}|}}
\hline
   Process & t=0           & t=1     & t=2                                                                         & t=3                                                                         & t=4                                                        \\ \hline
P1 & 1:\{1\}       & 1:\{1\} & \begin{tabular}[c]{@{}c@{}}1:\{1,2,5\}\\ 2:\{1,2\}\\ 5:\{1,5\}\end{tabular} & \begin{tabular}[c]{@{}c@{}}1:\{1,2,5\}\\ 2:\{1,2\}\\ 5:\{1,5\}\end{tabular} & \begin{tabular}[c]{@{}c@{}}1:\{1,2,5\}\\ 2:\{1,2,5\}\\ 5:\{1,2,5\}\end{tabular} \\ \hline
P4 & $\varnothing$ & 4:\{4\} & 4:\{4\}                                                                     & \begin{tabular}[c]{@{}c@{}}3:\{3,4\}\\ 4:\{3,4\}\end{tabular}               & \begin{tabular}[c]{@{}c@{}}3:\{3,4\}\\ 4:\{3,4\}\end{tabular}                   \\ \hline
\end{tabular}

    \vspace{2em}
    \captionof{table}{\cameraz{Potential commits P2 computes at
    $t{\in}[0,4]$ for each set of \camerad{$n-f$} replies, based on
    \Cref{tab:disseminate:known}.}}
    \label{tab:adopt}
    \vspace{-0.5em}
\begin{tabular}{|c|*{3}{>{\centering\arraybackslash}p{\kccolthree}|}}
\hline
Replies             & t=0    & t=1    & $t{\in}[2,4]$\\ \hline
Includes P1's & $\bot$ & $\bot$ & \cameraz{P1's value} \\ \hline
Excludes P1's & $\bot$ & \cameraz{P1's value} & \cameraz{P1's value} \\ \hline
\end{tabular}
    \vspace{0.5em}
\end{minipage}
\end{figure}

\subsection{Committing a \cameraz{value}}
\label{sec:ac:commit}

\cameraz{Algorithm~\ref{alg:ac:commit} checks whether $R_p$ is met: after
verifying that it accepted $v$ itself, proposer $p$ checks for every key $q$
that its current view of $q$'s evidence covers $R_p[q]$
(\crefrange{alg:commit:for}{alg:commit:true}).}

\begin{figure}
    \centering
    \Description{Pseudocode of the can\_commit check and of fast-path abandonment on failure suspicion.}
    \input{pseudocode/commit}
    \input{pseudocode/fd}
\end{figure}

\cameraz{For example, return to the execution in
\Cref{fig:dissemination}.
Suppose \cameraa{P1} uses requirement
\t{\{1:\{1,2,5\},2:\{2\},5:\{5\}\}}, while \cameraa{P4} uses
\camerab{\t{\{4:\{3,4,5\},3:\{3\},5:\{5\}\}}.}
By \Cref{tab:disseminate:known2}, \cameraa{P1} meets its requirement at
$t{=}2$, one round trip after proposing, and may commit.
\cameraa{P4} never meets its requirement.}

\cameraz{Failures can prevent required evidence from arriving.
A proposer therefore abandons the fast path if it suspects one of its
required witnesses~(Algorithm~\ref{alg:ac:fd}).
As discussed in \S\ref{sec:bg:model}, a false suspicion is safe: it merely
triggers recovery earlier.}

\subsection{Adopting a \cameraz{value}}
\label{sec:ac:adopt}

\cameraz{When a proposer cannot commit, it freezes $n{-}f$ processes and
treats their \t{Frozen} replies as a census of the recoverable evidence.
Algorithm~\ref{alg:ac:adopt} tests each value $v$ appearing in this
census against the requirement of $v$'s proposer.
A reply rules out $v$ in either of two ways:
it shows that some required acceptor accepted a different value
(\crefrange{alg:ac:adopt:for3}{alg:ac:adopt:false2}), or
\cameraa{the reply}
comes from a
required witness that accepted $v$ but lacks evidence required by
$R_v$~(\crefrange{alg:ac:adopt:for2}{alg:ac:adopt:false1}).
Any value not ruled out remains a possible fast commitment.}

\begin{figure}
    \centering
    \Description{Pseudocode of the potential\_commit and could\_potentially\_commit functions used for recovery.}
    \input{pseudocode/adopt}
\end{figure}

\cameraz{The key safety invariant is:
\emph{if a value $v$ committed, then every set of $n{-}f$
\t{Frozen} replies preserves $v$ and rules out every conflicting
value.}
Validity ensures that more than $f$ required witnesses accepted $v$,
so every census contains a reply carrying $v$.
Because $R_v$ was actually met, no reply can rule $v$ out.
Now consider a conflicting value $v'$.
Compatibility gives more than $f$ witnesses, required by $v$ or $v'$,
that observed an acceptance in the intersection of their quorums.
Therefore, every census contains at least one such witness.
Its report either reveals that an acceptor required by $v'$ chose $v$,
or shows that a required witness for $v'$ lacked evidence that $R_{v'}$
would have demanded.
Thus, every conflicting $v'$ is ruled out.
If some value committed, it is therefore the only surviving candidate
and must be adopted.
This realizes the sufficiency direction of the conditions in
\S\ref{sec:constraints:condition}, of which a
\cameraz{full proof is given in the}
\if\extended 1%
appendix.%
\else%
\cameraz{extended version~\cite{extended}.}%
\fi%
}

\cameraz{As a concrete example, return to \Cref{fig:dissemination}\cameraa{.
Suppose a conflict or failure suspicion triggers recovery at P2 at some
$t{\in}[0,4]$, and P2 instantaneously receives \t{Frozen} replies from
any $n{-}f{=}3$ processes.}\footnote{\cameraz{After the freeze, this hypothetical
execution diverges from \Cref{fig:dissemination}.}}
\Cref{tab:adopt} lists the resulting potential commit \cameraa{for
every such reply set, grouped by whether it includes P1. At $t{=}1$, P1's reply
rules out its value, leaving no potential commit~($\bot$).
\camerac{This is} because its requirement calls for evidence of P2's and P5's
acceptances, which P1 has not yet witnessed. Without P1's reply, recovery cannot
observe this missing evidence, so P1's value remains possible.}
From $t{=}2$, when \cameraa{P1} can commit, \cameraa{P2} always adopts
\cameraa{P1's} value, regardless of which three processes \cameraa{reply}.}

\subsection{Non-voting processes}
\label{sec:ac:nonvoting}

\cameraz{So far, every proposer is a voting process whose failure counts
toward 
\camerab{the failure budget $f$.}
\camerad{In geo-replicated systems,} clients \camera{benefit from proposing
directly, \cameraa{which saves} a hop to a replica. \camerad{Yet}
\camerab{we should not include clients as regular proposer processes in the protocol,
  because doing so disrupts the resiliency of the system.
Indeed, clients are numerous and flaky---\camerad{more than a majority may crash}---so adding them as regular
  processes could cause executions to exceed the failure budget.}}
\sysname thus \cameraa{lets clients act as} \emph{non-voting processes}
\cameraa{that} \camerac{merely} propose values and disseminate evidence.}

\cameraz{A non-voting process never appears as a required witness or
acceptor in any requirement.
No commit therefore depends on its state, so its failure does not count
toward \camerab{the failure budget.}}

\cameraz{Non-voting processes otherwise run the same protocol, with two
exceptions.
Their evidence is not tracked: they do not add \t{\{me\}} in
\cref{alg:ac:knownacceptorsme1}, and
\cref{alg:ac:q} ranges only over voting processes.
They also ignore \t{Freeze} messages~(\cref{alg:ac:upon2}), since
recovery needs reports only from voting processes.
Accordingly, the $n$ used in
\cref{alg:ac:propose:wait} counts only voting processes.}

\section{Synthesizing Optimal Fast Paths}
\label{sec:optimizer}

The template of \Cref{sec:template} works with any valid and compatible requirement combination.
\cameraz{This section explains how \sysname finds \cameraa{one that minimizes the target latency metric} for a \camerac{given} deployment\cameraa{,
both initially and when conditions change~(\S\ref{sec:smr})}.}
We first turn topology measurements into candidate requirements~(\S\ref{sec:optimizer:budgets}), then choose a compatible combination that optimizes the target metric~(\S\ref{sec:optimizer:solve}), and finally reduce network traffic without affecting latency~(\S\ref{sec:optimizer:network}).\todo{A: missing 5.3, which is weird but fine. C: fine, and would be weird to list in between.}

\subsection{From topology to requirements}
\label{sec:optimizer:budgets}

The first step is to discretize the design space. Although time is continuous, a proposer's \cameraz{requirement} can change only when the proposer learns new evidence. Thus, each evidence arrival defines a relevant \emph{budget}: the earliest time at which the proposer can satisfy \cameraz{a stronger requirement}.

For each proposer $p$, \sysname simulates a single-proposer execution of the evidence-spreading algorithm from \Cref{sec:ac:disseminate}, using measured latencies. The simulation records when $p$ learns each piece of first- or second-order evidence. For every resulting budget $b$, \sysname derives \cameraz{the requirement} $R_{p,b}$, i.e., the pieces of evidence that $p$ can gather by time $b$.
\cameraz{Every candidate has the normal form of
\S\ref{sec:constraints:condition}}\camerad{, which Algorithm~\ref{alg:ac:disseminate} ensures by construction~(\crefrange{alg:ac:knownacceptorsme1}{alg:ac:knownacceptorsq}).}
Lastly, invalid requirements are discarded.

This step replaces an intractable search over arbitrary requirements (up to $2^{n\times{}n}$ combinations) with a search over the budgets at which the proposer learns something new. The derived requirements may include unnecessary evidence: for example, a 10\,ms budget also includes evidence that arrived after 1\,ms, even if compatibility does not need it. This extra evidence does not increase latency and can later be trimmed.

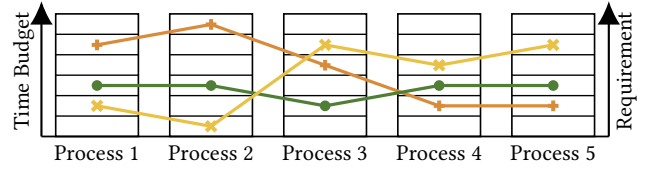
\begin{figure}
    \centering
    \definecolor{kcorange}{RGB}{216,139,57}
\definecolor{kcgreen}{RGB}{79,126,52}
\definecolor{kcyellow}{RGB}{233,194,69}
\resizebox{\columnwidth}{!}{%
\begin{tikzpicture}[
  x=1cm,y=1cm,
  box/.style={draw=black,line width=.65pt},
  grid/.style={draw=black,line width=.55pt},
  axis/.style={
    draw=black,line width=.85pt,
    -{Latex[length=3.3mm,width=3.2mm]}
  },
  plabel/.style={font=\small},
  axislabel/.style={font=\small},
  orange/.style={draw=kcorange,line width=1.28pt},
  green/.style={draw=kcgreen,line width=1.28pt},
  yellow/.style={draw=kcyellow,line width=1.28pt}
]

\def\BoxW{1.16}
\def\BoxH{1.72}
\def\Xone{0.00}
\def\Xtwo{1.60}
\def\Xthree{3.20}
\def\Xfour{4.80}
\def\Xfive{6.40}

\def\LeftAxisX{-0.20}
\def\RightAxisX{7.76}
\def\LabelY{-0.23}
\def\AxisBottom{0.00}
\def\AxisTop{1.90}

\def\CellH{0.2867}

\newcommand{\budgetbox}[2]{%
  \draw[box] (#1,0) rectangle ++(\BoxW,\BoxH);
  \foreach \i in {1,...,5} {
    \draw[grid] (#1,{\i*\CellH}) -- ++(\BoxW,0);
  }
  \node[plabel] at ({#1+\BoxW/2},\LabelY) {Process #2};
}

\newcommand{\plusmark}[2]{%
  \draw[orange,line width=1.6pt]
    (#1-0.075,#2) -- (#1+0.075,#2);
  \draw[orange,line width=1.6pt]
    (#1,#2-0.075) -- (#1,#2+0.075);
}

\newcommand{\circlemark}[2]{%
  \fill[kcgreen] (#1,#2) circle[radius=0.075];
}

\newcommand{\xmark}[2]{%
  \draw[yellow,line width=1.65pt]
    (#1-0.068,#2-0.068) -- (#1+0.068,#2+0.068);
  \draw[yellow,line width=1.65pt]
    (#1-0.068,#2+0.068) -- (#1+0.068,#2-0.068);
}

\budgetbox{\Xone}{1}
\budgetbox{\Xtwo}{2}
\budgetbox{\Xthree}{3}
\budgetbox{\Xfour}{4}
\budgetbox{\Xfive}{5}

\def\COne{0.58}
\def\CTwo{2.18}
\def\CThree{3.78}
\def\CFour{5.38}
\def\CFive{6.98}

\draw[black,line width=.72pt]
  (\LeftAxisX,0) -- (\RightAxisX,0);

\draw[axis]
  (\LeftAxisX,\AxisBottom) -- (\LeftAxisX,\AxisTop);
\node[axislabel,rotate=90]
  at (\LeftAxisX-0.25,0.92) {Time Budget};

\draw[axis]
  (\RightAxisX,\AxisBottom) -- (\RightAxisX,\AxisTop);
\node[axislabel,rotate=90]
  at (\RightAxisX+0.25,0.92) {Requirement};


\draw[orange]
  (\COne,1.285) --
  (\CTwo,1.572) --
  (\CThree,1.000) --
  (\CFour,0.430) --
  (\CFive,0.430);

\plusmark{\COne}{1.285}
\plusmark{\CTwo}{1.572}
\plusmark{\CThree}{1.000}
\plusmark{\CFour}{0.430}
\plusmark{\CFive}{0.430}

\draw[green]
  (\COne,0.715) --
  (\CTwo,0.715) --
  (\CThree,0.430) --
  (\CFour,0.715) --
  (\CFive,0.715);

\circlemark{\COne}{0.715}
\circlemark{\CTwo}{0.715}
\circlemark{\CThree}{0.430}
\circlemark{\CFour}{0.715}
\circlemark{\CFive}{0.715}

\draw[yellow]
  (\COne,0.430) --
  (\CTwo,0.143) --
  (\CThree,1.285) --
  (\CFour,1.000) --
  (\CFive,1.285);

\xmark{\COne}{0.430}
\xmark{\CTwo}{0.143}
\xmark{\CThree}{1.285}
\xmark{\CFour}{1.000}
\xmark{\CFive}{1.285}

\end{tikzpicture}%
}%

    \caption{Requirements optimization. Each column lists one proposer's time budgets; each cell represents a budget and the \cameraz{requirement} the proposer can meet within it. Lines show compatible combinations with one \cameraz{requirement per proposer}. The optimizer picks the best line for the objective.}
    \label{fig:budgets}
    \Description{Grid with one column per process, each listing that process's increasing time budgets and the requirement it can meet within each budget. Lines connect one budget per process to form compatible combinations, and the optimizer picks the best combination.}
\end{figure}

\subsection{Choosing optimal requirement combinations}
\label{sec:optimizer:solve}

Given the candidate budgets, \sysname selects one per proposer. Figure~\ref{fig:budgets} visualizes this search: a choice is feasible only if the selected budgets induce compatible requirements.

\Cref{optiprob} gives the average-latency version of the optimization problem. Here, $B_p$ is the set of budgets available to proposer $p$, $R_{p,b}$ is the requirement $p$ can meet within budget $b$, and $P_p$ is the probability that $p$ proposes. The objective minimizes expected fast-path latency, while the constraints enforce pairwise compatibility\cameraa{~(\S\ref{sec:constraints:condition})}.

\begin{problem}
\caption{Average-latency optimization problem}
\label{optiprob}
\vspace{-9pt}
{\setlength{\fboxrule}{0.5pt}\setlength{\fboxsep}{6pt}%
\noindent\fcolorbox{black}{gray!3}{\begin{minipage}{\dimexpr\columnwidth-2\fboxsep-2\fboxrule\relax}
\begin{itemize}[leftmargin=10pt, topsep=0pt]
    \item $B_p$: time budgets available to process $p$.
    \item $R_{p,b}$: requirement process $p$ can meet with budget $b$.
    \item $P_p$: probability that process $p$ is the proposer.
\end{itemize}
\newcommand{\argmin}{\mathop{\mathrm{arg\,min}}}
\newcommand{\vect}[1]{\mathbf{#1}}
\begin{equation*}
\label{eq:budget-opt}
\begin{aligned}
&\vect b^\star
= \argmin_{(b_1,\dots,b_n)\in B_1\times\cdots\times B_n}
   \; \sum_{p=1}^n P_p\,b_p \\
\text{s.t.}\quad
&compatible\big(R_{p,b_p},\,R_{q,b_q}\big) \quad \forall\, p,q\in\{1,\dotsc,n\}.
\end{aligned}
\end{equation*}
\end{minipage}}}
\end{problem}

By changing only the objective, \sysname can also target median or tail latency while keeping the same compatibility constraints. \cameraa{Because compatibility couples proposers' budgets, different objectives may select different Pareto-optimal solutions; across our evaluated deployments, their number ranges from one to over a hundred.} Whatever the objective, any feasible solution instantiates the same template and is correct by construction.

To make the search fast, \sysname prunes \emph{dominated} budgets.\todo{a: can be trimmed if need be} A budget is dominated if it is no cheaper for the objective and does not enable compatibility with any requirement that a cheaper budget cannot already support. Concretely, for each candidate requirement, \sysname computes the cheapest requirements of other proposers that are compatible with it, and removes candidates that are never minimal for any counterpart. This pruning is effective: with up to 31 AWS replicas, it leaves only up to \cameraa{10} candidate budgets per proposer, which a branch-and-bound solver explores quickly~(\S\ref{sec:eval:strategies}).

\cameraz{The optimizer also accounts for suspected failures.
It excludes candidate requirements that depend on suspected processes.
A suspected proposer may nevertheless still be correct and propose, so
it receives a requirement guaranteed compatible with every requirement
the optimizer may select.}

\paragraph{Optimality}
\cameraz{The optimizer yields latency-optimal fast paths under fixed
and known network latencies, permanent crashes with perfect failure
detection, contention-free executions, and a fixed latency metric.}
A proof appears in the
\if\extended 1%
appendix.
\else%
\cameraz{extended version~\cite{extended}}.
\fi%
\cameraz{Briefly, for each proposer $p$ and budget $b$, the simulation
computes all evidence that $p$ can obtain by time $b$ under the given
topology.
Any deterministic fast path committing by $b$ can therefore rely only
on evidence contained in this snapshot.
If the requirements induced by a vector of budgets are incompatible,
committing within those budgets would violate the recoverability
conditions of \Cref{sec:constraints}, which every consensus protocol
must satisfy.
Conversely, every compatible vector can instantiate the template of
\Cref{sec:template}.
Hence, \cameraz{under these assumptions, no deterministic consensus fast path
can have lower latency according to the target metric}.}

\subsection{\cameraz{Robustness with multiple requirements}}
\label{sec:optimizer:robustness}

\cameraz{For simplicity, the presented \sysname template assigns a single requirement to each proposer, and
the proposer abandons the fast path upon suspecting a required witness
(Algorithm~\ref{alg:ac:fd}) \camerad{so} as not to wait indefinitely.
A proposer can instead maintain several requirements, commit when any is
met, and abandon only when all are ruled out.}
\camerac{We leave optimizing such sets for robustness to future work, but
observe that} \cameraa{Paxos-like requirements with quorums of $n{-}f$
processes}~(\S\ref{sec:constraints:examples}) \cameraa{are compatible with
every valid requirement. Adding one for every such quorum} therefore bounds
conflict-free latency by \cameraa{that of Paxos, without needing a timeout,}
at the cost of extra messages.

\subsection{\cameraz{Reducing network traffic}}
\label{sec:optimizer:network}

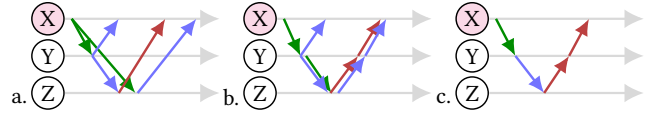
\begin{figure}
    \centering
    \resizebox{\columnwidth}{!}{%
\begin{tikzpicture}[
  x=1cm,y=1cm,
  proc/.style={
    circle,draw=black,fill=white,
    minimum size=5.8mm,inner sep=0pt,
    line width=.75pt,font=\Large
  },
  xproc/.style={proc,fill=magenta!18},
  timeline/.style={
    draw=black!15,line width=1.15pt,
    -{Latex[length=3.5mm,width=3.4mm]}
  },
  greenmsg/.style={
    draw=green!55!black,line width=1.25pt,
    -{Latex[length=3.0mm,width=2.65mm]}
  },
  bluemsg/.style={
    draw=blue!55,line width=1.25pt,
    -{Latex[length=3.0mm,width=2.65mm]}
  },
  redmsg/.style={
    draw=red!65!black!75,line width=1.25pt,
    -{Latex[length=3.0mm,width=2.65mm]}
  },
  panelletter/.style={font=\Large}
]

\def\RowX{0.67}
\def\RowY{0.00}
\def\RowZ{-0.67}
\def\TimeStart{0.31}
\def\TimeEnd{3.10}
\def\PanelStep{3.83}

\newcommand{\panelbase}[1]{%
  \begin{scope}[shift={(#1,0)}]
    \draw[timeline] (\TimeStart,\RowX) -- (\TimeEnd,\RowX);
    \draw[timeline] (\TimeStart,\RowY) -- (\TimeEnd,\RowY);
    \draw[timeline] (\TimeStart,\RowZ) -- (\TimeEnd,\RowZ);
    \node[xproc] at (0,\RowX) {X};
    \node[proc]  at (0,\RowY) {Y};
    \node[proc]  at (0,\RowZ) {Z};
  \end{scope}%
}

\newcommand{\msg}[6]{%
  \draw[#2]
    ($ (#1,0) + (#3,#4) $) --
    ($ (#1,0) + (#5,#6) $);
}

\panelbase{0}
\node[panelletter,anchor=east] at (-0.25,-0.75) {a.};

\msg{0}{greenmsg}{0.43}{\RowX}{0.80}{\RowY}
\msg{0}{greenmsg}{0.43}{\RowX}{1.58}{\RowZ}

\msg{0}{bluemsg}{0.80}{\RowY}{1.28}{\RowX}
\msg{0}{bluemsg}{0.80}{\RowY}{1.28}{\RowZ}
\msg{0}{redmsg}{1.28}{\RowZ}{2.12}{\RowX}
\msg{0}{bluemsg}{1.62}{\RowZ}{2.68}{\RowX}

\panelbase{\PanelStep}
\node[panelletter,anchor=east] at (\PanelStep-0.25,-0.75) {b.};

\msg{\PanelStep}{greenmsg}{0.43}{\RowX}{0.73}{\RowY}
\msg{\PanelStep}{greenmsg}{0.83}{\RowY}{1.29}{\RowZ}

\msg{\PanelStep}{bluemsg}{0.73}{\RowY}{1.19}{\RowX}
\msg{\PanelStep}{bluemsg}{0.73}{\RowY}{1.23}{\RowZ}

\msg{\PanelStep}{redmsg}{1.27}{\RowZ}{1.77}{\RowY}
\msg{\PanelStep}{redmsg}{1.77}{\RowY}{2.18}{\RowX}
\msg{\PanelStep}{bluemsg}{1.41}{\RowZ}{1.90}{\RowY}
\msg{\PanelStep}{bluemsg}{1.90}{\RowY}{2.28}{\RowX}

\panelbase{2*\PanelStep}
\node[panelletter,anchor=east] at (2*\PanelStep-0.25,-0.75) {c.};

\msg{2*\PanelStep}{greenmsg}{0.43}{\RowX}{0.78}{\RowY}
\msg{2*\PanelStep}{bluemsg}{0.78}{\RowY}{1.30}{\RowZ}
\msg{2*\PanelStep}{redmsg}{1.30}{\RowZ}{1.75}{\RowY}
\msg{2*\PanelStep}{redmsg}{1.75}{\RowY}{2.11}{\RowX}

\end{tikzpicture}%
}%
    \caption{\camerad{\t{Accept\&Spread} messages satisfying $X$'s requirement
    \t{\{X:\{X,Y,Z\},Y:\{X,Y\},Z:\{X,Y,Z\}\}}:
    (a)~as induced by the requirement, (b)~with latency shortcuts, and
    (c)~without redundant messages.
    Colors mark the newest evidence each message adds: $X$'s acceptance
    (green), $Y$'s or $Z$'s acceptance and record of $X$'s (blue), and
    $Z$'s record of $Y$'s (red).}}
    \label{fig:network}
    \Description{Three diagrams of processes X, Y, and Z exchanging messages to satisfy X's requirement: (a) all messages induced by the requirement, (b) messages rerouted through faster relays, and (c) the reduced set after removing redundant messages.}
\end{figure}

\cameraz{Algorithm~\ref{alg:ac:disseminate} broadcasts evidence whenever
it grows.
This propagates evidence as quickly as possible, but may disseminate
evidence that no selected requirement needs.
Once \sysname has chosen a requirement combination, it replaces these
broadcasts with a dissemination graph that preserves commit latency
while reducing communication.}

\cameraz{Consider the requirement in \Cref{fig:network}.
To satisfy it, $X$ must learn that $X$, $Y$, and $Z$ accepted the value,
that $Y$ witnessed $X$'s and $Y$'s acceptances, and that $Z$ witnessed
all three.
These dependencies directly induce a communication graph sufficient to
satisfy $X$'s requirement, shown in \Cref{fig:network}(a).}

\cameraz{\sysname then optimizes this graph.
First, it uses latency shortcuts:
measured wide-area latencies need not satisfy the triangle inequality,
so relaying evidence from $X$ to $Z$ through $Y$, for example, may be
faster than sending it directly.
\sysname uses such shortcuts when they reduce latency
(\Cref{fig:network}(b)).
If a shortcut relies on a process outside the requirement, the failure
detector also monitors that process so that its failure triggers
fast-path abandonment.}

\cameraz{Second, \sysname removes redundant messages.
For example, if $Z$ reports to $X$ after learning that $Y$ accepted,
an earlier message from $Z$ reporting only its own acceptance may be
unnecessary.
Likewise, if $X$ learns $Y$'s acceptance through $Z$, $Y$ need not also
report it directly to $X$.
\Cref{fig:network}(c) shows the resulting graph.}

\cameraz{All processes deterministically compute the same dissemination
graph from the requirements and latency table.
The graph specifies when and where each piece of evidence is sent and
also enables compact messages: rather than carrying an explicit set of
known acceptors, a message can identify the graph edge it traverses,
from which the receiver infers the corresponding evidence.}

\cameraz{\t{Accept\&Spread} messages can also carry only a short proposer
identifier rather than the full value.
Replicas buffer such evidence until they receive and log the
corresponding value, preventing a commit whose payload is unavailable.}

\cameraz{Overall, requirements determine \emph{what} evidence a commit
needs, while the dissemination graph determines \emph{how} that evidence
is delivered.
\camerac{Because the graph changes only the delivery---not the requirements,
when they can be met, or the commit checks---commit latency and safety are
unchanged.}}

\section{Implementation: \sysname-Based SMR} \label{sec:smr}

We implement \smr, a replication engine that uses \sysname-synthesized fast paths to build SMR logs~(\S\ref{sec:bg:fastpaths}).
Each voting process hosts a service replica \cameraz{and a client; other clients run as non-voting processes~(\S\ref{sec:ac:nonvoting})}. 
Our implementation consists of \cameraz{8,764} lines of Rust code and includes the \smr runtime, a CPU solver that finds optimal requirements~(\S\ref{sec:optimizer:solve}), and a strongly consistent geo-replicated key-value store \camerad{built on \smr}.

\paragraph{In a \camerad{nutshell}}
\cameraz{\smr runs one
consensus instance per log slot. Each instance's fast path is optimized for the current topology, proposer rates, and latency objective~(\S\ref{sec:optimizer}). A process proposes its command to the first slot $s$ it believes is empty.} For slot $s$, \smr instantiates \sysname's adopt-commit template with the current requirements:
a fast commit triggers immediate decision in the slot,
while \cameraz{a conflict or fault triggers adoption of a safe command, which is passed} to fallback consensus~(\S\ref{sec:bg:ac}).
If a process fails to commit in slot $s$, it re-proposes its command in slot $s{+}1$ and repeats until it commits.
Processes announce commits, allowing replicas to execute promptly and proposers to
\cameraz{track} empty slots.

\paragraph{Conflicts}
Upon conflicts, \smr falls back to Paxos~\cite{paxos}.
While Paxos takes two \rts to a majority, the first constitutes a mere leader-election phase~(\S\ref{sec:constraints:examples}):
\cameraz{\smr overlaps this phase with \sysname's adopt-commit using the same messages.
Thus, on a conflict, the highest-id concurrent proposer becomes leader and, after adopting a safe command, only needs to run Paxos's second phase.}
\cameraz{Adoption adds no separate round:
processes freeze as soon as they observe a conflict or finish their part of the dissemination graph~(\S\ref{sec:optimizer:network})\camera{.
They then keep forwarding the messages prescribed by the graph and piggyback their frozen state on the last one.}
The leader can therefore complete adoption by the time a fast commit would complete.
In practice, fast commits take one \rt in most deployments and only slightly more otherwise~(\S\ref{sec:eval:e2e}), so conflicts are about as slow as Paxos.}
Moreover, \cameraz{when} the leader knows \sysname's adopt-commit did not commit (Algorithm~\ref{alg:ac:propose}, line~\ref{alg:ac:propose:any}), \cameraz{it} batches all the commands it knows conflicted \cameraz{into one Paxos proposal}.

\paragraph{Sharding}
\smr shards requests
so that non-interfering commands can be replicated concurrently in independent logs without conflicting.
For instance, if replicating a key-value store, updates that target different keys never conflict.
Additionally, shards use different requirements to account for the distribution of proposers in each of them~(\S\ref{sec:optimizer:solve}).

\paragraph{Latency \camerad{monitoring}}
\smr is bootstrapped with knowledge requirements based on the initial network latency. Then, in the background, \smr processes monitor and share inter-process latencies in the form of a latency table. \cameraz{Upon a crash, broken link, or proposal-rate shift, a process proposes an updated table as a cross-shard command. Consensus installs the same table everywhere; each process then runs the deterministic solver locally to derive the same requirements.}

\paragraph{Delegated \camerad{execution}}
\cameraz{A non-voting client hosts no service replica, so it would still have to contact a replica to execute its command and reply, adding latency.} To avoid this extra step, \smr can delegate commitment to a nearby replica. \cameraz{The client} remains the logical proposer and starts the fast path, but the delegate replica performs the \t{can\_commit} check, commits once the required evidence is available, executes the command, and replies to the client. \cameraz{\smr then uses} requirements synthesized for proposer--delegate pairs.

\paragraph{Reads}
\cameraz{Like} existing SMR systems, \smr does not replicate \cameraz{reads}.
Instead, it \cameraz{executes} them at a single replica after checking
\camerac{that} \cameraz{its state} is up to date~\cite{pqr}.
\camerac{To do so, a read queries \camerac{a quorum that intersects every
write quorum}, thereby discovering any write that may have committed.
The read then waits until the highest such slot has been executed.}
\cameraz{Reads thus conflict with writes, but never delay them.}

\section{Evaluation} \label{sec:evaluation}

We evaluate the fast paths of consensus protocols,
comparing \sysname (using our \smr engine)
against prior baselines.

\paragraph{Baselines.}
As we focus on \cameraz{common-case replication latency}, we compare \sysname against \cameraz{protocols spanning} the main
approaches to reducing SMR latency: a stable-leader protocol, leaderless fast-quorum protocols, and
wide-area latency-optimized \cameraz{ones}. Our baselines are
Multi-Paxos~\cite{multipaxos}, EPaxos~\cite{epaxos},
SwiftPaxos~\cite{swiftpaxos}, and Pando~\cite{pando}.

\cameraz{All baselines are re-implemented in \smr's Rust codebase, with their fast-path latency validated against the specifications, so language and runtime do not confound our resource comparison~(\S\ref{sec:eval:resources}).}
\cameraz{Whenever a protocol} leaves choices such as leader placement or quorum selection, we use the configuration that minimizes average latency in our deployment.
For Pando, this means disabling erasure coding.
\cameraa{\camerac{Moreover,} EPaxos runs in non-thrifty mode.}
    
\paragraph{Testbed and configurations.}

We deploy all systems on AWS using \t{t3.medium} instances \cameraz{(Intel Xeon Platinum 8000-series), except the \camerab{conflict} and load experiments (\S\ref{sec:eval:conflict}--\ref{sec:eval:load}), which use non-burstable \t{c6i.large} instances (Intel Xeon Ice Lake). All run single-threaded}. All deployments have optimal resilience ($2f{+}1$ replicas), \cameraz{with one client per region}.

Our end-to-end experiments use four 7-replica deployments: three continental deployments in North America (NA), Europe (EU), and East Asia (EA), and one
deployment spanning the Northern Hemisphere (NH). Our scalability experiments vary the deployment size from 3 to 31 replicas across AWS regions worldwide.

For our experiments, we configure \sysname to optimize average end-to-end latency, matching the optimization objective used to configure the baselines.

\paragraph{Workloads.}
\cameraz{Unless otherwise specified, we focus on write requests, as all systems rely on the same
well-understood read algorithm~\cite{pqr}. Clients issue writes to
keys, with each key managed by an independent SMR shard. Our default workload is conflict-free:
the key space is partitioned per client, so concurrent client requests do not conflict.
This isolates the \camerab{conflict}-free fast path our evaluation targets.}
\camerab{We later study workloads with conflicts.}

Each write executes in $\approx$1\,\us, and both requests and
replies are 16\,B. All experiments use \cameraz{$100{,}000$} keys and\cameraz{, unless stated otherwise,} an aggregate
client throughput of $1{,}000$\,req/s. Clients generate requests at
random intervals following an exponential distribution, and requests are
distributed uniformly across clients.
Once all processes are ready, each \camerad{experiment} starts with a \cameraz{5\,s} warmup, after which clients measure the latency of each request during a 10\,s interval. Once a client finishes its measurements, it sustains the load on the system until all clients \cameraz{finish their measurements}.

\vspace{1em}

Our evaluation addresses the following questions:

\begin{myitemize}
    \item What is \sysname's end-to-end \camerac{latency} (\S\ref{sec:eval:e2e})?
    \item How do failures impact \sysname's latency (\S\ref{sec:eval:failures})?
    \item How does \sysname's latency scale with replicas (\S\ref{sec:eval:scalability})?
    \item What is \sysname's resource consumption (\S\ref{sec:eval:resources})?
    \item \cameraz{How does contention affect \sysname's latency (\S\ref{sec:eval:conflict})?}
    \item \cameraz{How much load can \sysname sustain (\S\ref{sec:eval:load})?}
    \item How fast are optimal requirements computed (\S\ref{sec:eval:strategies})?
\end{myitemize}

\subsection{End-to-\camerad{end latency}} \label{sec:eval:e2e}

We first measure the end-to-end latency of our \sysname-based key-value store and compare it with the baselines across the four 7-replica deployments.

\begin{figure}
    \centering
    \includegraphics[width=\columnwidth]{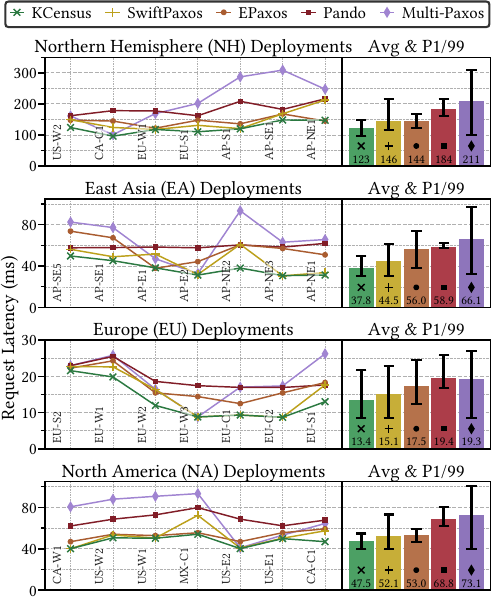}
    \caption{Latency across clients for the NH, EA, EU, and NA deployments. Bars indicate the average while the error bars indicate 1\textsuperscript{st} and 99\textsuperscript{th} \camerad{percentiles}.}
    \label{fig:eval:e2e}
    \Description{Bar chart of per-client and aggregate request latency, with 1st and 99th percentile error bars, for five protocols in the Northern Hemisphere, East Asia, Europe, and North America deployments. KCensus has the lowest average latency in all four deployments.}
\end{figure}

\newcommand{\num}[1]{%
    \fpeval{round(#1)}%
}
\newcommand{\ms}[1]{%
    \num{#1}\,ms%
}
\newcommand{\set}[2]{
    \xdef#1{\fpeval{#2}}%
}
\newcommand{\increase}[2]{%
    \set{\result}{((#2) - (#1))/(#1) * 100}%
    \num{\result}\%%
}
\newcommand{\decrease}[2]{%
    \set{\result}{((#1) - (#2))/(#1) * 100}%
    \num{\result}\%%
}

\newcommand{\ratiox}[2]{%
    \set{\result}{(#1)/(#2)}%
    \num{\result}$\times$%
}

\set{\knh}{123.40517806977687}
\set{\kea}{37.7864971578787}
\set{\keu}{13.355980808245349}
\set{\kna}{47.49875022131313}

\set{\snh}{146.42126781978658}
\set{\sea}{44.534699301153616}
\set{\seu}{15.130983158514491}
\set{\sna}{52.1137408662611}

\set{\enh}{144.39654018344316}
\set{\eea}{56.02687342195963}
\set{\eeu}{17.4741263519015}
\set{\ena}{52.99810890615982}

\set{\pnh}{183.99733549741876}
\set{\pea}{58.852808196498245}
\set{\peu}{19.42398131223755}
\set{\pna}{68.8425298962198}

\set{\mnh}{210.72543692108206}
\set{\mea}{66.07824246487178}
\set{\meu}{19.307854023498695}
\set{\mna}{73.06144780165874}

\set{\bnh}{144.39654018344316}
\set{\bea}{44.534699301153616}
\set{\beu}{15.130983158514491}
\set{\bna}{52.1137408662611}

\set{\keaw}{49.718419906205924}
\set{\beaw}{60.95218274068663}

\set{\keaper}{49.777565}
\set{\beaper}{61.113502}

\Cref{fig:eval:e2e} reports latency per client and aggregated. In all four deployments, \sysname achieves lower average latency than the
baselines.
Compared to the fastest baseline in each deployment, \sysname lowers
average latency by \decrease{\bnh}{\knh} in NH, \decrease{\bea}{\kea} in EA, \decrease{\beu}{\keu} in
EU, and \decrease{\bna}{\kna} in NA. \cameraz{The gain is largest in the NH and EA deployments, where} \cameraz{the \camera{fastest baselines} pay a larger latency penalty for reaching beyond the nearest majority.}

\set{\kbestsinglegain}{38.06773370548862}
\set{\bbestsinglegain}{60.426329602681726}

\sysname also has the best 99\textsuperscript{th} percentile in all deployments. In EA, the worst-located client of the \cameraz{best baseline} has \increase{\keaw}{\beaw} higher latency than \sysname's worst-located client, leading to a \increase{\keaper}{\beaper} higher 99\textsuperscript{th} percentile.
\cameraz{More generally, \sysname matches or beats the best baseline at every client location, reducing latency by up to \decrease{\bbestsinglegain}{\kbestsinglegain}.}

\smallskip

\cameraz{Compatibility often adds no latency beyond reaching a client's
closest majority.
With $f{=}3$, two strategies need four \camerad{intersection witnesses}.
If either proposer belongs to the intersection, its quorum supplies all
four witnesses.
Otherwise, \camerad{if the intersection has size $i$, its replicas together with the two proposers already provide $i{+}2$} witnesses.
Thus, a two-replica intersection suffices, while a single-replica intersection needs only one additional witness.}

\cameraz{This makes the cheapest valid strategy sufficient for every
client in EU and NA: each pays only the latency to its closest majority.
In EU, every pair of closest-majority quorums intersects in at least two
replicas.
\camera{NA proposers use three closest-majority quorums: one shared by the
eastern proposers, one by the western ones, and MX-C1's own. Only the
eastern proposers' and MX-C1's quorums intersect in a single replica,
US-E2.}
While waiting for \camera{their quorum's farthest member,} CA-W1, eastern proposers spread evidence of US-E2's
acceptance across the east, supplying the fourth witness without
extending the fast path.}

\cameraz{EA and NH have more single-replica intersections, so the fourth witness is not always free.
\camera{Still, 5 of 7 clients in each use the cheapest strategy; the others
add less than 0.3\,ms to average latency in EA, and about 2\,ms in NH,
which lacks a hub shared by most closest-majority quorums.}}

\cameraz{Across all four deployments, only one client---AP-NE1
in NH---needs a fast-path quorum of 5 replicas rather than a bare
majority.
Every other client completes before the latency of reaching its fifth
closest replica.}

\smallskip

\cameraz{Overall, \sysname \camerac{significantly} lowers both average and tail latency relative
to the baselines.}

\subsection{Impact of \camerad{failures on latency}} \label{sec:eval:failures}

We next study whether \camerac{the evaluated protocols} preserve low end-to-end latency as replicas fail.
\cameraz{We repeat \Cref{sec:eval:e2e} for every placement of 0--3 replica failures. \camera{Clients co-located with failed replicas keep issuing requests, as non-voting proposers in \sysname and SwiftPaxos, and otherwise through the replica offering them the lowest execution latency.}}
\camerac{We focus on NH, as it is the most representative deployment for global services, and report its latency averaged across clients and fault placements.}

\begin{figure}
    \centering
    \includegraphics[width=\columnwidth]{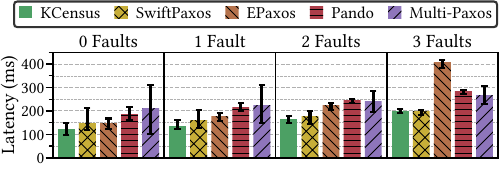}
\caption{\camerac{Average request latency in NH deployments under 0--3 replica failures.
Error bars span the best- and worst-located clients' averages across fault placements.}}
    \label{fig:eval:failures}
    \Description{Bar chart of average request latency in the Northern Hemisphere deployment for five protocols with zero to three failed replicas. KCensus has the lowest latency with up to two failures; with three failures, SwiftPaxos is slightly faster.}
\end{figure}

\set{\knhzero}{123.40517806977687}
\set{\knhone}{136.68000508358384}
\set{\knhtwo}{164.96023976227943}
\set{\knhthree}{201.72918758785903}

\set{\snhzero}{146.42126781978658}
\set{\snhone}{156.99515939884566}
\set{\snhtwo}{174.80448624307695}
\set{\snhthree}{196.40513624800303}

\set{\enhzero}{144.39654018344316}
\set{\enhone}{176.40194788709448}
\set{\enhtwo}{221.04900439537923}
\set{\enhthree}{405.25623172038286}

\set{\pnhzero}{183.99733549741876}
\set{\pnhone}{214.219989647838}
\set{\pnhtwo}{243.20013161126596}
\set{\pnhthree}{282.38708492471653}

\set{\mnhzero}{210.72543692108206}
\set{\mnhone}{223.2379682948885}
\set{\mnhtwo}{238.85929776172253}
\set{\mnhthree}{267.8782818443487}

\set{\bnhzero}{144.39654018344316}
\set{\bnhone}{156.99515939884566}
\set{\bnhtwo}{174.80448624307695}

\set{\bnhthreeNoSwift}{267.8782818443487}

\set{\knhbestthree}{190.92324211106194}
\set{\snhbestthree}{182.74381896897788}
\set{\knhworstthree}{209.02856424394085}
\set{\snhworstthree}{204.64511100858275}

\Cref{fig:eval:failures} shows that \sysname maintains low latency \cameraz{as it re-synthesizes fast paths to adapt to failures.} Compared to the fastest baseline at each fault count,
\sysname lowers average latency by \decrease{\bnhzero}{\knhzero} with no
failures, \decrease{\bnhone}{\knhone} with one failure, and
\decrease{\bnhtwo}{\knhtwo} with two failures.

At three failures, the surviving configurations force all proposers onto the same quorum---the structure used by SwiftPaxos---and SwiftPaxos edges ahead of \sysname by \increase{\snhthree}{\knhthree}.
\cameraz{SwiftPaxos owes this margin to speculative execution, which lets a proposer commit in a single round trip even when it is co-located with a failed replica. Because SwiftPaxos relies on a leader to order commands, it rarely has to roll back a speculative execution, \cameraz{making speculation cheap}. \sysname does not speculate, but comes close via delegated execution~(\S\ref{sec:smr}).}

\subsection{Scalability of \camerad{latency with replicas}}
\label{sec:eval:scalability}

We now study how latency changes as the number of replicas grows.
We re-run the experiment in \Cref{sec:eval:e2e} on deployments ranging
from 3 to 31 servers.
We consider two deployment series. The \emph{Random} series
starts from three random AWS regions and adds two random regions at each
step. The \emph{Parisian} series starts from Paris
and repeatedly adds the two remaining regions with the lowest
latency to Paris.

\begin{figure}
    \centering
    \includegraphics[width=\columnwidth]{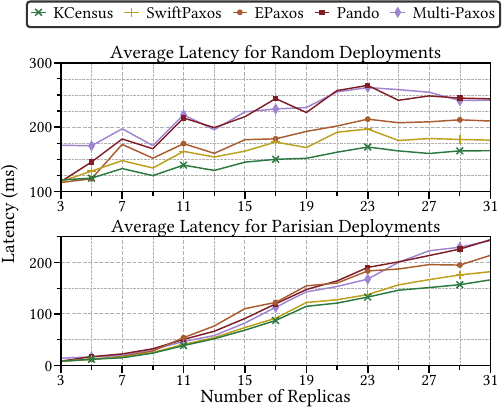}
    \caption{Average request latency for Random and Parisian deployments
    as the number of replicas grows from 3 to 31.}
    \label{fig:eval:scalability}
    \Description{Two line charts of average request latency versus the number of replicas, from 3 to 31, for Random and Parisian deployments. With three replicas the fast-path protocols have similar latency; with more replicas, KCensus has the lowest latency.}
\end{figure}

\set{\krs}{117.93325352489667}
\set{\kre}{163.38571309088135}

\set{\srs}{115.48821751837349}
\set{\sre}{179.87019468030743}

\set{\ers}{114.0920723580716}
\set{\ere}{209.53452527969387}

\set{\prs}{115.24585802609315}
\set{\pre}{243.9664635262159}

\set{\mrs}{171.93606497318848}
\set{\mre}{241.97009917493708}

\newcommand{\rmd}{\camera{21}\xspace}
\set{\krmd}{161.0747981707439}
\set{\brmd}{192.1465463236146}

\set{\krsv}{135.6523165208939}
\set{\brsv}{148.02982733678908}

\set{\brs}{114.0920723580716}
\set{\bre}{179.87019468030743}

\set{\kps}{8.24852924335358}
\set{\kpe}{166.1393274931313}

\set{\sps}{7.951105014115308}
\set{\spe}{181.9647736652724}

\set{\eps}{7.96198807397343}
\set{\epe}{213.67472378318718}

\set{\pps}{8.286347170925723}
\set{\ppe}{243.94267675626656}

\set{\mps}{14.031526771378369}
\set{\mpe}{242.3382292325237}

\set{\bps}{7.951105014115308}
\set{\bpe}{181.9647736652724}

\Cref{fig:eval:scalability} shows that \sysname consistently remains the lowest-latency configuration as the deployment grows. With
three replicas, all \camera{fast-path} protocols have limited placement choices and achieve
similar latency. As more replicas are added, the protocols diverge:
\sysname can exploit the additional placement choices, whereas the
baselines remain constrained by their leader, quorum, or overlap
structure.

In Random deployments, latency is not monotonic. Adding regions can both introduce more distant clients and create new quorum choices, so the
average depends on the resulting topology. 
At 3 replicas, \sysname performs similarly to EPaxos, Pando\camerad{,} and SwiftPaxos as they are all optimal. The same applies to EPaxos at 5 replicas. With more replicas, \sysname dominates\camera{, reducing latency over the
fastest baseline by \decrease{\brsv}{\krsv} (at 7 replicas) to
\decrease{\brmd}{\krmd} (at \rmd)}.

In Parisian deployments, latency grows more smoothly because replicas are
added in \camerad{order of} increasing distance from Paris.

\camera{Overall, latency depends on where replicas are rather than on how
many there are, and \sysname synthesizes its requirements to make the most of
each placement.}

\subsection{Resource \camerad{consumption}} \label{sec:eval:resources}

We now quantify the cost of \sysname using the scalability experiments of \Cref{sec:eval:scalability}.
We measure \cameraz{per-replica network traffic and messages}, \cameraz{CPU use}, and memory consumption.

\paragraph{Bandwidth \camerad{consumption}.}
\sysname's communication pattern is synthesized from its knowledge requirements, so its network cost does not follow a fixed quorum pattern. \Cref{fig:eval:resources:bandwidth} reports \cameraz{average bytes and messages per application request per replica} in Random deployments.

\begin{figure}
    \centering
    \includegraphics[width=\columnwidth]{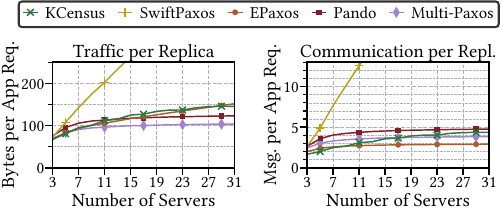}
    \caption{\cameraz{Average bytes (left) and messages (right) per application request per replica} in Random deployments.}
    \label{fig:eval:resources:bandwidth}
    \Description{Two line charts of average bytes and messages per request per replica versus the number of replicas in Random deployments. SwiftPaxos grows fastest in both; KCensus stays close to the other protocols.}
\end{figure}

\newcommand{\mb}[1]{%
    \num{(#1) / 1024 / 1024}\,MiB%
}
\newcommand{\kb}[1]{%
    \num{(#1) / 1024}\,KiB%
}
\newcommand{\bytes}[1]{%
    \num{#1}\,B%
}


\cameraz{\sysname's traffic remains low as one identifier encodes knowledge~(\S\ref{sec:optimizer:network}). At 31 replicas, it sends 146\,B per request per replica, 40\% more than Multi-Paxos but 4\% less than EPaxos. Every SwiftPaxos replica acknowledges every other, making its aggregate traffic quadratic and its per-replica traffic linear; it reaches 491\,B, 3.4$\times$ \sysname's.}

\cameraz{\sysname sends the fewest messages through 9 replicas. As its requirements add witnesses and evidence routes, \sysname reaches 4.4 per request per replica at 31 replicas---50\% more than EPaxos but fewer than Pando. SwiftPaxos sends 34.1 (7.8$\times$ as many), despite grouping acknowledgements.}

\paragraph{Compute \camerad{consumption}.}
The left side of \Cref{fig:eval:resources} reports the \cameraz{average CPU use of a server, where 100\% means one core busy for the whole run}, including user and kernel time. \camera{Message handling dominates, so each system's CPU use tracks its communication cost. \sysname's} \cameraz{is the lowest of all systems at small scales and grows with the amount of evidence disseminated. \sysname uses 9.0\% of a core at 3 replicas and 14.4\% at 31, where it is 19\% above EPaxos and 38\% above Multi-Paxos. SwiftPaxos instead grows linearly with the number of replicas, as each server acknowledges every other, reaching 55.9\% at 31 replicas, 3.9$\times$ \sysname's use.}

\begin{figure}
    \centering
    \includegraphics[width=\columnwidth]{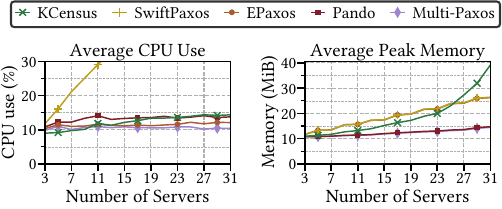}
    \caption{\cameraz{Average CPU use per server} (left) and \cameraz{average peak memory per server} (right) for Random deployments.}
    \label{fig:eval:resources}
    \Description{Two line charts of average CPU use and average peak memory per server versus the number of replicas in Random deployments. SwiftPaxos has the highest CPU use at scale; KCensus has the highest memory use at 31 replicas, below 40 MiB.}
\end{figure}

\paragraph{Memory \camerad{consumption}.}
The right side of \Cref{fig:eval:resources} reports \cameraz{the average peak memory of a server}. \sysname uses \cameraz{memory} to store knowledge and synthesized dissemination metadata\cameraz{, so its footprint grows with the number of replicas: from 10.9\,MiB at 3 replicas to 39.2\,MiB at 31, the largest of all systems at that scale. Multi-Paxos and Pando, which only retain the value a replica accepted, remain below 15\,MiB with 31 replicas, while EPaxos and SwiftPaxos, which track a dependency set per command, sit in between at 26\,MiB.}

\vspace{1em}

\cameraz{Overall, \sysname's resource consumption remains in the range of existing fast-path protocols: at 31 replicas, it sends slightly less data per replica than EPaxos, and it is several times cheaper than SwiftPaxos on networking and CPU. What \sysname trades for lower latency is a modest message and CPU overhead over leader-based protocols, and the largest memory footprint of all systems at 31 replicas, which remains small in absolute terms---under 40\,MiB per server.}

\subsection{\cameraz{Latency under \camerad{contention}}}
\label{sec:eval:conflict}

\cameraz{We measure \camerab{the effect of conflicts} in the NH deployment
\camerab{by having clients access 100,000 shared keys.}
Clients issue 50\% reads and 50\% writes at an
aggregate $1{,}000$\,req/s, drawing keys either uniformly or from a
Zipfian distribution with skew $0.99$, matching YCSB's workload
A~\cite{ycsb}. At this load, uniform access produces negligible
conflicts and serves as a reference; Zipfian access concentrates
concurrent requests on hot keys,
\camerab{creating many conflicts across clients.}}

\begin{figure}
    \centering
    \includegraphics[width=\columnwidth]{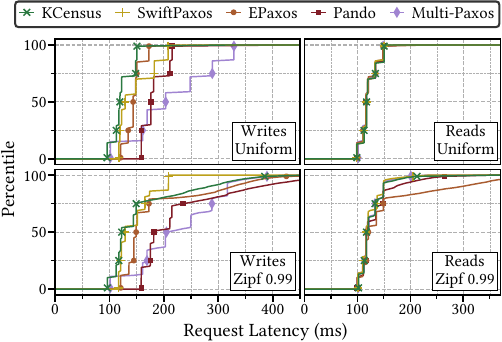}
    \caption{\cameraz{Latency CDFs of writes and reads in NH with a shared
    key space, under uniform and Zipfian access.}}
    \label{fig:eval:conflict}
    \Description{Cumulative distribution plots of write and read latency in the Northern Hemisphere deployment for five protocols, under uniform and Zipfian key access. Under Zipfian access, KCensus writes show a fast-path mode and a slower tail from conflicts.}
\end{figure}

\cameraz{\Cref{fig:eval:conflict} reports
writes and reads separately, as reads are not replicated~(\S\ref{sec:smr}).
Uniform access \camerab{resembles the conflict-free}
result of \S\ref{sec:eval:e2e}: \sysname's median write latency is
\ms{119.24}, and its CDF \camerab{is} concentrated around the fast-path
latency.}

\cameraz{Under Zipfian access, \sysname's median \camerab{is}
\ms{122.55}, within \camerad{\increase{119.24}{122.55} of the
uniform median}, but \camerab{the} 99\textsuperscript{th} percentile rises
to \ms{385.37}. The CDF exposes two regimes: \camerab{conflict-free} writes
retain the fast-path latency, while conflicting writes adopt a safe
command and complete \camerab{using} Paxos's second phase~(\S\ref{sec:smr}).
Averaged over writes, \sysname reaches \ms{158.08}; SwiftPaxos is faster
at \ms{145.55}, with a shorter tail of \ms{208.54}, because its
leader-based fallback runs alongside the fast path and can order
conflicting commands without recovery. Nevertheless,
\sysname commits half of its writes faster than any baseline and remains
ahead of EPaxos~(\ms{182.86}), Multi-Paxos, and Pando
(both above \ms{220}).}

\cameraz{Reads \camerab{resemble} uniform access, since we give
every baseline \smr's read optimization~(\S\ref{sec:smr}).
Under Zipfian access, \sysname averages \ms{126.77}, with a
99\textsuperscript{th} percentile of \ms{212.61}.
EPaxos degrades most, to \ms{156.95} and \ms{432.79}, because its reads
wait for every ongoing command at a majority rather than only the
highest ongoing slot. \cameraz{SwiftPaxos avoids this:
if the leader replies in the read quorum, the reader waits
only for writes the leader accepted.}}

\subsection{\cameraz{Latency under \camerad{load}}}
\label{sec:eval:load}

\cameraz{We measure how latency changes with load and how much
throughput each system sustains before saturation.}
\cameraz{We reuse the workload of \S\ref{sec:eval:conflict},}
\cameraz{increasing the request rate until each system saturates.
All systems batch
\camerab{requests that queue}
up: \sysname, Pando, and Multi-Paxos group a
key's pending commands into its next slot, while EPaxos and SwiftPaxos batch as
their reference implementations do, grouping commands into one instance and
acknowledgements \cameraa{for different keys} into one message, respectively. The one exception is that
EPaxos never groups commands across keys\cameraa{, because doing so} adds conflicts and does not
improve its throughput in our experiments.}

\begin{figure}
    \centering
    \includegraphics[width=\columnwidth]{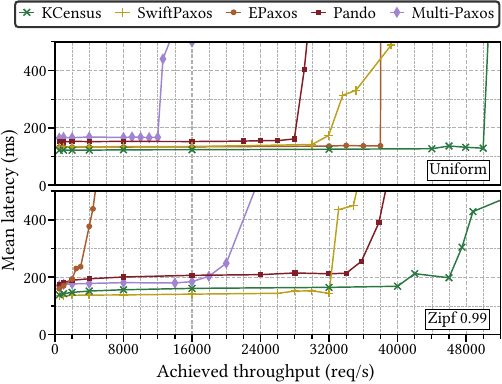}
    \caption{\cameraz{Mean request latency versus achieved throughput in
    NH under uniform and Zipfian access.} \cameraz{We omit runs beyond saturation, where more load reduced throughput.}}
    \label{fig:eval:load}
    \Description{Two line charts of mean request latency versus achieved throughput in the Northern Hemisphere deployment, under uniform and Zipfian access, for five protocols. KCensus sustains the highest throughput in both.}
\end{figure}

\cameraz{\Cref{fig:eval:load} plots mean latency against achieved throughput. Under uniform access, \sysname remains near \ms{124} through
44k\,req/s and below \ms{137} up to 50k, where it saturates, giving it the lowest
latency at every measured load and the highest saturation throughput.}
\cameraz{EPaxos's latency rises sharply beyond 38k\,req/s, SwiftPaxos's beyond
32k, Pando's beyond 28k, and Multi-Paxos's beyond 12k.}
\cameraz{\sysname, EPaxos, and SwiftPaxos saturate when their busiest replica exhausts its core; Multi-Paxos saturates earlier, bottlenecked by its leader.}

\cameraz{Under Zipfian access, \camerab{conflicts} raise \sysname's latency, but it
stays below \ms{170} up to 40k\,req/s and then degrades gradually
\camerab{as we batch more commands on hot keys.}}
\cameraz{EPaxos degrades earliest, exceeding \ms{370} by 4k\,req/s as \camerab{conflicts} grow its dependency
sets.}
\cameraz{SwiftPaxos \camerab{has} lower latency than \sysname, as in \S\ref{sec:eval:conflict}, until it rises sharply beyond 32k\,req/s.}
\cameraz{Multi-Paxos and Pando \camerad{rise sharply beyond 18k and 34k\,req/s}, and remain slower than \sysname throughout.}

\cameraz{Overall, \sysname achieves the highest throughput at comparable
latency under both distributions; \camerab{conflicts} raise its latency but does not cause
a throughput collapse.}

\subsection{Time to \camerad{optimize requirements}} \label{sec:eval:strategies}

Whenever \sysname processes detect a significant change in network topology or latency,
they recompute optimal knowledge requirements.
Due to the combinatorial complexity of the problem, one could fear
that this cost would prevent \sysname from quickly adapting to network fluctuations.
We evaluate this cost by running the same experiments as in \Cref{sec:eval:scalability}, and measuring how long requirements optimization takes.
\Cref{fig:eval:strategies} reports the results.

\begin{figure}
    \centering
    \includegraphics[width=\columnwidth]{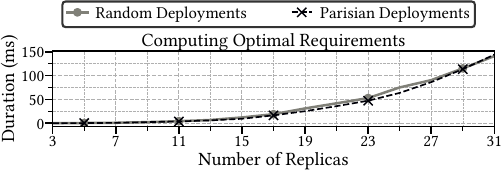}
    \caption{Time to optimize requirements on a single core.}
    \label{fig:eval:strategies}
    \Description{Line chart of the time to optimize requirements on a single core versus the number of replicas, for Random and Parisian deployments, growing from about 0.1 ms at 3 replicas to about 144 ms at 31 replicas.}
\end{figure}

Optimization time
is similar for both kinds of deployments, ranging from \cameraz{0.1\,ms for 3 replicas to 144\,ms} for 31 replicas.
This cost is orders of magnitude lower than the time interval between significant changes in topology, so it is negligible.

\section{Related Work} \label{sec:related}

\sysname relates to work on strongly consistent replication, low-latency consensus, topology-aware communication, and knowledge-based distributed protocols.

\paragraph{Strongly \camerad{consistent wide-area replication}.}
\cameraz{Wide-area systems face a fundamental latency--consistency tension.
Many systems reduce latency by weakening consistency~\cite{bayou, dynamo,
cassandra, pnuts, cops, gentlerain, pileus, gemini, icg, weak-1, weak-2,
weak-3, weak-4, analysis-strong-consistency, existential-consistency,
eiger}.
Strongly consistent services instead tend to rely on consensus to order updates
across regions~\cite{zookeeper, chubby, etcd, spanner, megastore, was,
cassandra-guarantees}.
Rather than relaxing consistency, \sysname reduces the latency of this
consensus layer by synthesizing low-latency fast paths.}

\paragraph{Low-\camerad{latency consensus}.}
A large body of work reduces consensus latency by fixing a particular communication pattern, quorum rule, or coordinator placement~\cite{multipaxos, ring-paxos, mencius, wpaxos, fast-paxos, epaxos, flexible-paxos, janus, atlas, mercury, wheat, aware}. Each of these choices \camerad{is} one point in the fast-path design space: they work well in some cases, but no single prior scheme is
best for
all topologies, workloads, and latency objectives.
\sysname takes a complementary approach: it characterizes which fast paths are recoverable, and synthesizes \camerad{an} optimal ensemble for a given deployment.

Prior work considers fast paths in the context of quorum systems~\cite{refinedquorums}. In contrast, \sysname captures a broader design space than quorum choice alone.
\cameraz{Existing work studies fast-quorum bounds~\cite{lamport-lower-bounds, generic-broadcast, two-step-bounds}; such bounds do not apply to \sysname, where each proposer is assigned a specific strategy.}

Several protocols reduce how often proposals conflict. Some exploit commutativity~\cite{lamport2005generalized, epaxos, gryff, janus}, while others use sequencing or timing assumptions to steer concurrent commands apart~\cite{corfu, epaxos-revisited}. These techniques are orthogonal to \sysname: they reduce the probability of conflicting fast paths, whereas \sysname \camerad{optimizes fast-path latency}.

Other low-latency systems build on assumptions outside \sysname's model. Some target Byzantine settings~\cite{zyzzyva, bitcoin-ng, arbitrum, l2-1, l2-2}, \cameraz{including Heterogeneous Paxos~\cite{hetero-paxos}, which configures whom each learner trusts and how many failures it tolerates}; others exploit datacenter hardware such as RDMA or programmable switches~\cite{dare, aguilera2020microsecond, derecho, ukharon, ubft, no-paxos, pre-p4xos, p4xos}. \sysname targets crash-fault consensus over ordinary \camerad{WANs}.

\paragraph{Topology-\camerad{aware communication}.}
Prior work showed that communication structure matters for consensus performance~\cite{pigpaxos, ring-paxos, multi-ring-paxos, wheat, tree-byzcoin, tree-kauri, tree-motor, tree-omniledger, tree-canopus}. Ring- and tree-based protocols distribute load and improve throughput~\cite{ring-paxos, multi-ring-paxos, lcr, totem, tree-byzcoin, tree-kauri, tree-motor, tree-omniledger, tree-canopus}, while Pig Paxos~\cite{pigpaxos} uses aggregation and piggybacking to reduce bandwidth bottlenecks. \sysname's optimized dissemination graphs can resemble such structures, but are derived from a different objective: they deliver exactly the knowledge required by the synthesized fast paths along low-latency routes.

\paragraph{Knowledge-\camerad{based protocols}.}
Reasoning about knowledge \cite{knowledge, knowledge-and-common-knowledge, kbp} 
is a classical way to characterize the information needed to solve distributed tasks~\cite{communication-efficient-canonical}. This line of work has led to optimal synchronous consensus protocols under crash, Byzantine, and omission failures~\cite{castaneda2022unbeatable, early-decision-predicate, knowledge-and-common-knowledge-byz, common-knowledge-and-consistent-coordination, optimal-eba-of}. \sysname shares similarities, but works in a different model: asynchronous crash-fault fast paths with non-uniform latencies. This vastly changes both the problem and the solution.

\section{Conclusion}
\label{sec:conclusion}

We presented \sysname, a framework \cameraz{that synthesizes consensus fast paths}. By \cameraz{characterizing the evidence needed for safe recovery}, \sysname turns fast-path design into a checkable optimization problem. The resulting \cameraz{fast paths} adapt to network topology, proposer distribution, and latency objective.
In our evaluation, this achieves the lowest latency across deployments, cutting average latency by up to \cameraz{16\%}.

Future work includes mechanically checking \sysname and extending it to support EPaxos-style dependency tracking, Byzantine faults, \cameraz{speculative execution}, and erasure coding.

\section*{Acknowledgments} \label{sec:acknowledgments}
\cameraz{We thank our shepherd and anonymous reviewers for their thoughtful
feedback, and our anonymous artifact evaluators for assessing our prototype.
We are also grateful to Zakaria Choukri for his work on the infrastructure supporting our AWS experiments.}
\cameraa{This work was supported in part by the Swiss National Science Foundation (SNSF) under grant 200021\_215383.}

\newpage

\bibliographystyle{ACM-Reference-Format}
\bibliography{references}

\if\extended 1
 \appendix
 \onecolumn 
 \section{\sysname Correctness} \label{sec:proofs}

Recall the properties of the adopt-commit variant which we consider (presented in \Cref{sec:bg:ac}):
\begin{description}[style=unboxed,leftmargin=0.7cm]
 \item [Termination:] \t{Propose} adopts or commits if called at a correct process.
 \item [Validity:] \t{Propose} adopts or commits only proposed \camera{values}.
 \item [Agreement:] if \t{Propose} commits $\camera{v}$ at some process, then \t{Propose} never adopts or commits \camera{\t{v'}}$\neq$\camera{\t{v}} at other processes.
\end{description}
We show in \Cref{section:adopt-commit correctness} that these properties are satisfied by our adopt-commit protocol presented in \Cref{sec:ac}.

A consensus protocol should satisfy the following properties:
\begin{description}[style=unboxed,leftmargin=0.7cm]
 \item [Termination:] \t{Propose} decides if called at a correct process.
 \item [Validity:] \t{Propose} decides only proposed \camera{values}.
 \item [Agreement:] a single \camera{value} may be decided by \t{Propose} calls.
\end{description}

As mentioned in \Cref{sec:bg:ac}, consensus may be solved using the adopt-commit abstraction in the fast path.
Algorithm~\ref{alg:ac-to-consensus} formalizes the construction sketched in \Cref{sec:bg:ac}, showing how
adopt-commit protocols can be used to implement consensus.
In \Cref{proof:consensus-correctness} we prove that the protocol in Algorithm~\ref{alg:ac-to-consensus} satisfies the three required consensus properties.

\StartLineAt{1}

\begin{lstlisting}[caption={Consensus Using Adopt-Commit},label={alg:ac-to-consensus}]
def Consensus::Propose(@\camera{\texttt{v}}@):
  call AdoptCommit::Propose(@\camera{\texttt{v}}@)@\label{alg:consensus:propose:ac propose}@
upon event AdoptCommit::Commit(@\camera{\texttt{v}}@):
  trigger Consensus::Decide(@\camera{\texttt{v}}@)@\label{alg:consensus:propose:decide upon commit}@
upon event AdoptCommit::Adopt(@\camera{\texttt{v}}@):
  call FallbackConsensus::Propose(@\camera{\texttt{v}}@)@\label{alg:consensus:propose:fallback propose}@
upon event FallbackConsensus::Decide(@\camera{\texttt{v}}@):
  trigger Consensus::Decide(@\camera{\texttt{v}}@)@\label{alg:consensus:propose:decide upon fallback decide}@
\end{lstlisting}

\subsection{\sysname's Adopt-Commit Correctness}\label{section:adopt-commit correctness}
In this section we prove the correctness of our adopt-commit protocol presented in \Cref{sec:ac}, which serves as the core building block of \sysname.
Specifically, we prove that our adopt-commit protocol satisfies the properties presented in the beginning of this section.
We first clarify some assumptions.
We assume a network model with perfect links. This implies that any message sent by a correct process to a correct process is guaranteed to be eventually received at the target.
We further assume that any sending of messages executed throughout the protocol is done in the background and the calling block immediately continues to the next step. 
Additionally, we assume that any block within a method call or a message-handling routine is executed atomically, unless it contains wait steps. That is, such a block is not interrupted by input events (i.e., a call to \t{Propose} or incoming messages).
In particular, this implies that \t{can\_commit}, which contains no wait steps, executes without any incoming messages being processed at the process during its execution. As a result, the value of \t{known\_acceptors} remains unchanged throughout.

We next establish several supporting claims to facilitate the correctness proof of our adopt-commit protocol.
We start with some statements concerning the behavior of the gossip dissemination process.
The subsequent observation ensures that \t{Accept\&Spread} messages disseminate proposed \camera{values}.
\begin{observation}\label{observation:AcceptAndSpread disseminate proposed values}
Only proposed \camera{values} are being sent in \t{Accept\&Spread} messages.
\end{observation}
\begin{proof}
\t{Accept\&Spread} messages are sent \camerad{on} two occasions. They are initiated on \cref{alg:ac:propose:disseminate} and propagated
on \cref{alg:ac:broadcast}. We prove the lemma by induction.
For the base case, an \t{Accept\&Spread} message is initiated on \cref{alg:ac:propose:disseminate} with \camera{value} $\camera{v}$ during a \t{Propose($\camera{v}$)} call. Hence, such an \t{Accept\&Spread} message includes a proposed \camera{value}.
Now suppose the lemma holds for all \t{Accept\&Spread} messages sent so far, and consider an \t{Accept\&Spread} message sent on \cref{alg:ac:broadcast}. It is sent with a \camera{value} from a received \t{Accept\&Spread} message, which is a proposed \camera{value} by the inductive hypothesis.
\end{proof}

The following observation and lemma ensure the dissemination of \t{Accept\&Spread} messages when no process is frozen by some adopter.
\begin{observation}\label{observation:bcast after accept}
Upon receiving an \t{Accept\&Spread} message for the first time, a correct process sends an \t{Accept\&Spread} message to all other processes (on \cref{alg:ac:broadcast}) containing the \camera{value} from the received message, unless the process has previously received a \t{Freeze} message.
\end{observation}
\begin{proof}
Before receiving the first \t{Accept\&Spread} message, since the handler of such messages is the only place where \t{known\_accep\-tors} is modified, its values are still $\emptyset$. Let $\camera{v}$ be the \camera{value} in the first received \t{Accept\&Spread} message at a correct process $p$ (in case $p$ ever receives such a message). Then upon its receipt, if $p$ has not previously received a \t{Freeze} message, $p$ modifies \t{known\_acceptors[$p$]} on \cref{alg:ac:knownacceptorsme1} from $\emptyset$ to a value that includes $p$. As the value of \t{known\_acceptors[$p$]} grew, $p$ next sends to others an \t{Accept\&Spread} message containing $\camera{v}$ on \cref{alg:ac:broadcast}.
\end{proof}

\begin{lemma}\label{lemma:if correct proposes everyone receives AcceptAndSpread}
If some correct process calls \t{Propose} and does not receive any \t{Freeze} message, then every correct process eventually receives an \t{Accept\&Spread} message.  
\end{lemma}
\begin{proof}
Let $p$ be a correct process that calls \t{Propose}.
During its \t{Propose} call, $p$ sends itself an \t{Accept\&Spread} message on \cref{alg:ac:propose:disseminate}. If $p$ did not receive any \t{Freeze} message, then by \Cref{observation:bcast after accept}, $p$ sends to others an \t{Accept\&Spread} message (on \cref{alg:ac:broadcast}) following the receipt of this message if it has not done so earlier.
By our network model, which assumes perfect links, every correct process eventually receives an \t{Accept\&Spread} message sent by $p$.    
\end{proof}

The following observations pertain to accepted \camera{values}. They establish that each process accepts at most one \camera{value}, and that \camerad{every process appearing in the values of the \t{known\_acceptors} array at any process has accepted} the same \camera{value} as that process.

\begin{observation}\label{observation:accepted assinged once}
The value of \t{accepted} is assigned only non-$\bot$ values, and it is modified at most once on each process from the initial $\bot$ value to another value.
\end{observation}
\begin{proof}
The only assignment to \t{accepted} (other than its initialization) is done on \cref{alg:ac:assignaccepted}. On the first receipt of an \t{Accept\&Spread} message, the \camera{value} in this message, which is a proposed \camera{value} and not $\bot$ by \Cref{observation:AcceptAndSpread disseminate proposed values}, is assigned to \t{accepted}. From this point on, the value of \t{accepted} is not $\bot$ so it remains the same.
\end{proof}

\begin{observation}\label{observation:processes in known_acceptors accepted the same value}
All processes appearing in the values of the \t{known\_acceptors} array at process $p$, accepted the same \camera{value} as $p$ did, i.e., assigned it to their \t{accepted} local variable.
\end{observation}
\begin{proof}
We proceed by induction. The base case vacuously holds since the \t{known\_acceptors} values are initially empty (\cref{alg:ac:knownacceptors}).
Now suppose the observation holds for all \t{known\_acceptors} values at a certain point in the execution, and consider the first following update of \t{known\_acceptors} at some process $p$ on \cref{alg:ac:knownacceptorsme1} or \cref{alg:ac:knownacceptorsq} upon handling an \t{Accept\&Spread}($\camera{v}$, \t{knowledge\_of\_q}) message from some process $q$.
$p$ reaches this line after having assigned a \camera{value} to \t{accepted} on \cref{alg:ac:assignaccepted}.
Since $p$ skipped the \t{if} clause on \crefrange{alg:ac:ifdiffaccepted}{alg:ac:conflict1}, \t{accepted} at $p$ equals $\camera{v}$.
This is the only \camera{value} $p$ accepts by \Cref{observation:accepted assinged once}.
The only processes added to \t{known\_acceptors} values on \cref{alg:ac:knownacceptorsme1,alg:ac:knownacceptorsq} are $p$ itself, and the processes in the values of the \t{knowledge\_of\_q} array from the \t{Accept\&Spread} message received from $q$. \t{knowledge\_of\_q} was the value of $\t{known\_acceptors}$ at $q$ when it sent the \t{Accept\&Spread} message. By the inductive hypothesis, its values include processes that assigned the same \camera{value} as $q$ to their \t{accepted} local variable, which is the \camera{value} $q$ sent in its \t{Accept\&Spread} message---namely, $\camera{v}$.
\end{proof}

Next, we present a couple of insights about knowledge expansion and its consequences.
\begin{observation}[knowledge does not decrease]\label{observation:knowledge only increases}
Processes are never removed from values of the \t{known\_acceptors} array at any process.
\end{observation}
\begin{proof}
The \t{known\_acceptors} values are updated only on \cref{alg:ac:knownacceptorsme1,alg:ac:knownacceptorsq}, where they are modified by taking the union with the current value, thus either remaining the same or gaining additional processes.
\end{proof}

\begin{corollary}\label{corollary: my known_acceptors[p] is subset of this value at p}
For any process $p$, the value of \t{known\_acceptors[$p$]} at some process $r$ at any point in the execution is a subset of the final value of \t{known\_acceptors[$p$]} at $p$.
\end{corollary}
\begin{proof}
For $p=r$, QED follows from \Cref{observation:knowledge only increases}.
We proceed by induction for the remaining case, showing that for any processes $p$ and $r$ such that $p \neq r$, the value of \t{known\_acceptors[$p$]} at $r$ at any point in the execution is a subset of the final value of \t{known\_acceptors[$p$]} at $p$. The base case vacuously holds since the \t{known\_acceptors} values are initially empty (\cref{alg:ac:knownacceptors}).
Now suppose the observation holds for all \t{known\_acceptors} values at a certain point in the execution, and consider the first following update of \t{known\_acceptors[$p$]} at some process $r$ on \cref{alg:ac:knownacceptorsq} such that $p \neq r$. It is updated to the union of its current value and the value of \t{knowledge\_of\_q[$p$]} (where \t{knowledge\_of\_q} is received in an \t{Accept\&Spread} message from process $q$). By the inductive hypothesis, the value of \t{known\_acceptors[$p$]} before this update is a subset of the final value of \t{known\_acceptors[$p$]} at $p$. So it remains to prove that the value of \t{knowledge\_of\_q[$p$]} is also a subset of the final value of \t{known\_acceptors[$p$]} at $p$. \t{knowledge\_of\_q[$p$]} was the value of \t{known\_acceptors[$p$]} at $q$ when it sent its \t{Accept\&Spread} message to $r$.
If $p=q$, then as mentioned at the beginning of this proof, \t{known\_acceptors[$p$]} at $p$ is a subset of its final value.
Else, by the inductive hypothesis, \t{known\_acceptors[$p$]} at $q$ when $q$ sent its \t{Accept\&Spread} message to $r$ is a subset of the final value of \t{known\_acceptors[$p$]} at $p$. 
\end{proof}

We now consider the local state of processes in case some process commits.

\begin{lemma}\label{lemma: requirements processes in known_acceptors if commit}
\camera{If process $p$ commits, every process
$q\in\t{REQUIREMENTS}[p]\t{.keys()}$ appears in $\t{known\_acceptors}[q]$ at $p$.}
\end{lemma}
\begin{proof}
\camera{Since \t{can\_commit} returns \t{True},
\crefrange{alg:commit:for}{alg:commit:false2} ensure that
$\t{REQUIREMENTS}[p][q]\subseteq\t{known\_acceptors}[q]$ at $p$, and the
normal form gives $q\in\t{REQUIREMENTS}[p][q]$.}
\end{proof}

\begin{corollary}[Consequence of \Cref{lemma: requirements processes in known_acceptors if commit} and \Cref{observation:processes in known_acceptors accepted the same value}]\label{corollary: requirements processes have same accepted if commit}
If a process $p$ commits \camera{value} $\camera{v}$, then every process in $\t{REQUIREMENTS}[p]$\t{.keys()} accepted $\camera{v}$, i.e., assigned $\camera{v}$ to its \t{accepted} local variable.
\end{corollary}
\begin{proof}
\camera{Let $q\in\t{REQUIREMENTS}[p]\t{.keys()}$ for a process $p$ that
commits. By \Cref{lemma: requirements processes in known_acceptors if commit},
$q\in\t{known\_acceptors}[q]$ at $p$. By
\Cref{observation:processes in known_acceptors accepted the same value}, $q$
accepted the same value as $p$.}
\end{proof}

The following observation identifies compatible requirements.

\begin{observation}\label{observation: compatible requirements intersect}
The processes in the keys of two compatible \camera{requirement} dictionaries intersect.
\end{observation}
\begin{proof}
Consider two compatible \camera{requirement} dictionaries $req_a$ and $req_b$.
Since they are compatible, the condition $|\t{witnesses\_a} \cup \t{witnesses\_b}| > f$ on \cref{alg:comp:bigger than f} evaluates to \t{True}.
Hence, it is impossible that both \t{witnesses\_a} and \t{witnesses\_b} are empty. By their definitions on \cref{alg:comp:witnesses_a,alg:comp:witnesses_b}, this implies that $req_a$\t{.keys()} $\cap$ $req_b$\t{.keys()} $\neq \emptyset$.
\end{proof}

Next, we identify behaviors related to adopting processes.

\begin{lemma}\label{lemma: adopter receives reply from process required to know about intersection}
For any distinct proposed \camera{values} $\camera{v}$ and $\camera{v'}$, an adopter receives a \t{Frozen} reply from some process required by $\camera{v}$ or $\camera{v'}$ to know about a process in the intersection of the keys in the \camera{requirement} dictionaries of $\camera{v}$ and $\camera{v'}$.
\end{lemma}
\begin{proof}
Let $\camera{v}$ and $\camera{v'}$ be distinct proposed \camera{values}.
An adopter $A$ waits for \t{Frozen} replies from $n-f$ processes (\cref{alg:ac:propose:wait}). As at most $f$ processes are faulty, the adopter is guaranteed to receive $n-f$ replies. By definition of compatible requirements (\cref{alg:comp:def}), at least $f+1$ processes are required by $\camera{v}$ or $\camera{v'}$ to know about some process in the intersection of the keys in the \camera{requirement} dictionaries of $\camera{v}$ and $\camera{v'}$. These $f+1$ processes intersect with the $n-f$ processes whose replies are received by adopter $A$. QED follows.
\end{proof}

\begin{lemma}\label{lemma: could_potentially_commit returns false if another value commits}
If \camera{value} $\camera{v}$ is committed, then any adopter invoking \t{could\_potentially\_commit} with a proposed \camera{value} $\camera{v'} \neq \camera{v}$ on some iteration of \cref{alg:ac:adopt:couldcommit} will obtain \t{False}.
\end{lemma}
\begin{proof}
Consider a process $r$ that commits \camera{value} $\camera{v}$ whose associated \camera{requirement is} $req_{\camera{v}}$, a process $r'$ proposing \camera{value} $\camera{v'} \neq \camera{v}$ whose associated \camera{requirement is} $req_{\camera{v'}}$, and an adopter $A$ that invokes \t{could\_potentially\_commit} with $\camera{v'}$ on some iteration of \cref{alg:ac:adopt:couldcommit}.
By \Cref{lemma: adopter receives reply from process required to know about intersection}, there exists process $q \in req_{\camera{v}}.keys() \cap req_{\camera{v'}}.keys()$, such that
$A$ receives a \t{Frozen} reply from some process $p$ satisfying 
(1) $p \in req_{\camera{v}}.keys()$ and $q \in req_{\camera{v}}[p]$, or
(2) $p \in req_{\camera{v'}}.keys()$ and $q \in req_{\camera{v'}}[p]$ (or both (1) and (2)).

First assume (1) holds. In this case, intuitively, since $r$ commits, the \camera{requirement} $req_{\camera{v}}$ \camera{is met}, and in particular $p$ knows about $q$ accepting $\camera{v}$; $A$ learns about it from $p$, and concludes that $q$ did not accept $\camera{v'}$, hence $req_{\camera{v'}}$ \camera{is not met} and thus $\camera{v'}$ may not be committed.

Formally, since $r$ commits, then a call by $r$ to \t{can\_commit($\camera{v}$)} on \cref{alg:ac:propose:until} returns \t{True}. This means that when committing, $req_{\camera{v}}[p] \subseteq \t{known\_acceptors}[p]$ at $r$ (\crefrange{alg:commit:not-in-2}{alg:commit:false2}), which implies $q \in \t{known\_acceptors}[p]$ at $r$.
By \Cref{corollary: my known_acceptors[p] is subset of this value at p}, the value of \t{known\_acceptors[$p$]} at $r$ when it commits is a subset of the final value of \t{known\_acceptors[$p$]} at $p$.
When $p$ handles the \t{Freeze} message from $A$ and replies with a \t{Frozen} message, $p$'s \t{known\_acceptors} has reached its final value, since it is modified only upon \t{Accept\&Spread} messages handled before the \t{Freeze} message is handled.
Consequently, $q \in \t{known\_acceptors}[p]$ at $p$ when it sends the \t{Frozen} reply to $A$. When $A$ executes \t{could\_potentially\_commit} with $\camera{v'}$, if it does not return \t{False} before iterating over $p$'s reply on \crefrange{alg:ac:adopt:for3}{alg:ac:adopt:false2}, it returns \t{False} on that iteration since \camerad{$p$ accepted $\camera{v} \neq \camera{v'}$ (by \Cref{corollary: requirements processes have same accepted if commit}) and} $q \in \t{acceptors} \cap req_{\camera{v'}}.keys()$.

Now assume (2) holds. 
In this case, intuitively, 
if $p$ did not accept $\camera{v'}$, then $A$ learns about it, and concludes that $req_{\camera{v'}}$ \camera{is not met} and thus $\camera{v'}$ may not be committed.
Else ($p$ accepted $\camera{v'}$), since $r$ commits, the \camera{requirement} $req_{\camera{v}}$ \camera{is met}, and in particular $q$ accepts $\camera{v}$, hence $p$ does not possess knowledge that $q$ accepted $\camera{v'}$. $A$ learns about it from $p$, and concludes that $req_{\camera{v'}}$ \camera{is not met} and thus $\camera{v'}$ may not be committed.

Formally,
assume first that $p$ did not accept $\camera{v'}$.
When $A$ executes \t{could\_potentially\_commit} with $\camera{v'}$, if it does not return \t{False} before iterating over $p$'s reply on \crefrange{alg:ac:adopt:for3}{alg:ac:adopt:false2}, it returns \t{False} on that iteration since $p \in req_{\camera{v'}}.keys()$.
Now assume $p$ accepted $\camera{v'}$.
Since $r$ commits $\camera{v}$, then $q$ accepted $\camera{v}$ by \Cref{corollary: requirements processes have same accepted if commit}. By \Cref{observation:accepted assinged once}, $q$ does not accept $\camera{v'}$. By \Cref{observation:processes in known_acceptors accepted the same value}, $p$ does not add $q$ to its $\t{known\_acceptors}[p]$.
Therefore, when $A$ executes \t{could\_potentially\_commit} with $\camera{v'}$, if it does not return \t{False} before iterating over $p$'s reply on \crefrange{alg:ac:adopt:for2}{alg:ac:adopt:false1}, it returns \t{False} on that iteration since $p \in req_{\camera{v'}}.keys()$, and $q \in req_{\camera{v'}}[p]$ but $q \notin \t{acceptors}$.
\end{proof}

\begin{lemma}\label{lemma: adopter receives Frozen with committed value}
If \camera{value} $\camera{v}$ is committed, then any adopter receives a \t{Frozen} reply with \camera{value} $\camera{v}$.
\end{lemma}
\begin{proof}
Consider an adopter $A$ receiving $n-f$ \t{Frozen} replies on \cref{alg:ac:propose:wait}. As we assume \camera{the requirement of each process is} valid, the \camera{requirement dictionary of $\camera{v}$'s proposer includes} at least $f+1$ processes as keys. Hence, at least one of the \t{Frozen} replies received by adopter $A$ must be from a process $p$ in the keys of \camera{this requirement} dictionary. By \Cref{corollary: requirements processes have same accepted if commit}, $p$ accepted $\camera{v}$. It assigned $\camera{v}$ to \t{accepted} on \cref{alg:ac:assignaccepted} before receiving the \t{Freeze} message, after which $p$ stops handling \t{Accept\&Spread} messages. Hence, upon receiving a \t{Freeze} message from $A$, $p$ replied with a \t{Frozen} message with \camera{value} $\camera{v}$.
\end{proof}

\begin{lemma}\label{lemma: could_potentially_commit returns true if commit}
If \camera{value} $\camera{v}$ is committed, then calling \t{could\_potentially\_commit} with $\camera{v}$ on \cref{alg:ac:adopt:couldcommit} by an adopter returns \t{True}.
\end{lemma}
\begin{proof}
Consider a process $r$ that commits \camera{value} $\camera{v}$ whose associated \camera{requirement is} $req_{\camera{v}}$, and an adopter $A$ that invokes \t{could\_potentially\_commit} with $\camera{v}$ on some iteration of \cref{alg:ac:adopt:couldcommit}.
Further consider a \t{Frozen}($p, \camera{val}, acceptors$) reply from process $p$ received by adopter $A$. We need to show that \t{could\_potentially\_commit} (\cref{alg:ac:adopt:def2}) with $\camera{v}$ does not return \t{False} during the iteration of the for loop in which this reply is processed within \crefrange{alg:ac:adopt:for3}{alg:ac:adopt:false1}.

First assume that $p \in req_{\camera{v}}.keys()$.
By \Cref{corollary: requirements processes have same accepted if commit}, $p$ accepted $\camera{v}$. The acceptance must happen before sending a \t{Frozen} reply to \camerad{$A$} (when handling the incoming \t{Freeze} message), because the value of \t{accepted} may change only before receiving a \t{Freeze} message.
Hence, $\camera{val}$ equals $\camera{v}$. Thus, when executing \t{could\_potentially\_commit} with $\camera{v}$, $A$ iterates over the considered \t{Frozen} reply on \crefrange{alg:ac:adopt:for2}{alg:ac:adopt:false1}.
Since $r$ commits, then a call by $r$ to \t{can\_commit($\camera{v}$)} on \cref{alg:ac:propose:until} returns \t{True}. This means that when committing, $req_{\camera{v}}[p] \subseteq \t{known\_acceptors}[p]$ at $r$ (\camerad{\crefrange{alg:commit:for}{alg:commit:false2}}).
By \Cref{corollary: my known_acceptors[p] is subset of this value at p}, the value of \t{known\_acceptors[$p$]} at $r$ when it commits is a subset of the final value of \t{known\_acceptors[$p$]} at $p$.
When $p$ handles the \t{Freeze} message from $A$ and replies with a \t{Frozen} message, $p$'s \t{known\_acceptors} has reached its final value, since it is modified only upon \t{Accept\&Spread} messages handled before the \t{Freeze} message is handled.
Consequently, $req_{\camera{v}}[p] \subseteq \t{known\_acceptors}[p]$ at $p$ when it sends the \t{Frozen} reply to $A$.
Thus, the condition on \cref{alg:ac:adopt:knowledge} evaluates to \t{False}, and \t{False} is not returned on \cref{alg:ac:adopt:false1} during the processing of $p$'s reply.

Otherwise ($p \notin req_{\camera{v}}.keys()$), if $\camera{val} = \camera{v}$ then the condition on \cref{alg:ac:adopt:knowledge} clearly evaluates to \t{False}.
Else, $\camera{val} \neq \camera{v}$. Then when executing \t{could\_potentially\_commit} with $\camera{v}$, $A$ iterates over this \t{Frozen} reply on \crefrange{alg:ac:adopt:for3}{alg:ac:adopt:false2}.
By \Cref{observation:processes in known_acceptors accepted the same value},
all processes in \camerad{$acceptors$} (which is the value of \t{known\_acceptors[p]} at process $p$ when it sends $A$ the \t{Frozen} reply) accepted the same \camera{value} as $p$ did, namely, $\camera{val}$. Due to \Cref{observation:accepted assinged once}, these processes cannot accept another \camera{value}, specifically not $\camera{v}$. By \Cref{corollary: requirements processes have same accepted if commit}, this implies that none of these processes is in $req_{\camera{v}}.keys()$.
Hence, the condition on \cref{alg:ac:adopt:acceptor} evaluates to \t{False}, and \t{False} is not returned on \cref{alg:ac:adopt:false2} during the processing of $p$'s reply.
\end{proof}

Next, we establish situations in which calling \t{can\_commit} returns \t{True}.

\begin{lemma}\label{lemma:can_commit returns true forever}
After reaching a configuration from which calling \t{can\_commit($\camera{v}$)} would return \t{True} at process $r$, a call to \t{can\_commit($\camera{v}$)} at $r$ from any subsequent configuration will return \t{True} as well.
\end{lemma}
\begin{proof}
Consider a configuration $C$ from which calling \t{can\_commit($\camera{v}$)} would return \t{True} at process $r$. We will show that on calling \t{can\_commit($\camera{v}$)} from any subsequent configuration, \t{False} may not be returned on \cref{alg:commit:not-accepted,alg:commit:false2}.

First, observe that \t{accepted = \camera{v}} at configuration $C$, due to \cref{alg:commit:not-accepted}. By \Cref{observation:accepted assinged once}, this holds also at any subsequent configuration.
Additionally, considering \camerad{\crefrange{alg:commit:for}{alg:commit:false2}}, for any process $p$ in the constant set \t{REQUIREMENTS[$r$].keys()}, the condition \t{REQUIREMENTS[$r$][$p$] $\subseteq$ known\_acceptors[$p$]} holds at configuration $C$. This holds also at any subsequent configuration, since \t{REQUIREMENTS} is constant, and values of \t{known\_acceptors} may only grow due to \Cref{observation:knowledge only increases}.
Overall, during a call to \t{can\_commit($\camera{v}$)} from any configuration subsequent to $C$, the conditions on \cref{alg:commit:not-accepted,alg:commit:not-in-2} are not satisfied, hence the call returns \t{True}.
\end{proof}

\begin{claim}\label{claim:if correct receive same value can_commit returns true}
If all correct processes receive \t{Accept\&Spread} messages such that they all receive the same \camera{value} $\camera{v}$ in their first received \t{Accept\&Spread} message, and none of them receives any \t{Freeze} message, then a correct process that calls \t{Propose($\camera{v}$)} eventually reaches a configuration from which calling \t{can\_commit($\camera{v}$)} would return \t{True}.
\end{claim}
\begin{proof}
Let the \camera{value} that appears in the first received \t{Accept\&Spread} message of all correct processes be $\camera{v}$. Let $\mathcal{C}$ denote the set of correct processes in the system. Assume none of them receives any \t{Freeze} message, thus they all handle any received \t{Accept\&Spread} message.
Let $r$ be a correct process that calls \t{Propose($\camera{v}$)}.

Assume some process in \t{REQUIREMENTS[$r$].keys()} is faulty. $r$'s failure detector is guaranteed to detect it at some point and send a \t{DoAdopt} message to $r$ on \cref{alg:fd:conflict}. 
Hence, $r$ is guaranteed to break from the wait on \cref{alg:ac:propose:until}.
But $r$ cannot break from the wait due to receiving a \t{DoAdopt} message, since otherwise $r$ would send a \t{Freeze} message to all on \cref{alg:ac:propose:freeze}, contradicting our assumption that no correct process receives a \t{Freeze} message. Therefore, the configuration at which $r$ breaks from the wait is one from which calling \t{can\_commit($\camera{v}$)} returns \t{True}, and we are done.

Now assume all processes in \t{REQUIREMENTS[$r$].keys()} are correct.
For each process $p \in \mathcal{C}$, \Cref{observation:bcast after accept} implies that after $p$ receives an \t{Accept\&Spread} message for the first time, it sends \t{Accept\&Spread}($\camera{v}$, $\mathcal{K}$) to others on \cref{alg:ac:broadcast}, where $\mathcal{K}$ is \camerad{an array with one set of processes per process}.
We have $p \in \mathcal{K}[p]$ because $p$ is explicitly added to \t{known\_acceptors}[$p$] at $p$ on \cref{alg:ac:knownacceptorsme1}. 
As a result, every process $q \in \mathcal{C}$ eventually receives from every process $p \in \mathcal{C}$ an \t{Accept\&Spread}($\camera{v}$, $\mathcal{K}$) message such that $p \in \mathcal{K}[p]$, and adds $p$ to \t{known\_acceptors}[$q$] at $q$ on \cref{alg:ac:knownacceptorsme1}.
Consequently, at one of its executions of \cref{alg:ac:knownacceptorsme1}, $q$ assigns to \t{known\_acceptors}[$q$] at $q$ a superset of $\mathcal{C}$ for the first time.
This implies that for $q=r$, \t{known\_acceptors}[$q$] at $r$ eventually contains a superset of $\mathcal{C}$. 
As for $q \neq r$, $q$ sends to others on \cref{alg:ac:broadcast} an \t{Accept\&Spread} message with a \t{known\_acceptors} value in which \t{known\_acceptors}[$q$] is a superset of $\mathcal{C}$.
Therefore, for each process $q \neq r$ in \t{REQUIREMENTS[$r$].keys()} (which is assumed to be a correct process), $r$ eventually receives from $q$ an \t{Accept\&Spread}($\camera{v}$, $\mathcal{K}_q$) where $\mathcal{C} \subseteq \mathcal{K}_q[q]$, and assigns to \t{known\_acceptors}[$q$] at $r$ a superset of $\mathcal{C}$ on \cref{alg:ac:knownacceptorsq}.
\camera{Because the normal form gives
$\t{REQUIREMENTS}[r][q]\subseteq\t{REQUIREMENTS}[r].\t{keys}()$
and every process in this key set is correct,
$\t{REQUIREMENTS}[r][q]\subseteq\mathcal{C}$. Hence, $r$ eventually reaches a
configuration in which
$\t{REQUIREMENTS}[r][q]\subseteq\t{known\_acceptors}[q]$ for every key $q$.}
As values of \t{known\_acceptors} may only grow due to \Cref{observation:knowledge only increases}, $r$ reaches a configuration from which for every process $q \in$ \t{REQUIREMENTS[$r$].keys()}, $\t{REQUIREMENTS}[r][q] \subseteq \t{known\_acceptors}[q]$ at $r$, hence calling \t{can\_commit($\camera{v}$)} from this configuration would return \t{True}.
\end{proof}

We next address the termination of \t{Propose} calls in certain scenarios.

\begin{observation}\label{observation:adopt returns}
If a correct process reaches the wait on \cref{alg:ac:propose:wait} while executing a \t{Propose} call, then this wait eventually ends.
\end{observation}
\begin{proof}
Let $r$ be a correct process running a call to \t{Propose} and reaching \cref{alg:ac:propose:wait}.
Since $r$ is correct, given our system model, which ensures at most $f$ faulty processes and assumes perfect links, the messages sent by $r$ on \cref{alg:ac:propose:freeze} are guaranteed to eventually reach $n-f$ correct processes. These processes will then reply, and their replies will eventually reach $r$, allowing the wait on \cref{alg:ac:propose:wait} to eventually complete. 
\end{proof}

\begin{observation}\label{claim:can_commit returns true implies propose returns}
Assume there exists a configuration in a given execution from which calling \t{can\_commit($\camera{v}$)} at a correct process $r$ would return \t{True}.
Then if $r$ calls \t{Propose($\camera{v}$)} at any point in the execution, this call is guaranteed to eventually return.
\end{observation}
\begin{proof}
Assume a correct process $r$ calls \t{Propose($\camera{v}$)} at some point.
Recall our assumption that any block within a method call is executed atomically unless it contains wait steps. So it suffices to show that the wait steps performed during a \t{Propose} call (on \cref{alg:ac:propose:until,alg:ac:propose:wait}) eventually return.
If $r$ reaches the wait on \cref{alg:ac:propose:wait}, it is guaranteed to end this wait by \Cref{observation:adopt returns}.
As for \cref{alg:ac:propose:until}, 
if the first configuration from which calling \t{can\_commit($\camera{v}$)} at $r$ would return \t{True} occurs after $r$ passes \cref{alg:ac:propose:until} then we are done. Otherwise, by \Cref{lemma:can_commit returns true forever}, \t{can\_commit($\camera{v}$)} returns \t{True} at some point when $r$ executes \cref{alg:ac:propose:until}, thus it ends the wait.

\end{proof}

\begin{claim}\label{claim:DoAdopt implies propose returns}
Assume a correct process $r$ receives a \t{DoAdopt} message in a certain execution.
Then if $r$ calls \t{Propose} at any point in the execution, this call is guaranteed to eventually return.
\end{claim}
\begin{proof}
Assume a correct process $r$ calls \t{Propose} at some point in a certain execution, and receives a \t{DoAdopt} message at some point in that execution.
Recall our assumption that any block within a method call is executed atomically unless it contains wait steps. So it suffices to show that the wait steps performed during a \t{Propose} call (on \cref{alg:ac:propose:until,alg:ac:propose:wait}) eventually end.
The wait on \cref{alg:ac:propose:until} is guaranteed to end due to the receipt of a \t{DoAdopt} message.
As for \cref{alg:ac:propose:wait}, if $r$ reaches this wait, it is guaranteed to end it by \Cref{observation:adopt returns}.
\end{proof}

We are now ready to prove that our adopt-commit protocol satisfies termination, validity and agreement.

\begin{theorem}[Termination]\label{theorem:termination}
Every correct process that calls \t{Propose} in a certain execution of \sysname (Algorithm \ref{alg:ac:propose}) adopts or commits.
\end{theorem}
\begin{proof}
We show that every correct process that calls \t{Propose} returns from this call. The theorem then follows, since a \t{Propose} invocation adopts or commits immediately before returning (on either \cref{alg:ac:propose:commit} or \cref{alg:ac:propose:adopt}).

If some correct process receives a \t{Freeze} message, it sends a \t{DoAdopt} message to all on \cref{alg:ac:conflict2}. 
By our network model, which assumes perfect links, all correct processes eventually receive a \t{DoAdopt} message.
By \Cref{claim:DoAdopt implies propose returns}, we are done.

Otherwise, no correct process receives a \t{Freeze} message, thus they all handle any received \t{Accept\&Spread} message.
If no correct process calls \t{Propose}, termination is vacuously satisfied, hence, assume at least one correct process calls \t{Propose}. Then by \Cref{lemma:if correct proposes everyone receives AcceptAndSpread}, every correct process eventually receives an \t{Accept\&Spread} message.
We consider two distinct cases.

In the first case, two correct processes receive different \camera{values} in their first received \t{Accept\&Spread} message. By \Cref{observation:bcast after accept}, each of them sends its received \camera{value} to others via an \t{Accept\&Spread} message (on \cref{alg:ac:broadcast}). By our network model, which assumes perfect links, all correct processes eventually receive these conflicting \t{Accept\&Spread} messages. Each correct process, upon receiving the first \t{Accept\&Spread} message, sets its \t{accepted} to the included \camera{value} on \cref{alg:ac:assignaccepted}. When it receives an \t{Accept\&Spread} message with a different \camera{value}, it sends all---including itself---a \t{DoAdopt} message on \cref{alg:ac:conflict1}. By \Cref{claim:DoAdopt implies propose returns}, its \t{Propose} call returns.

In the complementary case, all correct processes receive the same \camera{value} in their first received \t{Accept\&Spread} message. Let this \camera{value} be $\camera{v}$. Consider some correct process $r$ that calls \t{Propose}. If it proposes a \camera{value} different from $\camera{v}$, then it sends an \t{Accept\&Spread} message with its \camera{value} to itself on \cref{alg:ac:propose:disseminate}, and upon handling it on \cref{alg:ac:upon1}, it sends all---including itself---a \t{DoAdopt} message on \cref{alg:ac:conflict1}. By \Cref{claim:DoAdopt implies propose returns}, its \t{Propose} call returns.
Otherwise ($r$ proposes \camera{value} $\camera{v}$), since by \Cref{claim:if correct receive same value can_commit returns true} it eventually reaches a configuration from which calling \t{can\_commit($\camera{v}$)} would return \t{True}, then we are done by \Cref{claim:can_commit returns true implies propose returns}.
\end{proof}

\begin{theorem}[Validity]
Every \t{Propose} call that returns in a certain execution of \sysname, adopts or commits a \camera{value} that was the input of some \t{Propose} call.
\end{theorem}
\begin{proof}
\t{Propose} may adopt or commit at one of two points. 
It can commit its input \camera{value} on \cref{alg:ac:propose:commit}, in which case validity is trivially satisfied.
Alternatively, it can adopt on \cref{alg:ac:propose:adopt} either a \camera{value} returned by the \t{potential\_commit} call on \cref{alg:ac:propose:potential} or some proposed \camera{value} picked on \cref{alg:ac:propose:any}. In the latter case we are done. In the first case, \t{potential\_commit} returns a non-$\bot$ \camera{value} appearing in one of the \t{Frozen} replies received at the process (\crefrange{alg:ac:adopt:for1}{alg:ac:adopt:couldcommit}). A \t{Frozen} message contains the \camera{value} that appears in the \t{accepted} local variable at the sending process (\cref{alg:ac:frozen}). The value of \t{accepted} is set only on \cref{alg:ac:assignaccepted}, to a \camera{value} received in an \t{Accept\&Spread} message.
By \Cref{observation:AcceptAndSpread disseminate proposed values}, such a \camera{value} is a proposed \camera{value}.
\end{proof}

\begin{theorem}[Agreement]\label{theorem:ac:agreement}
If some \t{Propose} call in a certain execution of \sysname commits $\camera{v}$, then all \t{Propose} calls that \camerad{return} adopt or commit \camera{value} $\camera{v}$.
\end{theorem}
\begin{proof}
Suppose a proposer $r$ commits its \camera{value} $\camera{v}$.
Let $r'$ be any other process that returns from a \t{Propose} call, and let $\camera{v'}$ be the \camera{value} it commits or adopts. We need to show that $\camera{v'}=\camera{v}$.
Since $r$ commits $\camera{v}$ (which may occur only on \cref{alg:ac:propose:commit}), it accepted $\camera{v}$ according to \cref{alg:commit:not-accepted}. By \Cref{corollary: requirements processes have same accepted if commit}, all processes in \t{REQUIREMENTS[$r$].keys()} have assigned \camera{value} $\camera{v}$ to their \t{accepted} local variable.

If $r'$ commits, then following the same reasoning as for $r$, $r'$ accepted $\camera{v'}$ and all processes in \t{REQUIREMENTS[$r'$].keys()} have assigned \camera{value} $\camera{v'}$ to their \t{accepted} local variable.
But as we assume a compatible combination of requirements, \t{REQUIREMENTS[$r$].keys()} and \t{REQUIREMENTS[$r'$].keys()} intersect by \Cref{observation: compatible requirements intersect}. Consider a process in the intersection. We showed that it should have assigned both \camera{values} $\camera{v}$ and $\camera{v'}$ to its \t{accepted} local variable. Since each process accepts at most one \camera{value} by \Cref{observation:accepted assinged once}, $\camera{v'}=\camera{v}$ in this case.

Otherwise, $r'$ adopts $\camera{v'}$.
We next show that $r'$ obtains the adopted $\camera{v'}$ on \cref{alg:ac:propose:potential} through a call to \t{potential\_commit} that returns $\camera{v}$, hence $\camera{v'}=\camera{v}$.
First observe that by \Cref{lemma: could_potentially_commit returns false if another value commits}, calling \t{could\_potentially\_commit} on \cref{alg:ac:adopt:couldcommit} with any $\camera{v^*} \neq \camera{v}$ returns \t{False}.
Additionally, $r'$ receives a \t{Frozen} reply with \camera{value} $\camera{v}$ by \Cref{lemma: adopter receives Frozen with committed value}. Consequently, after calls to \t{could\_potentially\_commit} return \t{False} in previous iterations of the for loop on \cref{alg:ac:adopt:for1}, $r'$ reaches an iteration in which it calls \t{could\_potentially\_commit} with $\camera{v}$.
By \Cref{lemma: could_potentially_commit returns true if commit}, this call returns \t{True}. Thus, the \t{potential\_commit} call at $r'$ returns $\camera{v}$.
\end{proof}

\subsection{\sysname's Correctness}\label{proof:consensus-correctness}

\sysname uses the consensus protocol in Algorithm~\ref{alg:ac-to-consensus} with the implementation of adopt-commit appearing in \Cref{sec:ac}. 
We prove that Algorithm~\ref{alg:ac-to-consensus} satisfies the  consensus properties presented in the beginning of this section. 

\begin{theorem}
Algorithm~\ref{alg:ac-to-consensus} satisfies termination, validity and agreement.
\end{theorem}
\begin{proof}
The correctness of the algorithm is derived from the correctness of adopt-commit and the fallback consensus.

For termination, the call to \t{AdoptCommit::Propose} (\cref{alg:consensus:propose:ac propose}) in a correct process must result in a $\t{Commit}$ or $\t{Adopt}$ event, by adopt-commit's termination property. In case of commit, the consensus protocol decides on \cref{alg:consensus:propose:decide upon commit}. In case of adopt, the fallback consensus is invoked (\cref{alg:consensus:propose:fallback propose}). Since the fallback protocol is guaranteed to decide (by its termination property), the main consensus protocol must also decide (\cref{alg:consensus:propose:decide upon fallback decide}).

As for validity, adopt-commit is invoked only with \camera{values} proposed in the consensus protocol (\cref{alg:consensus:propose:ac propose}).
By adopt-commit's validity property, a call to \t{AdoptCommit::Propose} can only result (unless crashed) in adopting or committing a \camera{value} that was previously proposed within adopt-commit, hence proposed by the consensus protocol.
In case of commit, the consensus protocol decides that \camera{value} (\cref{alg:consensus:propose:decide upon commit}), and we are done. In case of adopt, the fallback consensus is invoked with that \camera{value} (\cref{alg:consensus:propose:fallback propose}).
In the broader view, all processes whose adopt-commit calls trigger adoption will proceed (unless crashed) to invoke the fallback consensus with the \camera{value} they adopted, which is a \camera{value} originally proposed in the consensus protocol.
By its validity property, the fallback protocol is guaranteed to decide a \camera{value} proposed for it (which is, as explained, a \camera{value} proposed in the consensus protocol). Thus, the main consensus protocol must also decide a \camera{value} proposed for it (\cref{alg:consensus:propose:decide upon fallback decide}) (unless crashed).

Finally, we prove agreement. 
If the adopt-commit invocation does not result in commit in any process, then no process executes \cref{alg:consensus:propose:decide upon commit}.
All processes that complete the adopt-commit invocation, invoke the fallback consensus protocol (\cref{alg:consensus:propose:fallback propose}), and decide the same \camera{value} for the main consensus protocol (\cref{alg:consensus:propose:decide upon fallback decide}) by \camerad{the fallback} consensus's agreement (unless crashed).
Otherwise, assume the adopt-commit invocation on some process commits some \camera{value} $\camera{v}$.
By adopt-commit's agreement, the adopt-commit invocation on all processes may adopt or commit only $\camera{v}$. Every process may decide on \cref{alg:consensus:propose:decide upon commit} or \cref{alg:consensus:propose:decide upon fallback decide}.
In the first case, it decides the committed value $\camera{v}$. In the second case, since all processes that adopt invoke the fallback protocol with \camera{value} $\camera{v}$, then by \camerad{the fallback} consensus's validity, they decide $\camera{v}$ for the main consensus protocol (\cref{alg:consensus:propose:decide upon fallback decide}).
\end{proof}

 \renewcommand{\emph}[1]{\textit{#1}}

\section{\sysname Optimality} \label{sec:optimality}
We start by defining some terms used in our optimality proof.

\begin{definition}[happened-before \cite{lamport78happenedbefore}]
\emph{Happened-before} denotes the smallest relation on the set of steps of an execution satisfying the following three conditions: 
\begin{enumerate}
    \item If $a$ and $b$ are steps of the same process and $a$ comes before $b$ in the execution, then $a$ \emph{happened-before} $b$.
    \item If $a$ is the sending of a message by one process and $b$ is the receipt of the same message by another process, then $a$ \emph{happened-before} $b$.
    \item If $a$ \emph{happened-before} $b$ and $b$ \emph{happened-before} $c$, then $a$ \emph{happened-before} $c$.    
\end{enumerate}
\end{definition}

\begin{definition}[indistinguishable]
Two executions $e_1,e_2$ are \emph{indistinguishable} to some process $r$ until time $t$ (including), if in both executions, 
\begin{enumerate}
    \item \t{Propose} is invoked at $r$ at the same time $\leq t$ or not invoked at $r$ at a time $\leq t$.
    \item $r$ receives the same messages at the same times until time $t$ (including). 
\end{enumerate}
We denote this relationship as $e_1 \stackrel{r}{\approx}_t e_2$.
We similarly denote executions indistinguishable to some process $r$ until time $t$ \emph{excluding} as $e_1 \stackrel{r}{\approx}_{<t} e_2$.

If the executions are indistinguishable to $r$ at all times, we omit the time indication: $e_1 \stackrel{r}{\approx} e_2$.
\end{definition}

\begin{definition}[known message]
We define a message $m$ in an execution to be \emph{known} by process $r$ at time $t$, if either $m$ is received by $r$ up to time $t$, or $m$'s receipt happened-before some message receipt at $r$ occurring up to time $t$. We say in this case that $r$ \emph{knows} about $m$ at $t$.
\end{definition}

We assume an eventually-stable network. In this section, we prove optimality in this eventual state of stability (defined below), but the proof could be easily extended to earlier long-enough periods of synchrony.

\begin{definition}[stable network]
A network is \emph{stable} if all inter-process latencies are stable and known to processes, all faulty processes are permanently down, and weak failure detectors in the system suspect all faulty processes while they do not suspect any correct process as having crashed.
\end{definition}

\fix{\begin{definition}[weighted objective latency]
We define the \emph{weighted objective latency} of a consensus protocol, in a stable network setting, to be the objective-function time from proposal to decision, where the objective function is taken over the set of executions in which \t{Propose} is invoked on a single correct process, weighted by the probability that each correct process is the proposer. The objective function could be the expectation (yielding the weighted expected latency), or a percentile (yielding weighted tail or median latency).
\end{definition}}

\begin{definition}[optimal combination of requirements of \sysname]
Given an objective latency function, an optimal combination of requirements is any assignment to \t{REQUIREMENTS} of \sysname, which results in the minimal weighted objective latency achievable by \sysname for a given stable network across all combinations of requirements that satisfy \sysname's conditions on requirements. Details on how to obtain it appear in \Cref{sec:optimizer:solve}.
\end{definition}
When using \sysname, assuming a certain objective latency function, \sysname is run with an optimal combination of requirements. In the following proof, we sometimes run the \sysname protocol with other combinations of requirements just for the sake of the proof.

\medskip

\camera{The executions built in the proofs below suppress some messages, e.g., all those a process sends from some point on. A suppressed message is never dropped in violation of the model: it is either sent by a process that crashes later in that execution, or delayed until after every correct process has decided, which unbounded delays allow. Each suppression removes a suffix of the messages a link carries, so in-order delivery is preserved.}

We are ready to present our main optimality result:

\vspace{7pt}
{\setlength{\fboxrule}{0.5pt}\setlength{\fboxsep}{8pt}
\noindent\fcolorbox{black}{gray!13}{\begin{minipage}{\dimexpr\linewidth-2\fboxsep-2\fboxrule\relax}
\begin{theorem}[Optimality]\label{theorem:optimality}
Given a certain objective latency function, \sysname has optimal weighted objective latency in a stable network setting across all deterministic consensus protocols.

\textit{Assumption:} We consider consensus protocols in which some \t{Propose} invocation happened-before any message sending.
\camera{This assumption is without loss of generality for the executions we compare: a message sent with no proposal happening before it does not depend on the system input, so it is sent identically and can be delivered at the same times in all of them, and it carries no information about a proposed value. We rely on this assumption when deriving $AlgReqs$ below, where every message known by a proposer when it decides is read as spreading its value.}
\end{theorem}
\end{minipage}}}

\begin{proof}
Let \textcolor{cyan}{Alg} be any deterministic consensus protocol. Assume a given weighted objective latency.
We show that when the network is stable, Alg cannot have better weighted objective latency than \sysname. 

Consider a stable network with up to $f$ faulty processes.
Let \textcolor{cyan}{$\mathcal{F}$} indicate the set of faulty processes (the rest of the processes are \emph{correct}).
Let \textcolor{cyan}{$\mathcal{L}$} indicate inter-process latencies in the stable network for all pairs of correct processes, where
$\mathcal{L}[s,p]$ indicates the inter-process latency from $s$ to $p$.
We assume, as is standard, that message transmission takes a strictly positive amount of time, 
except when a process sends a message to itself, which is assumed to take zero time. 
Namely, $\mathcal{L}[s,p] > 0$ for $s \neq p$ and $\mathcal{L}[p,p] = 0$.

For simplicity, we intuitively assume that independent events cannot happen at a process at the same time. Formally, we assume that a process
$p$ does not send messages right before (at the exact same time) it receives a message from another process, and we assume that $p$ cannot receive more than one message from other processes at the same time. Namely, at any given time, $p$ may receive at most one message from another process, then it can send messages, then possibly receive a message from itself, followed by another round of message sending; the last two steps may repeat multiple times.
This assumption implies that if a process $p$ receives a message from another process at some time $t$, then this receipt is the first message receipt or sending event by $p$ at $t$. This will be important in the definition and usage of $t_{p,r}$ below.

We define a combination of requirements \textcolor{cyan}{$AlgReqs$} by considering, for each correct process $r$, an execution where \t{Propose} is invoked on $r$ only, and using it to define the \camera{requirement} $AlgReqs[r]$. This is done in the following manner.
For each process $r \notin \mathcal{F}$, let \textcolor{cyan}{$e_r$} be a single-proposer execution of Alg in which Alg.\t{Propose(\textcolor{cyan}{$\camera{v_r}$})} is invoked on $r$ at time 0---where $\camera{v_r}$ is a unique \camera{value} associated with each $r$, and the network is stable such that
processes in $\mathcal{F}$ are crashed, 
and message transmission among correct processes is according to the latencies of $\mathcal{L}$. By termination, $r$ must decide. 
Consider all messages known by $r$ when $r$ decides. Process these messages using Algorithm~\ref{alg:optimality:knowledge} in order of their reception time to compute \t{known\_acceptors} for all processes that receive any of these messages.
Set $AlgReqs[r]$ to \t{array\_to\_dict(known\_acceptors[$r$])}, where the value of \t{known\_acceptors[$r$]} is obtained after processing all these messages in Algorithm~\ref{alg:optimality:knowledge}, and \t{array\_to\_dict} (\cref{alg:optimality:array_to_dict}) transforms it into a dictionary (each process $p$ for which \t{known\_acceptors[$r$][$p$]} is not empty becomes a key, with \t{known\_acceptors[$r$][$p$]} as its associated value).

\begin{lstlisting}[caption={Knowledge Tracking: algorithm for execution $e_r$ of proposer $r$},label={alg:optimality:knowledge}]
known_acceptors: array[n][n] of set<pid> = {@$\emptyset$@} // Initialize all array cells to empty sets
known_acceptors[r][r] = {r}@\label{alg:optimality:init r}@
M = all messages @\text{in} $e_r$@ known by r @\text{when}@ r decides@\label{alg:optimality:known messages}@

for each message m from s to t in M, sorted by reception time:@\label{alg:optimality:for loop}@
  Let knowledge_of_s be the value of known_acceptors[s] right @\text{after}@ processing the last message in M received by s prior to sending m (@\text{or}@ its initial value @\text{if}@ no such message exists).
  known_acceptors[t][t] @$\cup$@= {t} @$\cup$@ knowledge_of_s[s]@\label{alg:optimality:knownacceptorsme}@
  for q in pids:@\label{alg:optimality:inner for loop}@
    known_acceptors[t][q] @$\cup$@= knowledge_of_s[q]@\label{alg:optimality:knownacceptorsq}@

def array_to_dict(known_acceptors_of_process):@\label{alg:optimality:array_to_dict}@
  dict = {}
  for p in pids:
    if known_acceptors_of_process[p] @$\neq$ $\emptyset$@:
      dict[p] = known_acceptors_of_process[p]
  return dict
\end{lstlisting}

\camera{We prove in
\Cref{claim:optimality:reflexive graph,claim:optimality:valid,claim:optimality:compatible}
that every derived requirement $AlgReqs[r]$ has the normal form of
\S\ref{sec:constraints:condition}, is valid, and is pairwise compatible with
the other derived requirements. Hence, \sysname could use $AlgReqs$.}
Let \textcolor{cyan}{$\sysname_{AlgReqs}$} be \sysname with the combination of requirements $AlgReqs$.
We show in \Cref{claim:optimality:KCensus_Alg not slower than Alg} below that when the network is stable (with faulty processes $\mathcal{F}$ and inter-process latencies $\mathcal{L}$), Alg does not have better weighted objective latency than $\sysname_{AlgReqs}$. Since $\sysname$ uses an optimal combination of requirements, $\sysname_{AlgReqs}$ cannot have better weighted objective latency than $\sysname$ for the same stable network. Overall, we get that when the network is stable, Alg does not have better weighted objective latency than $\sysname$. This implies that $\sysname$ is a consensus protocol with optimal weighted objective latency.
\end{proof}

We define some notations that will be used in the following proofs, in addition to the notations from the proof of \Cref{theorem:optimality}.
Henceforth, references to a process $r$ for which an execution $e_r$ exists, or for which $AlgReqs[r]$ is defined, implicitly assume that $r$ is correct.
Consider some process $r$. 
Let \textcolor{cyan}{$t_r$} be the time at which $r$ decides in $e_r$. 
Let \textcolor{cyan}{$K_r$} $= AlgReqs[r]$\t{.keys()}.
For each process $p \in K_r$, let \textcolor{cyan}{$t_{p,r}$} be the time in which $p$ receives in $e_r$ its first message from another process that is not known by $r$ at time $t_r$, or $t_r + \epsilon$ if no such message exists, where \textcolor{cyan}{$\epsilon$} is some positive number. 
(The reason we pick the first message \emph{from another process} is that self messages keep executions, which were indistinguishable to a process until receiving them, indistinguishable; so there is no need to discard them to obtain indistinguishable executions. We later define executions in which we discard all messages \emph{from other processes} not known by $r$ at $t_r$. To this end, we discard all the messages to / from process $p$ that $p$ does not receive / send before time $t_{p,r}$. All messages $p$ receives / sends at $t_{p,r}$ or later are indeed not known by $r$ at $t_r$, thanks to the fact that the receipt of the message \emph{from another process} by $p$ at $t_{p,r}$ is the first message receipt or sending event by $p$ at $t_{p,r}$, combined with \Cref{observation:optimality:message received before known received or sent - is known}.)

Next, we bring the main claims used in our optimality proof, followed by lemmas used to prove them.

\begin{claim}\label{claim:optimality:reflexive graph}
\camera{For each process $r$, the dictionary $AlgReqs[r]$ satisfies:
(1) every key appears within its value; (2) all processes in the values appear
as keys; and (3) $r$ is a key and
$AlgReqs[r][r]=AlgReqs[r].\t{keys}()$.}
\end{claim}
\begin{proof}
Let $r$ be some process.
We prove by induction on messages in $e_r$ known by $r$ at $t_r$, by the order they are processed in Algorithm~\ref{alg:optimality:knowledge}, that for every process $p$, conditions 
(1) and (2) hold for the dictionary \t{array\_to\_dict(known\_acceptors[$p$])} right after the message is processed. As $AlgReqs[r]$ equals the value of \t{array\_to\_dict(known\_acceptors[$r$])} after processing all these messages, QED follows.

For the base case, before any message is handled, all sets in \t{known\_acceptors[$p$]} for any process $p \neq r$ are empty, hence the transformation of \t{known\_acceptors[$p$]} into a dictionary yields an empty dictionary, that trivially satisfies conditions (1) and (2). As for process $r$, the only non-empty set in \t{known\_acceptors[$r$]} at initialization is \t{known\_acceptors[$r$][$r$]} which equals $\{r\}$. Hence, the transformation of \t{known\_acceptors[$r$]} into a dictionary yields a dictionary with $r$ as the only key, having $\{r\}$ as its associated value. This dictionary clearly satisfies conditions (1) and (2).

Suppose the statement holds until right before some process $t$ receives from some process $s$ a message $m$ known by $r$ at $t_r$.
When processing $m$ in Algorithm~\ref{alg:optimality:knowledge}, \t{known\_acceptors[$t$]} is updated.
By the inductive hypothesis, \t{array\_to\_dict(known\_acceptors[$t$])} for the value of \t{known\_acceptors[$t$]} right before processing $m$,
and \t{array\_to\_dict(knowledge\_of\_s)} (where \t{knowledge\_of\_s} is the value of \t{known\_acceptors[$s$]} right after processing the last message, known by $r$ at $t_r$, received by $s$ prior to sending $m$, or the initial value of \t{known\_acceptors[$s$]} if no such message exists),
satisfy conditions (1) and (2).

We first show that condition (1) is satisfied after processing $m$.
For any key that already existed in \t{array\_to\_dict(known\_accep\-tors[$t$])} before processing $m$, condition (1) holds by the inductive hypothesis (as mentioned above).
If $t$ is added to \t{array\_to\_di\-ct(known\_acceptors[$t$])} on \cref{alg:optimality:knownacceptorsme} (which happens if \t{known\_acceptors[$t$][$t$] is empty before processing $m$}), condition (1) holds since $t$ is added to \t{known\_acceptors[$t$][$t$]}. 
For any key $q$ added to \t{array\_to\_dict(known\_acceptors[$t$])} on \cref{alg:optimality:knownacceptorsq} ($q$ is added if \t{known\_acceptors[$t$][$q$]} $= \emptyset$ before processing $m$ while \t{knowledge\_of\_s[$q$]} $\neq$ $\emptyset$), $q \in$ \t{know\-ledge\_of\_s[$q$]} due to the inductive hypothesis (as mentioned above), and is thus added to \t{known\_acceptors[$t$][$q$]}.

We proceed to show that condition (2) is satisfied after processing $m$.
The processes that may be added to sets in \t{known\_acceptors[$t$]} when processing $m$, are $t$ (on \cref{alg:optimality:knownacceptorsme}) and processes appearing in the sets in \t{knowledge\_of\_s} (on \cref{alg:optimality:knownacceptorsme,alg:optimality:knownacceptorsq}). 
$t$ is guaranteed to be a key in \t{array\_to\_dict(known\_acceptors[$t$])} after processing $m$ due to \cref{alg:optimality:knownacceptorsme}.
As for processes appearing in the sets in \t{knowledge\_of\_s}, due to the inductive hypothesis (as mentioned above), each process $p$ appearing in some set in \t{knowledge\_of\_s} is a key in \t{array\_to\_dict(knowledge\_of\_s)}, i.e., \t{knowledge\_of\_s[$p$]} $\neq$ $\emptyset$. Hence, $p$ is guaranteed to be a key in \t{array\_to\_dict(known\_acceptors[$t$])} after \cref{alg:optimality:knownacceptorsq} is executed with $q=p$.

\camera{It remains to prove (3). Process $r$ is a key because
\t{known\_acceptors[$r$][$r$]} is initialized to $\{r\}$.
Condition (2) gives
$AlgReqs[r][r]\subseteq AlgReqs[r].\t{keys}()$.
Since \t{known\_acceptors[$r$]} changes only when $r$ receives a message, the
reverse inclusion follows from
\Cref{lemma:optimality:known_acceptors[p].keys() subseteq known_acceptors[p][p] after handling received known message}
applied to the last processed message received by $r$, or from
initialization if there is none.}
\end{proof}

\begin{claim}\label{claim:optimality:valid}
\camera{The requirement} $AlgReqs[r]$ \camera{of} each process $r$ \camera{is} valid (see definition in Algorithm \ref{alg:ac:compatible}).
\end{claim}
\begin{proof}
Assume, for sake of contradiction, that there exists some process $r$ whose associated \camera{requirement}, $AlgReqs[r]$, \camera{is} not valid. This means that $|K_r| \leq f$.

Let \textcolor{cyan}{$\mathcal{L'}$} indicate inter-process latencies for all pairs of processes not in $K_r$ in the following way: $\mathcal{L'}[p,p] = 0$ for any $p \notin K_r$, and for every two processes $s,p \notin K_r$ such that $s \neq p$, $\mathcal{L'}[s,p] = \mathcal{L}[s,p]$ if $s,p \notin \mathcal{F}$, else $\mathcal{L'}[s,p]$ is assigned an arbitrary positive value.
Let $r'$ be some process $\notin K_r$.
    
Consider a single-proposer execution \textcolor{cyan}{$e_1$} of Alg, in which: 
\begin{itemize}
    \item Alg.\t{Propose($\camera{v_{r'}}$)} is invoked on process $r'$ at time 0.
    \item Only the processes in $K_r$ are faulty, and they all crash at time 0. 
    \item Message transmission among processes not in $K_r$ is according to the latencies of $\mathcal{L'}$.   
\end{itemize}

Further consider an execution \textcolor{cyan}{$e_2$} of Alg, in which:
\begin{itemize}
    \item Alg.\t{Propose($\camera{v_r}$)} is invoked on process $r$ at time 0.
    \item Alg.\t{Propose($\camera{v_{r'}}$)} is invoked on process $r'$ at time 0.
    \item Only the processes in $K_r$ are faulty, and they all crash right after time $t_r$.
    \item Message transmission among processes in $K_r$ is according to the latencies of $\mathcal{L}$, 
    except for any messages sent by any $p \in K_r$ at time $t_{p,r}$ or later, which are lost. 
    \item Message transmission among processes not in $K_r$ is according to the latencies of $\mathcal{L'}$.
    \item Messages sent from a process in $K_r$ to a process not in that set, or vice versa, are lost.
        
\end{itemize}

We next show that in $e_2$, $r$ decides $\camera{v_r}$ while $r'$ decides $\camera{v_{r'}}$, contradicting the agreement property of the consensus protocol Alg.
We start with $r$. By termination, $r$ decides (at time $t_r$) in $e_r$. By validity, it decides $\camera{v_r}$. By \Cref{corollary:e_2 indistinguishable from e_r to r}, $e_2 \stackrel{r}{\approx}_{t_r} e_r$, hence $r$ decides $\camera{v_r}$ (at time $t_r$) in $e_2$ as well.
As for $r'$, by termination it decides in $e_1$. By validity, it decides $\camera{v_{r'}}$.
By \Cref{lemma:e_2 indistinguishable from e_1 to r'}, $e_2 \stackrel{r'}{\approx} e_1$, hence $r'$ decides $\camera{v_{r'}}$ in $e_2$ as well.
\end{proof}

\begin{claim}\label{claim:optimality:compatible}
The combination of requirements $AlgReqs$ is compatible (see definition in Algorithm \ref{alg:ac:compatible}).
\end{claim}
\begin{proof}

Assume, for sake of contradiction, that there exist processes $r$, $r'$ such that for at most $f$ processes in $K_r$ or $K_{r'}$ (or both), the corresponding value (namely $AlgReqs[r][p]$ or $AlgReqs[r'][p]$ respectively for process $p$) contains some process in the intersection $K_r \cap K_{r'}$. Denote this set of at most $f$ processes (about which $r$ knows at $t_r$ in $e_r$ that they know about some process in the intersection, or $r'$ knows at $t_{r'}$ in $e_{r'}$ that they know about some process in the intersection) by \textcolor{cyan}{$K$}.
Let \textcolor{cyan}{$\overline{K_{r'}}$} = $K_r - K_{r'}$.
Let \textcolor{cyan}{$t_{p,K_r-K}$} be the time in which process $p \in \overline{K_{r'}}$ receives in $e_r$ its first message from another process that is not known by any process $q \in K_r-K$ before time $t_{q,r}$, or $t_{p,r}$ if no such message is received until $t_{p,r}$.
(Note that for $p \in K_r-K$, $t_{p,K_r-K} = t_{p,r}$.)
Let \textcolor{cyan}{$t_d$} \camerad{$= \max\{t_r, t_{r'}\} + \epsilon$}.
Further let \textcolor{cyan}{$\mathcal{L^*}$} indicate inter-process latencies for all pairs of processes not in $K$ in the following way: $\mathcal{L^*}[p,p] = 0$ for any $p \notin K$, and for every two processes $s,p \notin K$ such that $s \neq p$, $\mathcal{L^*}[s,p] = \mathcal{L}[s,p]$ if $s,p \notin \mathcal{F}$, else $\mathcal{L^*}[s,p]$ is assigned an arbitrary positive value.

We next design $f_1$ (illustrated in \Cref{fig:f1}) and $f_2$, which will satisfy the following:
(1) $f_1$ and $e_{r'}$ are indistinguishable to $r'$ until it decides, so that $r'$ decides $\camera{v_{r'}}$ in $f_1$.
(2) $f_2$ and $e_r$ are indistinguishable to $r$ until it decides, so that $r$ decides $\camera{v_r}$ in $f_2$.
(3) $f_1$ and $f_2$ are indistinguishable to all processes not in $K$, so that they decide (by termination) the same value in both executions. At the same time, by agreement, they should decide $\camera{v_{r'}}$ in $f_1$, but $\camera{v_r}$ in $f_2$, which yields the desired contradiction.

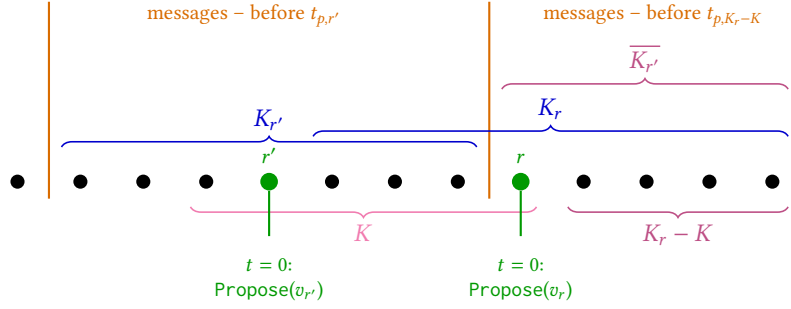
\begin{figure}
  \centering
  \resizebox{0.6\linewidth}{!}{\begin{tikzpicture}[
  x=1.25cm,y=1.25cm,
  font=\Large,
  dot/.style={circle,fill=black,inner sep=0pt,minimum size=2.8mm},
  prop/.style={circle,fill=green!60!black,inner sep=0pt,minimum size=3.5mm},
  brc/.style={decorate,decoration={brace,amplitude=6pt,raise=0pt},line width=0.9pt},
  brcd/.style={decorate,decoration={brace,mirror,amplitude=6pt,raise=0pt},line width=0.9pt},
  cut/.style={draw=orange!85!black,line width=1.2pt},
  set/.style={font=\LARGE},
]
\foreach \i in {0,1,2,3,5,6,7,9,10,11,12} {\node[dot] at (\i,0) {};}
\node[prop] (rp) at (4,0) {};
\node[prop] (r) at (8,0) {};
\node[green!60!black,above=2pt] at (rp.north) {$r'$};
\node[green!60!black,above=2pt] at (r.north) {$r$};

\draw[cut] (0.5,-0.27) -- (0.5,2.86);
\draw[cut] (7.5,-0.27) -- (7.5,2.86);
\node[orange!85!black] at (3.6,2.64) {messages -- before $t_{p,r'}$};
\node[orange!85!black] at (10.1,2.64) {messages -- before $t_{p,K_r-K}$};

\draw[brc,blue!80!black] (0.7,0.55) -- node[above=6pt,set] {$K_{r'}$} (7.3,0.55);
\draw[brc,blue!80!black] (4.7,0.72) -- node[above=6pt,set] {$K_r$} (12.25,0.72);
\draw[brc,magenta!75!black] (7.7,1.47) -- node[above=6pt,set] {$\overline{K_{r'}}$} (12.25,1.47);

\draw[brcd,magenta!65] (2.75,-0.4) -- node[below=6pt,set] {$K$} (8.25,-0.4);
\draw[brcd,magenta!75!black] (8.75,-0.4) -- node[below=6pt,set] {$K_r-K$} (12.25,-0.4);

\foreach \p in {rp,r} {\draw[green!60!black,line width=1.2pt] (\p.south) -- ++(0,-0.72);}
\node[green!60!black,align=center] at (4,-1.55) {$t=0$:\\\texttt{Propose}($v_{r'}$)};
\node[green!60!black,align=center] at (8,-1.55) {$t=0$:\\\texttt{Propose}($v_r$)};
\path (-0.25,0) (12.5,0);
\end{tikzpicture}}
  \caption{Illustration of execution $f_1$ until time $t_d$: Dots represent processes. Proposals are \camerad{shown} in green. The different sets that processes belong to are illustrated in blue or pink. The communication \camerad{between any process $p$ and the others} until time $t_d$ in $f_1$ is illustrated in orange.}
  \Description{\camerad{A line of processes partitioned into the overlapping sets $K_r$, $K_{r'}$, $K$, and their differences. Processes $r$ and $r'$ propose at time zero, and vertical markers show the message cutoffs used to construct execution $f_1$.}}
  \label{fig:f1}
\end{figure}

Intuitively, to achieve (1), we design $f_1$ as follows: we invoke a proposal of $\camera{v_{r'}}$ on $r'$ at time 0, and let any process $p \in K_{r'}$ receive from other processes messages known by $r'$ at $t_{r'}$ in $e_{r'}$ and only them until time $t_{r'}$. Specifically, $p$ only receives messages prior to $t_{p,r'}$ and messages $p$ sends starting at $t_{p,r'}$ are lost; and any messages among processes in $K_{r'}$ and other processes are discarded.
To achieve (2), we design $f_2$ symmetrically for $r$.

To achieve (3) for any process $p \in K_r-K$, since in $f_2$, until time $t_r$, $p$ receives from other processes messages known by $r$ at $t_r$ in $e_r$ and only them, in $f_1$ we also have $p$ receive from other processes messages known by $r$ at $t_r$ in $e_r$ and only them until time $t_r$. As $p \notin K$, $r$ does not have knowledge at $t_r$ in $e_r$ that $p$ knows about any process in the intersection $K_r \cap K_{r'}$.
This means that $p$ should not know in $f_1$ about processes in the intersection. We achieve that intuitively by having $p$, until time $t_r$ in $f_1$, 
receive from processes in $\overline{K_{r'}} \cap K$ only messages sent before they learned about the intersection, and not receive any messages from processes in $K_r \cap K_{r'}$.
More concretely, until time $t_r$ in $f_1$, we let any process $p \in \overline{K_{r'}}$ receive from other processes messages known by $r$ at $t_r$ in $e_r$ and only them, by enforcing the following: every process $p \in K_r-K$ only receives messages prior to $t_{p,r}$ and messages $p$ sends starting at $t_{p,r}$ are lost; every process $p \in \overline{K_{r'}} \cap K$ only receives messages prior to $t_{p,K_r-K}$ and messages $p$ sends starting at $t_{p,K_r-K}$ are lost; and any messages among processes in $\overline{K_{r'}}$ and other processes are discarded.
To achieve (3) for processes in $K_{r'}-K$, we design $f_2$ symmetrically for $r'$.

We now formally define $f_1$ and $f_2$.

Consider an execution \textcolor{cyan}{$f_1$} of Alg, in which:
\begin{itemize}
    \item Alg.\t{Propose($\camera{v_r}$)} is invoked on $r$ at time 0.
    \item Alg.\t{Propose($\camera{v_{r'}}$)} is invoked on $r'$ at time 0.
    \item Only the processes in $K$ are faulty, and they all crash right after time $t_d$.
    \item Message transmission of messages sent until time $t_d$ (including) is according to the latencies of $\mathcal{L}$, except for the following: (Intuitively, the first and third following items ensure that $f_1 \stackrel{p}{\approx}_{<t_{p,r'}} e_{r'}$ for every $p \in K_{r'}$ and, in particular, $f_1 \stackrel{r'}{\approx}_{t_{r'}} e_{r'}$;
    and the second and third items ensure that $f_1 \stackrel{p}{\approx}_{<t_{p,K_r-K}} e_r$ for every $p \in \overline{K_{r'}}$ and, in particular, $f_1 \stackrel{p}{\approx}_{<t_{p,r}} e_r$ for every $p \in K_r - K$.)
    \begin{itemize}
        \item 
        For each process $p \in K_{r'}$, any messages sent to $p$ before $t_d$, which do not reach it before $t_{p,r'}$ according to the latencies of $\mathcal{L}$, are lost. (This implies that no messages are received by $p$ in $f_1$ since time $t_{p,r'}$ until $t_d$ (including).)
        
        In addition, all messages sent by $p$ at time $t_{p,r'}$ or later until $t_d$ (including), are lost. 
        \item For each process $p \in \overline{K_{r'}}$, any messages sent to $p$ before $t_d$, which do not reach it before $t_{p,K_r-K}$ according to the latencies of $\mathcal{L}$, are lost. (This implies that no messages are received by $p$ in $f_1$ since time $t_{p,K_r-K}$ until $t_d$ (including).)
        
        In addition, all messages sent by $p$ at time $t_{p,K_r-K}$ or later until $t_d$ (including), are lost.
        \item Messages between two processes that are not both in $K_{r'}$ or both in $\overline{K_{r'}}$ are lost.
        \item Messages to processes in $K$ that should reach them after $t_d$ according to the latencies of $\mathcal{L}$, do not reach their target because it is down.
    \end{itemize}
    \item Message transmission of messages sent after time $t_d$ is according to the latencies of $\mathcal{L^*}$, except for messages to processes in $K$ which do not reach their target because it is down.
\end{itemize}

Let \textcolor{cyan}{$\overline{K_r}$} = $K_{r'} - K_r$.
Let \textcolor{cyan}{$t_{p,K_{r'}-K}$} be the time in which process $p \in \overline{K_r}$ receives in $e_{r'}$ its first message from another process that is not known by any process $q \in K_{r'}-K$ before time $t_{q,r'}$, or $t_{p,r'}$ if no such message is received until $t_{p,r'}$.
(Note that for $p \in K_{r'}-K$, $t_{p,K_{r'}-K} = t_{p,r'}$.)
Consider an execution \textcolor{cyan}{$f_2$} of Alg, which is symmetrical to $f_1$, obtained by exchanging the roles of $r$ and $r'$. In detail, in $f_2$:
\begin{itemize}
    \item Alg.\t{Propose($\camera{v_r}$)} is invoked on $r$ at time 0.
    \item Alg.\t{Propose($\camera{v_{r'}}$)} is invoked on $r'$ at time 0.
    \item Only the processes in $K$ are faulty, and they all crash right after time $t_d$.
    \item Message transmission of messages sent  until time $t_d$ (including) is according to the latencies of $\mathcal{L}$, except for the following: (Intuitively, the first and third following items ensure that $f_2 \stackrel{p}{\approx}_{<t_{p,r}} e_r$ for every $p \in K_r$ and, in particular, $f_2 \stackrel{r}{\approx}_{t_r} e_r$; 
    and the second and third items ensure that $f_2 \stackrel{p}{\approx}_{<t_{p,K_{r'}-K}} e_{r'}$ for every $p \in K_{r'} - K$.)
    \begin{itemize}
        \item For each process $p \in K_r$, any messages sent to $p$ before $t_d$, which do not reach it before $t_{p,r}$ according to the latencies of $\mathcal{L}$, are lost. (This implies that no messages are received by $p$ in $f_2$ at time $t_{p,r}$ or later until $t_d$ (including).)
        
        In addition, all messages sent by $p$ at time $t_{p,r}$ or later until $t_d$ (including), are lost.
        \item For each process $p \in \overline{K_r}$, any messages sent to $p$ before $t_d$, which do not reach it before $t_{p,K_{r'}-K}$ according to the latencies of $\mathcal{L}$, are lost. (This implies that no messages are received by $p$ in $f_2$ at time $t_{p,K_{r'}-K}$ or after until $t_d$ (including).)
        
        In addition, all messages sent by $p$ at time $t_{p,K_{r'}-K}$ or after until $t_d$ (including), are lost.
        \item Messages sent until time $t_d$ (including), 
        between two processes that are not both in $K_r$ or both in $\overline{K_r}$, are lost.
        \item Messages to processes in $K$ that should reach them after $t_d$ according to the latencies of $\mathcal{L}$, do not reach their target because it is down.
    \end{itemize}
    \item Message transmission of messages sent after time $t_d$ is according to the latencies of $\mathcal{L^*}$, except for messages to processes in $K$ which do not reach their target because it is down.
\end{itemize}

By \Cref{lemma:optimality:proposer not in other proposer's keys}, $r \notin K_{r'}$ and $r' \notin K_r$.
Additionally, by \Cref{lemma:optimality:proposer required to know about intersection}, $r,r' \in K$.
Consequently, $r \in K - K_{r'}$ and $r' \in K - K_r$.

We next show that in $f_2$, $r$ decides $\camera{v_r}$ while processes not in $K$ decide $\camera{v_{r'}}$, contradicting the agreement property of the consensus protocol Alg.
We start with $r$. By termination, $r$ decides (at time $t_r$) in $e_r$. By validity, it decides $\camera{v_r}$. It could be proven, just like proven in \Cref{corollary:e_2 indistinguishable from e_r to r} for $e_2$, that $f_2 \stackrel{r}{\approx}_{t_r} e_r$, hence $r$ decides $\camera{v_r}$ (at time $t_r$) in $f_2$ as well.
As for $r'$, by termination, $r'$ decides (at time $t_{r'}$) in $e_{r'}$. By validity, it decides $\camera{v_{r'}}$. It could be proven, just like proven in \Cref{corollary:e_2 indistinguishable from e_r to r} for $e_2$, that $f_1 \stackrel{r'}{\approx}_{t_{r'}} e_{r'}$, hence $r'$ decides $\camera{v_{r'}}$ (at time $t_{r'}$) in $f_1$ as well. 
By agreement and termination, all processes 
not in $K$
decide $\camera{v_{r'}}$ in $f_1$.
By \Cref{lemma:f_2 indistinguishable from f_1 to processes not in K}, $f_2 \stackrel{p}{\approx} f_1$ for any process $p \notin K$, hence these processes decide $\camera{v_{r'}}$ in $f_2$ as well.
\end{proof}

\begin{claim}\label{claim:optimality:KCensus_Alg not slower than Alg}
The weighted objective latency of Alg is not smaller than that of $\sysname_{AlgReqs}$.
\end{claim}
\begin{proof}
For every process $r \notin \mathcal{F}$, let \textcolor{cyan}{$e_r^*$} be a single-proposer execution of $\sysname_{AlgReqs}$ in which $\sysname_{AlgReqs}$.\t{Propose($\camera{v_r}$)} is invoked on $r$ at time 0, and the network is stable such that
processes in $\mathcal{F}$ are crashed, 
and message transmission among correct processes is according to the latencies of $\mathcal{L}$.
We next show that $r$ decides at time $\leq t_r$ in $e_r^*$. As $r$ decides at time $t_r$ in $e_r$, which is the equivalent execution of Alg, then the latency of Alg is not smaller than that of $\sysname_{AlgReqs}$ for $r$ as proposer. Taking the weighted objective latency over executions for each correct process as the proposer, we get QED.

By termination, $r$ decides in $e_r^*$.
Hence, the \t{Propose($\camera{v_r}$)} call invoked on $r$ must break from the wait on \cref{alg:ac:propose:until}.
Since no \t{DoAdopt} messages are received in $e_r^*$ by \Cref{no DoAdopt in e_r^*}, the \t{Propose($\camera{v_r}$)} call breaks from the wait on \cref{alg:ac:propose:until} only when \t{can\_commit($\camera{v_r}$)} returns \t{True}.
If \t{can\_commit($\camera{v_r}$)} by $r$ does not return \t{True} in $e_r^*$ before time $t_r$, then at $t_r$ it does by \Cref{lemma:optimality:can_commit returns true at t_r}.
This triggers a commit event in the adopt-commit protocol (\cref{alg:ac:propose:commit}), which in turn triggers a decision event in the consensus protocol on $r$ (\cref{alg:consensus:propose:decide upon commit}).
\end{proof}

Before presenting the main lemmas used in the proofs of the above claims, we bring several supporting observations and lemmas that will facilitate the proof of these main lemmas. These statements characterize the processes in $K_r$ as well as values in \t{known\_acceptors}.

\begin{observation}\label{observation:optimality:r in K_r}
For any process $r$, $r \in K_r$.  
\end{observation}
\begin{proof}
\t{known\_acceptors[$r$][$r$]} is initialized to a non-empty set on \cref{alg:optimality:init r}. This adds the key $r$ to \t{array\_to\_dict(known\_acceptors[$r$])}.
The dictionary \t{array\_to\_dict(known\_acceptors[$r$])} is only ever extended.   
As $AlgReqs[r]$ is assigned the value of \t{array\_to\_dict(k\-nown\_acceptors[$r$])} after completing the run of Algorithm~\ref{alg:optimality:knowledge}, $r \in AlgReqs[r]$\t{.keys()} $= K_r$.
\end{proof}

\begin{observation}\label{observation:optimality:message received before known received or sent - is known}
Consider a process $r$. A message $m$, received by some process $p$ in $e_r$ before $p$ receives or sends some message $m'$---which is known by $r$ at $t_r$, is also known by $r$ at $t_r$.
\end{observation}
\begin{proof}
Message $m'$ must be received by its target, as it is known by $r$.
Clearly, $m$'s receipt happened-before $m'$'s receipt. 
Additionally, as $m'$ is known by $r$ at $t_r$, $m'$ is received by $r$ up to time $t_r$, or the receipt of $m'$ happened-before the receipt of some message $m^*$ at $r$ occurring up to time $t_r$. 
In the first case, $m$'s receipt happened-before $m'$'s receipt by $r$ occurring at time $\leq t_r$, and in the latter case---$m$'s receipt happened-before $m^*$'s receipt by $r$ occurring at time $\leq t_r$. Thus, $m$ is known by $r$ at time $t_r$.
\end{proof}

Thanks to our assumptions that at a certain time a process $p$ may receive a single message from another process (possibly followed by messages from $p$ to itself); and that $p$ does not send messages right before receiving a message from another process at the exact same time---we may deduce the following:

\begin{corollary}\label{corollary:optimality:messages sent since t_pr are not known}
Consider a process $r$. Messages sent in $e_r$ by any process $p \in K_r$ at time $t_{p,r}$ or later are not known by $r$ at $t_r$.
\end{corollary}

\begin{lemma}\label{lemma:optimality:backward induction proved facts on known_acceptors and AlgReqs}
Consider some process $r$.
For any message known by $r$ at $t_r$, received by some process $p$ in $e_r$,
the following holds after processing the message in Algorithm~\ref{alg:optimality:knowledge}:

(1) $p \in$ \t{array\_to\_dict(known\_acceptors[$p$]).keys()}, 

(2) all of the keys in \t{array\_to\_dict(known\_acceptors[$p$])} are in $K_r$, and

(3) \t{known\_acceptors[$p$][$q$]} $\subseteq$ $AlgReqs[r][q]$ for any process $q$.
\end{lemma}
\begin{proof}
We prove by induction on messages in $e_r$ known by $r$ at $t_r$, ordered by the reverse of their processing order in Algorithm~\ref{alg:optimality:knowledge}.
For the base case, consider the last message $m$, known by $r$ at $t_r$, received in $e_r$ up to time $t_r$.
$m$ must be received by $r$, because if it is received by another process, its receipt must be followed by sending another message that is known by $r$ at $t_r$, which is received up to time $t_r$---contradicting the receipt of $m$ being last.
Condition (1) holds due to \cref{alg:optimality:init r}.
Conditions (2) and (3) hold because $AlgReqs[r]$ is set to the value of \t{array\_to\_dict(known\_acceptors[$r$])} after processing $m$.

Now suppose the statement holds for messages, known by $r$ at $t_r$, received after some message $m$, known by $r$ at $t_r$, is received by process $p$.
As $m \in M$ for $M$ defined on \cref{alg:optimality:known messages}, $m$ is processed in the for loop on \cref{alg:optimality:for loop}.
On \cref{alg:optimality:knownacceptorsme}, the key $p$ is added to the \t{array\_to\_dict(known\_acceptors[$p$])} dictionary (if it is not yet there), hence condition (1) holds.

As for conditions (2) and (3): if $p=r$, then since all processes in \t{array\_to\_dict(known\_acceptors[$r$]).keys()} after processing $m$ remain in it until $t_r$, they are in $K_r$ and condition (2) holds. Similarly, for any process $q$, all processes in \t{known\_acceptors[$r$][$q$]} after processing $m$ remain in it until $t_r$, hence, they are in $AlgReqs[r][q]$ and condition (3) holds.
Else ($p \neq r$),
as $m$ is a message known by $r$ at $t_r$, then after receiving $m$, $p$ must send to some process $p^*$ some message $m^*$, known by $r$ at $t_r$. 
We start with condition (2):
All processes that were in \t{array\_to\_dict(known\_acceptors[$p$]).keys()} after $m$ is processed in Algorithm~\ref{alg:optimality:knowledge}, remain there until right after processing the last message in $M$ received by $p$ prior to sending $m^*$.
After processing $m^*$, all processes that were in \t{array\_to\_dict(known\_acceptors[$p$]).keys()} after processing $m$, are in \t{array\_to\_dict(known\_acceptors[$p^*$]).keys()}, due to \crefrange{alg:optimality:inner for loop}{alg:optimality:knownacceptorsq}.
By the inductive hypothesis, all processes that were in \t{array\_to\_dict(known\_acceptors[$p^*$]).keys()} after processing $m^*$ are in $K_r$.
Similarly, for condition (3):
All processes that were in \t{known\_acceptors[$p$][$q$]} for any process $q$ after processing $m$ remain there until right after processing the last message in $M$ received by $p$ prior to sending $m^*$. 
After processing $m^*$, all processes that were in \t{known\_acceptors[$p$][$q$]} for any process $q$ after processing $m$, are in \t{known\_acceptors[$p^*$][$q$]}, due to \crefrange{alg:optimality:inner for loop}{alg:optimality:knownacceptorsq}.
By the inductive hypothesis, all processes that were in \t{known\_acceptors[$p^*$][$q$]} after processing $m^*$ are in $AlgReqs[r][q]$.
\end{proof}

\begin{lemma}\label{lemma:optimality:only processes in K_r send and receive known messages in e_r}
Consider some process $r$. Only processes in $K_r$ receive and send in $e_r$ messages known by $r$ at $t_r$.
\end{lemma}
\begin{proof}
We start with the case of a process $p$ that receives a message known by $r$ at $t_r$.
From conditions (1) and (2) in \Cref{lemma:optimality:backward induction proved facts on known_acceptors and AlgReqs},
it specifically follows that $p \in K_r$.

We proceed to the second case, of a process $s$ sending in $e_r$ a message $m$ known by $r$ at $t_r$.
If $s=r$, then $s \in K_r$ by \Cref{observation:optimality:r in K_r}.
Else, due to the assumption that a \t{Propose} invocation happened-before any message sending in Alg, and the fact that $r \neq s$ is the only proposer in $e_r$, $s$ receives some message $m^*$ before sending $m$.
By \Cref{observation:optimality:message received before known received or sent - is known}, $m^*$ is known by $r$ at $t_r$. Then by the first case above ($p=s$ receives a message $m^*$ known by $r$ at $t_r$), $s \in K_r$.
\end{proof}

The following lemmas are used in the proof of \Cref{claim:optimality:valid}.

\begin{lemma}\label{lemma:e_2 indistinguishable from e_r to p until last received - known in e_r or in e_2 if p in K_r}
Consider $r$ chosen in the proof of \Cref{claim:optimality:valid} and $e_2$ defined in that proof.
$e_2 \stackrel{p}{\approx}_{< t_{p,r}} e_r$ for any process $p \in K_r$.
\end{lemma}
\begin{proof}
We prove that $e_2 \stackrel{p}{\approx}_{t} e_r$ for any time $t$ such that a message is received by some process $p \in K_r$ at time $t < t_{p,r}$ in $e_r$ or $e_2$. This is enough, as $e_2$ and $e_r$ remain indistinguishable to $p$ from the last time $p$ receives such a message up to time $t_{p,r}$ (since during this time frame, there are no additional message receipts by $p$ in $e_2$ or $e_r$, nor any different proposals at $p$).
We prove by induction on message receipt times.
As \t{Propose} invocations on processes in $K_r$ are the same in $e_2$ and $e_r$, it remains to show that each process $p \in K_r$ receives the same messages at the same times in both executions until any time $t < t_{p,r}$ (including $t$) in which it receives a message in one of the executions.

For the base case, we consider time 0.
The base case trivially holds for all processes in $K_r$ except for $r$: 
Due to the assumption that a \t{Propose} invocation happened-before any message sending in Alg, and the fact that proposals are invoked only at time 0 in $e_2$ and $e_r$, and only on $r$ and $r' \notin K_r$,
no messages are received by processes $\neq r$ in $K_r$ up to time 0. Thus, $e_2$ and $e_r$ are indistinguishable to these processes at time 0. As for $r$, it might receive messages from itself at time 0. In that case, $e_2$ and $e_r$ are indistinguishable to $r$ at time 0 before it receives the first message from itself. Hence, $r$ sends the same message to itself in both executions, then receives it in both executions, and they remain indistinguishable to $r$ after receiving the message. The same argument applies to any subsequent messages $r$ sends itself at time 0. Therefore, $e_2$ and $e_r$ are still indistinguishable to $r$ at time 0 after receiving all of these self messages.

Suppose the statement holds prior to time $0 < t < t_{p,r}$ for some process $p \in K_r$, and $p$ receives only messages from itself at time $t$ in $e_r$ and $e_2$.
Then by the same arguments as in the base case together with the inductive hypothesis, these messages are received in both executions and leave the executions indistinguishable to $p$. 
So suppose the statement holds prior to time $0 < t < t_{p,r}$ for some process $p \in K_r$, and $p$ receives a message $m$ at time $t$ in $e_r$ or $e_2$ from another process.
Note that in case $m$ is received in $e_r$, then
since $p$ receives its first message from another process that is not known by $r$ at $t_r$ (if any) at time $t_{p,r}$, $m$ is known by $r$ at $t_r$.
For messages received by $p$ in $e_2$ or $e_r$ before time $t$, we may apply the inductive hypothesis, which yields that $e_2$ is indistinguishable to $p$ from $e_r$ until the receipt time of any such message.
Hence, $e_2$ and $e_r$ are indistinguishable to $p$ 
prior to $t$.
To show that $e_2$ and $e_r$ are indistinguishable to $p$ until $t$ (including), we will show that $m$ is received in both executions at time $t$. This is enough as by assumption, no message other than $m$ is received by $p$ at $t$ from other processes;  
and messages $p$ possibly receives from itself after $m$ at time $t$ 
are received in both executions and leave the executions indistinguishable to $p$, by the same arguments as in the base case. 

Let $s$ be the sender of $m$.
To apply the inductive hypothesis to messages received by $s$ in either execution until the time in which $s$ sends $m$, we need to show they satisfy the conditions of the induction's statement: $s \in K_r$, and these messages are received before time $t_{s,r}$. For the latter, we show the stronger statement that $s$ sends $m$ before time $t_{s,r}$. We start with the case in which $m$ is received in $e_r$. 
In this case, as $m$ is known by $r$ at $t_r$ (as noted above), then $s \in K_r$ by \Cref{lemma:optimality:only processes in K_r send and receive known messages in e_r}.
Additionally, 
as $m$ is known by $r$ at $t_r$, then
by \Cref{corollary:optimality:messages sent since t_pr are not known}, $s$ sends $m$ before time $t_{s,r}$.
As for the case that $m$ is received in $e_2$, $p \in K_r$ receives only messages from processes in $K_r$ in $e_2$, hence $s \in K_r$. In addition, $m$ is sent before $t_{s,r}$, as any messages $s$ sends in $e_2$ at $t_{s,r}$ or later are lost.

So by the inductive hypothesis, $e_2$ and $e_r$ are indistinguishable to $s$ until receiving the last message (if any) in either execution up to the time in which it sends $m$. They hence remain indistinguishable to $s$ from that moment (or from time 0 in case no message receipts by $s$ precede the sending of $m$) until it sends $m$ (in $e_r$ or $e_2$). Therefore, $m$ is sent in the other execution as well at the same time. 
Given that $m$ originally appears in $e_2$, it is not lost and reaches its target in $e_r$ as well, since no messages to correct processes are lost in $e_r$.
Given that $m$ originally appears in $e_r$, it is not lost and reaches its target in $e_2$ as well, since $m$ is a message from a process in $K_r$ to a process in $K_r$, sent by $s$ before $t_{s,r}$.
Furthermore, $m$ reaches its target at the same time in both executions---according to the corresponding latency in $\mathcal{L}$.
That is, $m$ is received in both executions at time $t$ and we are done.
\end{proof}

\begin{corollary}\label{corollary:e_2 indistinguishable from e_r to r}
$e_2 \stackrel{r}{\approx}_{t_r} e_r$,
for $r$ chosen in the proof of \Cref{claim:optimality:valid} and $e_2$ defined in that proof.
\end{corollary}
\begin{proof}
$r \in K_r$ by \Cref{observation:optimality:r in K_r}.
Additionally, $t_{r,r} = t_r + \epsilon$ by definition of $t_{r,r}$.
Thus, applying \Cref{lemma:e_2 indistinguishable from e_r to p until last received - known in e_r or in e_2 if p in K_r} to $r$ yields QED.
\end{proof}

\begin{lemma}\label{lemma:e_2 indistinguishable from e_1 to r'}
$e_2 \stackrel{r'}{\approx} e_1$,
for $r'$ chosen in the proof of \Cref{claim:optimality:valid} and $e_1$ and $e_2$ defined in that proof.
\end{lemma}
\begin{proof}
We prove that $e_2 \stackrel{p}{\approx} e_1$ for each $p \notin K_r$.
We start with \t{Propose} invocations. These are the same in $e_2$ and $e_1$ on all processes except for $r$, but $r \in K_r$ by \Cref{observation:optimality:r in K_r}.

We proceed to message receipts.
In both executions---processes not in $K_r$ do not receive messages from processes in $K_r$ (up to time $t_r$ they are lost in $e_2$ if sent, in other times the processes in $K_r$ are down in both executions), and all messages among processes not in $K_r$ are transmitted according to the latencies of $\mathcal{L'}$.
Therefore, all processes not in $K_r$ cannot distinguish between $e_2$ and $e_1$, and communicate only among themselves (namely, with processes that do not distinguish as well). As Alg is deterministic, the two executions are indistinguishable to processes not in $K_r$ at all times.
\end{proof}

The following lemmas are used in the proof of \Cref{claim:optimality:compatible}.

\begin{lemma}\label{lemma:optimality:source in known_acceptors[p][p] after handling received known message}
Consider some processes $r$ and $x$. Assume process $p$ receives in $e_r$ a message $m$, known by $r$ at $t_r$. Further assume either $m$ is sent by $x$ or a sending of some message by $x$ happened-before the sending of $m$.
Then $x \in$ \t{known\_acceptors[$p$][$p$]} right after handling $m$ in Algorithm~\ref{alg:optimality:knowledge}.
\end{lemma}
\begin{proof}
Consider some processes $r$ and $x$.
We prove by induction on the following messages, by the order they are processed in Algorithm~\ref{alg:optimality:knowledge}: any message $m$ in $e_r$, known by $r$ at $t_r$, received by some process $p$, such that either $m$ is sent by $x$ or a sending of some message by $x$ happened-before the sending of $m$.

For the base case, we consider any message $m^*$, known by $r$ at $t_r$, sent by $x$ to some process $p$.
If $x \neq r$, then before sending $m^*$, $x$ must have received a message, due to the assumption that a \t{Propose} invocation happened-before any message sending in Alg. This message is known by $r$ at $t_r$ by \Cref{observation:optimality:message received before known received or sent - is known}. Thus, it is processed in Algorithm~\ref{alg:optimality:knowledge}. When it is processed, $x$ is added to \t{known\_acceptors[$x$][$x$]} (if not already there) on \cref{alg:optimality:knownacceptorsme}.
If $x = r$, \t{known\_acceptors[$r$][$r$]} is initialized to $\{r\}$ on \cref{alg:optimality:init r}. 
Whether $x=r$ or not, when $m^*$ is processed in Algorithm~\ref{alg:optimality:knowledge}, $x \in$ \t{knowledge\_of\_s[$x$]}, and is thus added to \t{known\_acceptors[$p$][$p$]} (if not already there) on \cref{alg:optimality:knownacceptorsme}.

Suppose the statement holds until some process $p$ receives from process $s$ a message $m$ known by $r$ at $t_r$, such that $s \neq x$ and a sending of some message by $x$ happened-before the sending of $m$.
Then before sending $m$, $s$ must have received some message $m^*$, such that either $m^*$ is sent by $x$ or a sending of some message by $x$ happened-before the sending of $m^*$.
As $m$ is known by $r$ at $t_r$, then by \Cref{observation:optimality:message received before known received or sent - is known}, $m^*$ is a message known by $r$ at $t_r$. Consequently, we may apply the inductive hypothesis to $m^*$ and conclude that \t{known\_acceptors[$s$][$s$]} contains $x$ after handling $m^*$.
Hence, when processing $m$ in Algorithm~\ref{alg:optimality:knowledge}, \t{knowledge\_of\_s[$s$]} contains $x$. Thus, $x$ is added to \t{known\_acceptors[$p$][$p$]} (if it was not already there) on \cref{alg:optimality:knownacceptorsme}.
\end{proof}

\begin{lemma}\label{lemma:optimality:r in known_acceptors[p][p] after handling received known message}
For any process $r$ and any process $p$ that receives in $e_r$ a message $m$ known by $r$ at $t_r$, 
$r \in$ \t{known\_acceptors[$p$][$p$]} right after handling $m$ in Algorithm~\ref{alg:optimality:knowledge}.
\end{lemma}
\begin{proof}
Due to the assumption that a \t{Propose} invocation happened-before any message sending in Alg, either $m$ is sent by $r$ or a message sending by $r$ happened-before the sending of $m$. Hence, we may apply \Cref{lemma:optimality:source in known_acceptors[p][p] after handling received known message} with $x=r$ to $m$. This yields QED.

\end{proof}

\begin{lemma}\label{lemma:optimality:known_acceptors[p].keys() subseteq known_acceptors[p][p] after handling received known message}
For any process $r$ and any process $p$ that receives in $e_r$ a message $m$ known by $r$ at $t_r$, 
\t{array\_to\_dict(known\_accep\-tors[$p$]).keys()} $\subseteq$ \t{known\_acceptors[$p$][$p$]} right after handling $m$ in Algorithm~\ref{alg:optimality:knowledge}.
\end{lemma}
\begin{proof}
Let $r$ be some process.
We prove by induction on messages in $e_r$ known by $r$ at $t_r$, by the order they are processed in Algorithm~\ref{alg:optimality:knowledge}.

For the base case, before any message is handled, all sets in \t{known\_acceptors} except for \t{known\_acceptors[$r$][$r$]} are empty, hence the transformation of \t{known\_acceptors[$p$]} for each process $p \neq r$ into a dictionary yields an empty dictionary. As for $r$, the only key in \t{array\_to\_dict(known\_acceptors[$r$])} after initialization is $r$, and \t{known\_acceptors[$r$][$r$]} = $\{r\}$.

Suppose the statement holds until some process $t$ receives from process $s$ a message $m$ known by $r$ at $t_r$.
Let \t{knowledge\_of\_s} be the value of \t{known\_acceptors[$s$]} right after processing the last message, known by $r$ at $t_r$, received by $s$ prior to sending $m$ (or its initial value if no such message exists).
By the inductive hypothesis, right before processing $m$ in Algorithm~\ref{alg:optimality:knowledge}, \t{array\_to\_dict(known\_acceptors[$t$]).keys()} $\subseteq$ \t{known\_acceptors[$t$][$t$]}. 
The keys that are added to \t{array\_to\_dict(known\_accep\-tors[$t$])} when processing $m$ (if they are not already there) are $t$ on \cref{alg:optimality:knownacceptorsme}, and the keys in \t{array\_to\_dict(knowledge\_of\_s)} on \cref{alg:optimality:knownacceptorsq}.
$t$ is added to \t{known\_acceptors[$t$][$t$]} (if not already there) on \cref{alg:optimality:knownacceptorsme}. \t{array\_to\_dict(knowledge\_of\_s).keys()} $\subseteq$ \t{knowledge\_of\_s[$s$]} by the inductive hypothesis. Hence, the keys in \t{array\_to\_dict(knowledge\_of\_s)} are also added to \t{known\_accep\-tors[$t$][$t$]} (if not already there) on \cref{alg:optimality:knownacceptorsme}. 
\end{proof}

\begin{lemma}\label{lemma:optimality:processes in K_r receive a known message}
For any process $r$, every process in $K_r$ receives in $e_r$ some message which is known by $r$ at $t_r$.
\end{lemma}
\begin{proof}
Let $r$ be some process.
We prove by induction on messages in $e_r$ known by $r$ at $t_r$, by the order they are processed in Algorithm~\ref{alg:optimality:knowledge}, that for every process $p$ that appears in \t{array\_to\_dict(known\_acceptors[$t$]).keys()} after handling in Algorithm~\ref{alg:optimality:knowledge} a message $m$, known by $r$ at $t_r$, received by some process $t$, the following holds: if $p \neq t$, then \camerad{a message-sending event by $p$} happened-before $m$'s receipt in $e_r$ (namely, intuitively, there is a message chain from $p$ to $t$ ending with $m$).
For any process $p \neq r$ in $K_r$, applying the induction's statement to $m$---the last message received by $t=r$ up to time $t_r$---yields QED: 
By the induction's statement, a sending of some message $m'$ by $p$ happened-before $m$'s receipt by $r$ in $e_r$, thus, $m'$ is known by $r$ at $t_r$.
Due to the assumption that a \t{Propose} invocation happened-before any message sending in Alg, and since a \t{Propose} is not invoked on $p \neq r$ in $e_r$, $p$ received a message prior to sending $m'$. This received message is known by $r$ at $t_r$ by \Cref{observation:optimality:message received before known received or sent - is known}, and we are done. 
As for process $r$, it may \camerad{decide} in $e_r$ only after receiving some message (otherwise the execution is indistinguishable to $r$ from one in which another process communicates with all others before $r$ does and \camerad{decides} another \camera{value}).

We proceed with the proof by induction.
For the base case, consider a message $m$ from process $s$, which is known by $r$ at $t_r$, received by process $t$.
Let $p \neq t$ be a process that appears in \t{array\_to\_dict(known\_acceptors[$t$]).keys()} after processing $m$ in Algorithm~\ref{alg:optimality:knowledge}, such that $p=s$ or $p=r$. We need to show that \camerad{a message-sending event by $p$} happened-before $m$'s receipt in $e_r$.
If $p=s$ then we are trivially done. 
If $p=r$ then we are also done by assumption that a \t{Propose} invocation happened-before any message sending in Alg (including the sending of $m$). 

Suppose the statement holds for all messages received before a message $m$ from process $s$, which is known by $r$ at $t_r$, is received by process $t$.
Let $p \neq t,s,r$ be a process that appears in \t{array\_to\_dict(known\_acceptors[$t$]).keys()} after processing $m$ in Algorithm~\ref{alg:optimality:knowledge}. We need to show that \camerad{a message-sending event by $p$} happened-before $m$'s receipt in $e_r$.
If $p$ was in \t{array\_to\_dict(known\_acceptors[$t$]).keys()} before processing $m$, then we are done by the inductive hypothesis.
Otherwise, $p$ was added to \t{array\_to\_dict(known\_acceptors[$t$]).keys()} when processing $m$ on \cref{alg:optimality:knownacceptorsq} in Algorithm~\ref{alg:optimality:knowledge}.
This means that $p$ was in \t{array\_to\_dict(known\_acceptors[$s$]).keys()} right after processing the last message, known by $r$ at $t_r$, received by $s$ prior to sending $m$; or $p$ is in the initial value of \t{array\_to\_dict(known\_acceptors[$s$]).keys()} if no such message exists. The latter is impossible:
the initial value is an empty dictionary if $s \neq r$, else a dictionary that contains $r$ as its only key---and $p \neq r$ cannot be a key in any of these dictionaries.
Namely, before sending $m$, $s$ must have received some message.
Let $m^*$ be the last message that $s$ received before sending $m$.
By \Cref{observation:optimality:message received before known received or sent - is known}, $m^*$ is a message known by $r$ at $t_r$.
$p$, which was added to \t{array\_to\_dict(known\_acceptors[$t$]).keys()} on \cref{alg:optimality:knownacceptorsq} when processing $m$, was in \t{array\_to\_dict(known\_acceptors[$s$]).keys()} right after processing $m^*$.
Applying the inductive hypothesis to $m^*$ yields that since $p \neq s$, then \camerad{a message-sending event by $p$} happened-before $m^*$'s receipt by $s$. Consequently, \camerad{a message-sending event by $p$} happened-before $m$'s receipt by $t$.
\end{proof}

\begin{lemma}\label{lemma:optimality:r in AlgReqs[r][p] for p in K_r}
For any process $r$ and process $p \in K_r$,
$r \in AlgReqs[r][p]$.
\end{lemma}
\begin{proof}
Consider some process $r$.
Let $p$ be a process in $K_r$. By \Cref{lemma:optimality:processes in K_r receive a known message}, $p$ receives in $e_r$ a message $m$ known by $r$ at $t_r$.
Then by \Cref{lemma:optimality:r in known_acceptors[p][p] after handling received known message},
$r \in$ \t{known\_acceptors[$p$][$p$]} right after handling $m$ in Algorithm~\ref{alg:optimality:knowledge}.
By condition (3) in \Cref{lemma:optimality:backward induction proved facts on known_acceptors and AlgReqs}, \t{known\_acceptors[$p$][$p$]} $\subseteq$ $AlgReqs[r][p]$ after handling $m$. 
Therefore, $r \in AlgReqs[r][p]$.
\end{proof}

\begin{lemma}\label{lemma:optimality:proposer not in other proposer's keys}
$r \notin K_{r'}$ and $r' \notin K_r$, 
for $r$ and $r'$ chosen in the proof of \Cref{claim:optimality:compatible}.
\end{lemma}
\begin{proof}
Assume, for sake of contradiction, that $r \in K_{r'}$. 
As $r \in K_r$ by \Cref{observation:optimality:r in K_r}, we get $r \in K_r \cap K_{r'}$. For every process $p \in K_r$, $r \in AlgReqs[r][p]$ by \Cref{lemma:optimality:r in AlgReqs[r][p] for p in K_r}. Namely, $AlgReqs[r][p]$ contains $r$---a process in the intersection $K_r \cap K_{r'}$.
This implies that $K_r \subseteq K$ (for $K$ defined in the proof of \Cref{claim:optimality:compatible}).
$|K_r| > f$ by validity (\Cref{claim:optimality:valid}), so $|K| > f$---contradiction.

The proof for $r' \notin K_r$ is symmetrical.
\end{proof}

\begin{lemma}\label{lemma:optimality:proposer required to know about intersection}
$r \in K$,
for $r$ chosen in the proof of \Cref{claim:optimality:compatible} and $K$ defined in that proof.
\end{lemma}
\begin{proof}
Let $p$ be a process in $K_r \cap K_{r'}$ for $r$ and $r'$ chosen in the proof of \Cref{claim:optimality:compatible}. 
\camera{By condition (3) of \Cref{claim:optimality:reflexive graph},
$K_r=AlgReqs[r][r]$.}
We have $r \in K_r$ by \Cref{observation:optimality:r in K_r}, and $p \in K_r \cap K_{r'}$ satisfies $p \in AlgReqs[r][r]$. 
Hence, by definition of $K$, $r \in K$.
\end{proof}

\begin{lemma}\label{lemma:f_1 indistinguishable from e_r to p until last received - known in e_r or in f_1 if p in K_r}
Consider $r$ chosen in the proof of \Cref{claim:optimality:compatible}, and $f_1$, $K$, $\overline{K_{r'}}$ and $t_{p,K_r-K}$ defined in that proof.
$f_1 \stackrel{p}{\approx}_{< t_{p,K_r-K}} e_r$ for any process $p \in \overline{K_{r'}}$.
\end{lemma}
\begin{proof}
We prove that $f_1 \stackrel{p}{\approx}_{t} e_r$ for any time $t$ such that a message is received in $e_r$ or $f_1$ by some process $p \in \overline{K_{r'}}$ at time $t < t_{p,K_r-K}$. This is enough, as $f_1$ and $e_r$ remain indistinguishable to $p$ from the last time $p$ receives such a message up to time $t_{p,K_r-K}$ (since during this time frame, there are no additional message receipts by $p$ in $f_1$ or $e_r$, nor any different proposals at $p$).
We prove by induction on message receipt times.
As \t{Propose} invocations on processes in $\overline{K_{r'}}$ are the same in $f_1$ and $e_r$, it remains to show that each $p$ receives the same messages at the same times in both executions until time $t$ (including).

For the base case, we consider time 0.
The base case trivially holds for all processes in $\overline{K_{r'}}$ except for $r$: 
Due to the assumption that a \t{Propose} invocation happened-before any message sending in Alg, and the fact that proposals are invoked only at time 0 in $f_1$ and $e_r$, and only on $r$ and $r' \notin \overline{K_{r'}}$,
no messages are received by processes $\neq r$ in $\overline{K_{r'}}$ up to time 0. Thus, $f_1$ and $e_r$ are indistinguishable to these processes at time 0. As for $r$, it might receive messages from itself at time 0. In that case, $f_1$ and $e_r$ are indistinguishable to $r$ at time 0 before it receives the first message from itself. Hence, $r$ sends the same message to itself in both executions, then receives it in both executions, and they remain indistinguishable to $r$ after receiving the message. The same argument applies to any subsequent messages $r$ sends itself at time 0. Therefore, $f_1$ and $e_r$ are still indistinguishable to $r$ at time 0 after receiving all of these self messages.

Suppose the statement holds prior to time $0 < t < t_{p,K_r-K}$ where $p$ is a process in $\overline{K_{r'}}$, and $p$ receives only messages from itself at time $t$ in $e_r$ and $f_1$.
Then by the same arguments as in the base case together with the inductive hypothesis, these messages are received in both executions and leave the executions indistinguishable to $p$. 
So suppose the statement holds prior to time $0 < t < t_{p,K_r-K}$ for some process $p \in \overline{K_{r'}}$, and $p$ receives a message $m$ at time $t$ in $e_r$ or $f_1$ from another process.

Note that in case $m$ is received in $e_r$, it is known by $r$ at $t_r$: 
If $p \in K_r-K$, then $p$ receives $m$ in $e_r$ before $t_{p,K_r-K}=t_{p,r}$, hence, by definition of $t_{p,r}$, $m$ is known by $r$ at $t_r$.
Else, there exists a process $q \in K_r-K$ that knows about $m$ in $e_r$ before time $t_{q,r}$. This means that 
$m$'s receipt happened-before $q$'s receipt of some message $m_q$ from another process, which occurred before time $t_{q,r}$. 
By definition of $t_{q,r}$, $m_q$ is known by $r$ at $t_r$, which means that $m_q$'s receipt happened-before some message receipt at $r$ occurring up to time $t_r$.
Overall, $m$'s receipt happened-before some message receipt at $r$ occurring up to time $t_r$. This means that $m$ is known by $r$ at $t_r$.

For messages received by $p$ in $f_1$ or $e_r$ before time $t$, we may apply the inductive hypothesis, which yields that $f_1$ is indistinguishable to $p$ from $e_r$ until the receipt time of any such message.
Hence, $f_1$ and $e_r$ are indistinguishable to $p$ 
prior to $t$.
To show that $f_1$ and $e_r$ are indistinguishable to $p$ until $t$ (including), we will show that $m$ is received in both executions at time $t$. This is enough as by assumption, no message other than $m$ is received by $p$ at $t$ from other processes;  
and messages $p$ possibly receives from itself after $m$ at time $t$ 
are received in both executions and leave the executions indistinguishable to $p$, by the same arguments as in the base case. 

Let $s$ be the sender of $m$.
To apply the inductive hypothesis to messages received by $s$ in either execution until the time in which $s$ sends $m$, we need to show they satisfy the conditions of the induction's statement: $s \in \overline{K_{r'}}$, and these messages are received before time $t_{s,K_r-K}$.
For the latter, we show the stronger statement that $s$ sends $m$ before time $t_{s,K_r-K}$.

We start with the case in which $m$ is received in $e_r$. 
As $m$ is known by $r$ at $t_r$ (as shown above), $s \in K_r$ by \Cref{lemma:optimality:only processes in K_r send and receive known messages in e_r}. We show that $s \notin K_r \cap K_{r'}$, and then $s \in K_r - (K_r \cap K_{r'}) = \overline{K_{r'}}$ will follow.
As $m$ is received by $p \neq s$ before $t_{p,K_r-K}$, there exists a process $q \in K_r-K$ that knows in $e_r$ about $m$ before time $t_{q,r}$. 
Intuitively, $s \notin K_r \cap K_{r'}$ is true because $q \notin K$, so when $r$ decides in $e_r$ it does not know that $q$ knows about any message sent by a process in the intersection.
Formally, as $q$ knows about $m$ before $t_{q,r}$, there exists a message $m^*$, received by $q$ before time $t_{q,r}$, such that either $q=p$ and $m^*=m$, or $q \neq p$ and $m$'s receipt (and hence sending) happened-before the sending of $m^*$. In both cases, $m^*$ is known by $r$ at $t_r$ as it is received by $q$ before time $t_{q,r}$ from another process.
Then by \Cref{lemma:optimality:source in known_acceptors[p][p] after handling received known message}, $s \in \t{known\_acceptors}[q][q]$ right after handling $m^*$ in Algorithm~\ref{alg:optimality:knowledge}.
Additionally, by condition (3) in \Cref{lemma:optimality:backward induction proved facts on known_acceptors and AlgReqs}, after handling $m^*$, \t{known\_acceptors[$q$][$q$]} $\subseteq$ $AlgReqs[r][q]$. Therefore, $s \in AlgReqs[r][q]$.
But since $q \notin K$, and $q \in K_r$, then by definition of $K$:
$AlgReqs[r][q] \cap (K_r \cap K_{r'}) = \emptyset$. 
Therefore, $s \notin K_r \cap K_{r'}$. So indeed $s \in \overline{K_{r'}}$.
We still need to show that $s$ sends $m$ before time $t_{s,K_r-K}$. If $t_{s,K_r-K} = t_{s,r}$, 
then as $m$ is known by $r$ at $t_r$ (as shown above), $s$ sends $m$ before time $t_{s,K_r-K}$ by \Cref{corollary:optimality:messages sent since t_pr are not known}.
Else, by definition of $t_{s,K_r-K}$, $s$ receives from another process a message $m_s$ at time $t_{s,K_r-K}$ that is not known by any process $q' \in K_r-K$ before time $t_{q',r}$. Assume for sake of contradiction that $s$ sends $m$ at time $\geq t_{s,K_r-K}$. By our assumption, $s$ does not send messages right before receiving a message from another process at the exact same time. Namely, $s$ sends $m$ after receiving $m_s$. As $m$ is known by $q$ before time $t_{q,r}$, then $m_s$ is as well, which is a contradiction. So in any case, $s$ sends $m$ before time $t_{s,K_r-K}$.

We proceed to the case that $m$ is received in $f_1$. $m$ is received before time $t_{p,K_r-K} \leq t_{p,r} \leq t_r + \epsilon \leq t_d$ in $f_1$.
Then $m$ is sent until time $t_d$.
Messages sent until $t_d$ to process $p \in \overline{K_{r'}}$ in $f_1$ are received by it only if they are from processes in $\overline{K_{r'}}$.
Therefore, $s \in \overline{K_{r'}}$.
In addition, $m$ is sent before $t_{s,K_r-K}$, as any messages $s \in \overline{K_{r'}}$ sends in $f_1$ until time $t_d$ at $t_{s,K_r-K}$ or later are lost.

We conclude that we may apply the inductive hypothesis to messages received by $s$ until the time in which it sends $m$.
By the inductive hypothesis, $f_1$ and $e_r$ are indistinguishable to $s$ until receiving the last message (if any) in either execution up to the time in which it sends $m$. They hence remain indistinguishable to $s$ from that moment (or from time 0 in case no message receipts by $s$ precede the sending of $m$) until it sends $m$ (in $e_r$ or $f_1$). Therefore, $m$ is sent in the other execution as well at the same time.
Given that $m$ originally appears in $f_1$, it is not lost and reaches its target in $e_r$ as well, since no messages to correct processes are lost in $e_r$.
Given that $m$ originally appears in $e_r$, it is not lost and reaches its target in $f_1$ as well, since $m$ is a message from a process in $\overline{K_{r'}}$ to a process in $\overline{K_{r'}}$, reaching its target $p$ at $t < t_{p,K_r-K}$.
Furthermore, $m$ reaches its target at the same time in both executions---according to the corresponding latency in $\mathcal{L}$.
That is, $m$ is received in both executions at time $t$ and we are done.
\end{proof}

\begin{lemma}\label{lemma:f_1 indistinguishable from f_2 to processes in K_r-K}
For every $p \in K_r-K$: $f_1 \stackrel{p}{\approx}_{t_d} f_2$,
for $r$ chosen in the proof of \Cref{claim:optimality:compatible}, and $t_d$, $K$, $f_1$ and $f_2$ defined in that proof.
\end{lemma}
\begin{proof}
Let $p$ be a process in $K_r-K$.
By \Cref{lemma:f_1 indistinguishable from e_r to p until last received - known in e_r or in f_1 if p in K_r}, $f_1 \stackrel{p}{\approx}_{<t_{p,r}} e_r$.
In addition, it could be proven for $f_2$, similarly to the proof of \Cref{lemma:e_2 indistinguishable from e_r to p until last received - known in e_r or in e_2 if p in K_r} for $e_2$, that $f_2 \stackrel{p}{\approx}_{< t_{p,r}} e_r$.

Overall, we get $f_1 \stackrel{p}{\approx}_{<t_{p,r}} f_2$.
As $p$ receives no messages at $t_{p,K_r-K}=t_{p,r}$ or later until $t_d$ (including) \camerad{in either $f_1$ or $f_2$}, $f_1 \stackrel{p}{\approx}_{t_d} f_2$.
\end{proof}

\begin{lemma}\label{lemma:f_2 indistinguishable from f_1 to processes not in K}
$f_2 \stackrel{p}{\approx} f_1$, 
for $f_1$ and $f_2$ defined in the proof of \Cref{claim:optimality:compatible}, and any process $p \notin K$, where $K$ is also defined in that proof.
\end{lemma}
\begin{proof}
For every $p \in K_r-K$: $f_1 \stackrel{p}{\approx}_{t_d} f_2$, by \Cref{lemma:f_1 indistinguishable from f_2 to processes in K_r-K}.
A similar proof would show that this holds also for every $p \in K_{r'}-K$.
Any process $p \notin K_r \cup K_{r'}$ does not receive any messages in $f_1$ and $f_2$ until $t_d$, so $f_1 \stackrel{p}{\approx}_{t_d} f_2$ for it too.
Overall, we have $f_1 \stackrel{p}{\approx}_{t_d} f_2$ for any $p \notin K$. The processes not in $K$ are the ones which remain up after $t_d$, and message transmission after $t_d$ is according to the latencies of $\mathcal{L^*}$ in both $f_1$ and $f_2$, so the executions remain indistinguishable to these processes at all times.
\end{proof}

The following lemmas are used in the proof of \Cref{claim:optimality:KCensus_Alg not slower than Alg}.

\begin{lemma}\label{lemma:optimality:no faulty in reqs keys}
No process in $\mathcal{F}$ appears in $AlgReqs[r]$\t{.keys()} for any process $r$.
\end{lemma}
\begin{proof}
$AlgReqs[r]$ \camerad{is} computed by running Algorithm~\ref{alg:optimality:knowledge} over messages in $e_r$, computing \t{known\_acceptors} and then setting $AlgReqs[r]$ to the value of \t{array\_to\_dict(known\_acceptors[$r$])} after processing all messages known by $r$ when $r$ decides.
\camerad{A simple induction} on the messages known by $r$ when $r$ decides \camerad{shows} that the keys in any \t{array\_to\_dict(known\_accep\-tors[$p$])} for any process $p$ are processes receiving these messages. As no process in $\mathcal{F}$ receives any message in $e_r$, we are done.
\end{proof}

\begin{lemma}\label{no DoAdopt in e_r^*}
Consider a process $r \notin \mathcal{F}$.
No \t{DoAdopt} message is sent in $e_r^*$,
for $e_r^*$ defined in the proof of \Cref{claim:optimality:KCensus_Alg not slower than Alg}.
\end{lemma}
\begin{proof}
\t{DoAdopt} messages are sent in \sysname on \cref{alg:ac:conflict1,alg:ac:conflict2,alg:fd:conflict}.

First, we show that no \t{DoAdopt} message is sent in $e_r^*$ on \cref{alg:ac:conflict1}.
By \Cref{observation:AcceptAndSpread disseminate proposed values}, only \camera{value} $\camera{v_r}$ is sent in \t{Accept\&Spread} messages. Therefore, \t{accepted} on any process may only be set to $\camera{v_r}$ on \cref{alg:ac:assignaccepted}.
Hence, when reaching \cref{alg:ac:ifdiffaccepted}, \t{accepted} $=\camera{v_r}$.
Consequently, the condition on \cref{alg:ac:ifdiffaccepted} is never met, and \cref{alg:ac:conflict1} is never reached in $e_r^*$.

Second, we show that no \t{DoAdopt} message is sent in $e_r^*$ on \cref{alg:fd:conflict}:
By \Cref{lemma:optimality:no faulty in reqs keys}, for any process $p$, all processes in $AlgReqs[p]$\t{.keys()} are correct. 
As the network is stable in $e_r^*$, failure detectors do not suspect any correct process as having crashed. Therefore, the failure detector on any process $p$ may invoke \t{PotentialFailure} only with a process not in $AlgReqs[p]$\t{.keys()} as an argument, and then the condition on \cref{alg:fd:if} will not be met. Hence, \cref{alg:fd:conflict} is never reached in $e_r^*$.

Lastly, we show that no \t{DoAdopt} message is sent in $e_r^*$, including on \cref{alg:ac:conflict2}.
Consider the first sending step of a \t{DoAdopt} or \t{Freeze} message in $e_r^*$.
It cannot be a sending of a \t{DoAdopt} message, as we showed that no \t{DoAdopt} message is sent on \cref{alg:ac:conflict1,alg:fd:conflict}, and a \t{DoAdopt} message is sent on \cref{alg:ac:conflict2} only after a \t{Freeze} message is sent and received.
\camerad{It also cannot be} a sending of a \t{Freeze} message, as a \t{Freeze} message is sent only on \cref{alg:ac:propose:freeze} and only after a \t{DoAdopt} message is sent and received.
Therefore, there is no sending step of a \t{DoAdopt} (or \t{Freeze}) message in $e_r^*$.
\end{proof}

\begin{lemma}\label{lemma:optimality:first message is received in KCensus at the same time as Alg or before}
Consider a process $r \notin \mathcal{F}$. 
Any process $p \notin \mathcal{F}$ receives its first \t{Accept\&Spread} message in $e_r^*$ before or at the same time it receives its first message
in $e_r$ if any,
for $e_r^*$ defined in the proof of \Cref{claim:optimality:KCensus_Alg not slower than Alg}.
\end{lemma}
\begin{proof}
Denote the shortest weighted distance in the graph $\mathcal{L}$ from $r$ to a correct process $p$ by $dist_\mathcal{L}(r,p)$.
We prove by induction on $dist_\mathcal{L}(r,p)$
that every process $p \notin \mathcal{F}$ receives its first \t{Accept\&Spread} message in $e_r^*$ at time $\leq dist_\mathcal{L}(r,p)$. Due to the assumption that a \t{Propose} invocation happened-before any message sending in Alg, $p$ cannot receive a message in $e_r$ at a time earlier than $dist_\mathcal{L}(r,p)$.

For the base case, process $r$ has $dist_\mathcal{L}(r,r) = \mathcal{L}[r,r] = 0$.
Indeed, $r$ receives an \t{Accept\&Spread} message from itself at time 0 in $e_r^*$. 

Consider a correct process $p \neq r$.
Suppose the statement holds for any correct process $q$ such that $dist_\mathcal{L}(r,q) < dist_\mathcal{L}(r,p)$.
Let $P$ be a shortest path from $r$ to $p$ in $\mathcal{L}$, and let $s$ be the process preceding $p$ in this path. By the inductive hypothesis, $s$ receives its first \t{Accept\&Spread} message in $e_r^*$ at time \camerad{$\leq$} $dist_\mathcal{L}(r,s)$.
When handling its first \t{Accept\&Spread} message, process $s$ adds its id to its local \t{known\_acceptors$[s]$} value on \cref{alg:ac:knownacceptorsme1}, which triggers a sending of an \t{Accept\&Spread} message $m$ to $p$ on \cref{alg:ac:broadcast}.
This message reaches $p$ at time \camerad{$\leq$} $dist_\mathcal{L}(r,s) + \mathcal{L}[s,p]$.
As $P$ is a shortest path in $\mathcal{L}$ from $r$ to $p$, the weight of $P$ is $dist_\mathcal{L}(r,p)$. $P$'s sub-path from $r$ to $s$ is a shortest path from $r$ to $s$, hence its weight is $dist_\mathcal{L}(r,s)$. Overall, $dist_\mathcal{L}(r,p) = dist_\mathcal{L}(r,s) + \mathcal{L}[s,p]$. Therefore, message $m$ reaches $p$ at time \camerad{$\leq$} $dist_\mathcal{L}(r,p)$.
\end{proof}

\begin{lemma}\label{lemma:optimality:known_acceptors of Alg are subset of KCensus}
Consider a process $\camerad{r} \notin \mathcal{F}$. Let \t{K$_{Alg}^t$} be the value of \t{known\_acceptors} in Algorithm~\ref{alg:optimality:knowledge}
after processing all messages, known by $r$ at $t_r$, received in $e_r$ up to some time $t$, 
and let \t{K$_{\sysname_{AlgReqs}}^t[p]$} be the value of \t{known\_acceptors} at some process $p$ in Algorithm~\ref{alg:ac:disseminate}
right after time $t$
during $e_r^*$,
for $e_r^*$ defined in the proof of \Cref{claim:optimality:KCensus_Alg not slower than Alg}.

Then for any time $t$, \t{K$_{Alg}^t[p][q]$} $\subseteq$ \t{K$_{\sysname_{AlgReqs}}^t[p][q]$} \camerad{for all processes $p$ and $q$}.
\end{lemma}
\begin{proof}
We prove by induction on message receipt times in $e_r$ of messages, which are known by $r$ at $t_r$, that 
for any such message receipt time $t$,
\t{K$_{Alg}^t[p][q]$} $\subseteq$ \t{K$_{\sysname_{AlgReqs}}^t[p][q]$} \camerad{for all processes $p,q$}.
It is enough to consider times in which messages known by $r$ at $t_r$ are received in $e_r$, because these are the only times for which sets in \t{K$_{Alg}^t$} might be expanded, and sets in \t{K$_{\sysname_{AlgReqs}}^t$} may never shrink.

For the base case, consider time 0. 
Before any message is handled, all sets in \t{known\_acceptors} in $e_r$ are empty except for \t{known\_acceptors$[r][r]$} which equals $\{r\}$. The only messages that may be received at time 0 in $e_r$ are messages sent by $r$ to itself, due to the assumption that a \t{Propose} invocation happened-before any message sending in Alg together with the fact that message transmissions between distinct processes take positive time. Processing these messages from $r$ to itself does not change any value in \t{known\_acceptors}.
As for $e_r^*$, before any message is handled, all sets in \t{known\_acceptors} are empty. 
$r$ sends itself a message at time 0 in $e_r^*$ upon proposing.
Processing it adds $r$ to \t{known\_acceptors$[r][r]$} on \cref{alg:ac:knownacceptorsme1}.

Suppose the statement holds until time $t^p>0$ in which some process $p$ receives from process $s$ a message known by $r$ at $t_r$ in $e_r$. 
We need to show that \t{K$_{Alg}^{t^p}[u][v]$} $\subseteq$ \t{K$_{\sysname_{AlgReqs}}^{t^p}[u][v]$} \camerad{for all processes $u$ and $v$}.
By the inductive hypothesis, 
\t{K$_{Alg}^t[u][v]$} $\subseteq$ \t{K$_{\sysname_{AlgReqs}}^t[u][v]$} for any time $t < t^p$ \camerad{and all processes $u$ and $v$}.
The values of \t{K$_{Alg}^t[u]$} for any process $u \neq p$ do not change at time $t=t^p$ (assuming no two processes receive messages at the same time, otherwise we could apply the following proof to each of them).
So it remains to show that the containment still holds after values are possibly added to cells in \t{K$_{Alg}^t[p]$} upon processing messages received by $p$ at time $t^p$.

If $p$ receives any messages from itself at time $t^p$ in $e_r$, then 
the only value possibly added to cells in \t{K$_{Alg}^t[p]$} upon processing these messages is $p$ added to \t{K$_{Alg}^t[p][p]$} on \cref{alg:optimality:knownacceptorsme}. We next show that $p$ is already in \t{K$_{Alg}^t[p][p]$} before processing these messages, hence no values are added to \t{K$_{Alg}^t[p][p]$} upon processing these messages.
If $p=r$, then $r \in$ \t{K$_{Alg}^t[r][r]$} since initialization.
Else ($p \neq r$), $p$ must receive a message from another process prior to sending itself a message and receiving it, due to the assumption that a \t{Propose} invocation happened-before any message sending in Alg. When processing the message from the other process, $p$ adds its id to \t{K$_{Alg}^t[p][p]$} on \cref{alg:optimality:knownacceptorsme}. 

So it remains to consider a message $m$, known by $r$ at $t_r$ in $e_r$, that $p$ receives from another process $s$ at time $t^p$.
Let $t^s$ be the time $s$ sends $m$. Then $t^s = t^p - \mathcal{L}[s,p]$.
The values possibly added to cells in \t{K$_{Alg}^t[p]$} upon processing $m$ are: 
(a) $p$ to \t{K$_{Alg}^{t^p}[p][p]$} on \cref{alg:optimality:knownacceptorsme}, 
(b) \t{K$_{Alg}^{t^s}[s][s]$} to \t{K$_{Alg}^{t^p}[p][p]$} on \cref{alg:optimality:knownacceptorsme}, and 
(c)
\t{K$_{Alg}^{t^s}[s][q]$} to \t{K$_{Alg}^{t^p}[p][q]$} for any process $q$ on \cref{alg:optimality:knownacceptorsq}.
We next show that
(1) $p$ $\in$ \t{K$_{\sysname_{AlgReqs}}^{t^p}[p][p]$},
(2) \t{K$_{Alg}^{t^s}[s][s]$} $\subseteq$ \t{K$_{\sysname_{AlgReqs}}^{t^p}[p][p]$}, and
(3) \t{K$_{Alg}^{t^s}[s][q]$} $\subseteq$ \t{K$_{\sysname_{AlgReqs}}^{t^p}[p][q]$} for any process $q$.

Let $t_2$ be the first time at which \t{K$_{\sysname_{AlgReqs}}^t[s]$} gained the set of values it had right after time $t=t^s$.
So $t_2$ could either be at initialization before receiving any \t{Accept\&Spread} message, or when $s$ receives some \t{Accept\&Spread} message in $e_r^*$, as \t{K$_{\sysname_{AlgReqs}}^t[s]$} is expanded only during handling \t{Accept\&Spread} messages on \crefrange{alg:ac:upon1}{alg:ac:broadcast}.
We show below that the first alternative is impossible, as 
$s$ receives its first \t{Accept\&Spread} message at some time $t_1 \leq t^s$ and expands \t{K$_{\sysname_{AlgReqs}}^t[s]$} upon handling this message. 
Therefore, the second alternative is the correct one. Namely, $s$ receives an \t{Accept\&Spread} message at time $t_2$ in $e_r^*$. Upon processing it, $s$ sends to $p$ an \t{Accept\&Spread($\camera{v_r}$, \t{K$_{\sysname_{AlgReqs}}^{t^s}[s]$})} message on \cref{alg:ac:broadcast}.
$p$ receives this message at time $t_2 + \mathcal{L}[s,p] \leq t^s + \mathcal{L}[s,p] = t^p$. When handling it, on \cref{alg:ac:knownacceptorsme1} $p$ extends 
\t{K$_{\sysname_{AlgReqs}}^t[p][p]$} to contain the processes in $\{p\}$ $\cup$ \t{K$_{\sysname_{AlgReqs}}^{t^s}[s][s]$} at time $t=t^p$. Hence, condition (1) above is satisfied, and condition (2) is too: \t{K$_{\sysname_{AlgReqs}}^{t^s}[s][s]$} is a superset of  \t{K$_{Alg}^{t^s}[s][s]$} by the inductive hypothesis, which is applicable to $t^s$ as $t^s < t^p$ (since $\mathcal{L}[s,p] > 0$).
Additionally, on \cref{alg:ac:knownacceptorsq} $p$ extends 
\t{K$_{\sysname_{AlgReqs}}^t[p][q]$} to contain \t{K$_{\sysname_{AlgReqs}}^{t^s}[s][q]$} at time $t=t^p$ for each process $q$. Hence, condition (3) above is satisfied: \t{K$_{\sysname_{AlgReqs}}^{t^s}[s][q]$} is a superset of \t{K$_{Alg}^{t^s}[s][q]$} by the inductive hypothesis.

We show the missing part: that $s$ receives its first \t{Accept\&Spread} message in $e_r^*$ at some time $t_1 \leq t^s$, and expands \t{K$_{\sysname_{AlgReqs}}^t[s]$} upon handling this message. 
If $s \neq r$, then $s$ received a message in $e_r$ at time $\leq t^s$, due to the assumption that a \t{Propose} invocation happened-before any message sending in Alg. Thus, by \Cref{lemma:optimality:first message is received in KCensus at the same time as Alg or before}, 
$s$ receives its first \t{Accept\&Spread} message in $e_r^*$ at some time $t_1 \leq t^s$.
The latter is correct also if $s=r$, as $r$ receives its first \t{Accept\&Spread} message in $e_r^*$ at time 0, upon sending such a message to itself on \cref{alg:ac:propose:disseminate}.
The value of \t{K$_{\sysname_{AlgReqs}}^t[s][s]$} is $\emptyset$ before $s$ handles any \t{Accept\&Spread} message in $e_r^*$.
When $s$ receives its first \t{Accept\&Spread} message in $e_r^*$ at time $t_1$, it adds $s$ to \t{K$_{\sysname_{AlgReqs}}^{t_1}[s][s]$} on \cref{alg:ac:knownacceptorsme1}.

\end{proof}

\begin{lemma}\label{lemma:optimality:can_commit returns true at t_r}
Consider a process $r \notin \mathcal{F}$.
Calling \t{can\_commit($\camera{v_r}$)} in $e_r^*$ by $r$ at time $t_r$ would return \t{True},
for $e_r^*$ defined in the proof of \Cref{claim:optimality:KCensus_Alg not slower than Alg}.
\end{lemma}
\begin{proof}
\t{accepted} is set to $\camera{v_r}$ at $r$, as upon proposing, $r$ handles an \t{Accept\&Spread($\camera{v_r}$, $\emptyset$)} message from itself.
In addition, $AlgReqs[r] =$ \t{array\_to\_dict(known\_acceptors$[r]$)} for the value of \t{known\_acceptors} in Algorithm~\ref{alg:optimality:knowledge} after processing all messages received in $e_r$ which are known by $r$ at $t_r$.
Applying \Cref{lemma:optimality:known_acceptors of Alg are subset of KCensus} to time $t_r$ and process $r$ and every process $q \in$ $AlgReqs[r]$\t{.keys()} yields $AlgReqs[r][q]$ $\subseteq$ \t{K$_{\sysname_{AlgReqs}}^{t_r}[r][q]$}, for \t{K$_{\sysname_{AlgReqs}}^{t_r}[r]$} defined in the proof of \Cref{lemma:optimality:known_acceptors of Alg are subset of KCensus}.
QED follows.
\end{proof}

 \section{Fast-Path Correctness} \label{sec:fast-path-proofs}

In this section we prove that the validity and \camera{\cameraa{compatibility}} conditions, appearing in \Cref{sec:constraints:condition}, are an exact characterization of correct fast paths.

\camera{We call \camerad{the \emph{fast path} of a protocol Alg the set of its executions} in a stable network without contention, and we say that a recovery rule \emph{preserves fast decisions} if, in every execution in which some process decides a value $v$ on the fast path, the rule adopts $v$. Alg's fast path is \emph{correct} if such a rule exists.}

\camera{Both proofs reuse the construction from the proof of
\Cref{theorem:optimality}: for each correct process $r$, it takes the
single-proposer execution $e_r$ in a stable network, collects the messages $r$
knows when it decides, and processes them with
\Cref{alg:optimality:knowledge} to obtain $AlgReqs[r]$, whose keys are the
processes whose acceptance $r$ knows of when it decides.
\Cref{claim:optimality:reflexive graph} shows that the resulting
combination has the required normal form, whatever Alg does; when Alg is a
correct consensus protocol, \Cref{claim:optimality:valid} shows that each
requirement is valid and \Cref{claim:optimality:compatible} that every pair is
compatible.}

\begin{theorem}[Sufficiency]\label{theorem:sufficiency}
Any \fixb{fast path} satisfying \fixg{the validity and \camera{\cameraa{compatibility}} conditions could leverage the knowledge derivable from its communication to admit} a recovery rule that preserves fast decisions in every execution.
\end{theorem}
\begin{proof}
Let Alg be any deterministic consensus protocol with a fast path which is taken in executions with a stable network and no contention, such that this fast path satisfies the validity and \camera{\cameraa{compatibility}} conditions.
Derive $AlgReqs$ for Alg \camera{as above; being derived in a stable network without contention, it reflects the requirements of Alg's fast path}.
\camera{Here Alg is not assumed correct, so validity and compatibility come from the conditions rather than from \Cref{claim:optimality:valid,claim:optimality:compatible}:} \camera{the requirement} $AlgReqs[r]$ \camera{of} each process $r$ \camera{is} valid because Alg's fast-path quorums satisfy the validity condition; and the combination of requirements $AlgReqs$ is compatible since Alg's fast path satisfies the \camera{\cameraa{compatibility}} condition.
In addition, by \Cref{claim:optimality:reflexive graph}, $AlgReqs$ \camera{has the normal form required by \sysname}.
Hence, \sysname could use $AlgReqs$. \fixg{For each single-proposer execution, we alter \sysname to disseminate messages using the same communication pattern as Alg's fast path in its single-proposer execution with the same proposer.} The resulting consensus protocol obtains the knowledge derivable from the communication patterns of Alg, and its adoption rule preserves fast decisions in every execution \camera{by \sysname's agreement~(\Cref{theorem:ac:agreement})}.
\end{proof}

\begin{theorem}[Necessity]\label{theorem:necessity}
\camera{The requirements derived from the communication of any correct fast path satisfy the validity and \cameraa{compatibility} conditions, even if this fast path is not defined in terms of acceptance.}
\end{theorem}
\begin{proof}
Let Alg be any deterministic consensus protocol with a correct fast path which is taken in executions with a stable network and no contention.
As its fast path is correct, Alg has a safe recovery rule. We next show that \camera{the requirements derived from Alg's fast-path communication satisfy} the validity and \camera{\cameraa{compatibility}} conditions.

Derive $AlgReqs$ for Alg \camera{as above}.
By \Cref{claim:optimality:valid}, \camera{the requirement} $AlgReqs[r]$ \camera{of} each process $r$ \camera{is} valid. This means that Alg's fast-path quorum contains more than $f$ processes. Namely, the validity condition is satisfied.

By \Cref{claim:optimality:compatible}, the combination of requirements $AlgReqs$ is compatible. This means that
\camera{for every pair of proposers $p$ and $q$ with quorums $Q_p=AlgReqs[p].\t{keys}()$ and $Q_q=AlgReqs[q].\t{keys}()$, at least $f+1$ processes received a chain of messages from some process in $Q_p\cap Q_q$, sent after it learned of the proposal, as known by $p$ or by $q$ when deciding in its single-proposer execution.}
Namely, \camera{the requirements derived from Alg's fast-path communication satisfy} the \camera{\cameraa{compatibility}} condition.
\end{proof}

\fi

\end{document}